\documentclass[12pt]{article}
\usepackage{titling}
\usepackage{graphicx}
\usepackage{amsmath,amssymb,setspace,amsfonts,enumitem,fullpage,amsthm,comment,color,bm,nth}
\usepackage[longnamesfirst]{natbib}
\usepackage {palatino}
\usepackage{epsf}
\usepackage[margin=10pt,labelfont=bf,labelsep=endash]{caption}
\usepackage[colorlinks,pagebackref=false]{hyperref}
\usepackage[usenames,dvipsnames]{xcolor}
\usepackage{titlesec}
\usepackage[bottom]{footmisc} %Pushes footnotes to bottom of page
\usepackage{caption}
\usepackage{subcaption}
\usepackage{tikz}
\allowdisplaybreaks
\usetikzlibrary{arrows,backgrounds,calc,decorations.pathreplacing,decorations.markings}
\colorlet{myblue}{blue!80!green}
\colorlet{mybluelight}{myblue!50}
\tikzset{
  > = latex',
  axis/.style    = {very thick},
  aborder/.style = {draw},
  acomp/.style   = {fill=black, fill opacity=0.1},
  rect/.style    = {very thick},
  form/.style    = {font=\scriptsize},
  sm/.style      = {font=\small},
  vsm/.style     = {font=\scriptsize}
}

\definecolor{dark-red}{rgb}{0.4,0.15,0.15}
\definecolor{dark-blue}{rgb}{0.15,0.15,0.75}
\definecolor{medium-blue}{rgb}{0,0,0.5}
\hypersetup{
    pdftitle={Waiting For Fake News},
    pdfauthor={Boleslavsky},
    citecolor=dark-blue,
    bookmarksnumbered=true,
    urlcolor=BrickRed,
    linkcolor=dark-blue
}

\makeatletter
  \renewcommand\@seccntformat[1]{\csname the#1\endcsname.{\hskip.7em\relax}} %Gets period after section title
\makeatother
\makeatletter
\renewenvironment{proof}[1][\proofname] {\par\pushQED{\qed}\normalfont\topsep6\p@\@plus6\p@\relax\trivlist\item[\hskip\labelsep\bfseries#1\@addpunct{.}]\ignorespaces}{\popQED\endtrivlist\@endpefalse}
\makeatother
\newtheorem{claim}{Claim}

\newtheorem{corollary}{Corollary}
\newtheorem{lemma}{Lemma}

\newtheorem{prop}{Proposition}

\theoremstyle{remark}
\newtheorem{remark}{Remark}

\theoremstyle{definition}

\titlespacing\section{0pt}{10pt plus 2pt minus 2pt}{4pt plus 2pt minus 2pt} %Tightens up spacing after section title
\titlespacing\subsection{0pt}{6pt plus 2pt minus 2pt}{2pt plus 2pt minus 2pt} %Tightens up spacing after subsection titile
\titlespacing\subsubsection{0pt}{6pt plus 2pt minus 2pt}{0pt plus 2pt minus 2pt} %Tightens up spacing after subsubsection title

\renewcommand{\epsilon}{\varepsilon}

\let\oldfootnote\footnote
\renewcommand\footnote[1]{\oldfootnote{\hspace{.5mm}#1}}
\renewcommand{\bar}{\overline}

\newcommand{\appendixref}[1]{\hyperref[#1]{Appendix \ref{#1}}}

\newcommand{\ppp}{\text{principal}}
\newcommand{\PPP}{\text{Principal}}
 \newcommand{\aaa}{\text{agent}}
 \newcommand{\AAA}{\text{Agent}}
  \newcommand{\psafe}{\theta}
  \newcommand{\asafe}{\beta}
  \newcommand{\rcdf}{G}
  \newcommand{\rpdf}{g}
    \newcommand{\rhaz}{H_R}
    \newcommand{\faker}{\text{manipulative}}
    
    \newcommand{\acdf}{F_A}
    \newcommand{\apdf}{f_A}
    \newcommand{\ahaz}{H_A}
    \newcommand{\bl}{\mu_1}

    \newcommand{\pcdf}{F_P}
    \newcommand{\ppdf}{f_P}
     \newcommand{\phaz}{H_P}
    
    \newcommand{\apri}{\sigma}
    
    \newcommand{\mint}{\tau_M}
    \usepackage{tikz}
    \usepackage{caption}

\title{Waiting For Fake News\thanks{The author thanks Deepal Basak, Hector Chade, Ay\c{c}a Kaya, Aaron Kolb, Marilyn Pease, Mehdi Shadmehr, Curtis Taylor, Galina Vereshchagina, Mark Whitmeyer, Dong Wei, Kun Zhang, and seminar audiences at Arizona State, and University of Miami for helpful comments.}}
\author{Raphael Boleslavsky\\
Business Economics and Public Policy\\
Indiana University\\
rabole@iu.edu}
\date{\today}
\begin{document}
\begin{titlingpage}
\maketitle

\abstract{This paper studies a dynamic model of information acquisition, in which information might be secretly manipulated. A \ppp{} must choose between a safe action with known payoff and a risky action with uncertain payoff, favoring the safe action under the prior belief. She may delay her decision to acquire additional news that reveals the risky action's payoff, without knowing exactly when such news will arrive. An uninformed \aaa{} with a misaligned preference may have the capability to generate a false arrival of news, which is indistinguishable from a real one, distorting the information content of news and the \ppp{}'s search. The analysis characterizes the positive and normative distortions in the search for news arising from such manipulation, and it considers three remedies that increase the \ppp{}'s payoff: a commitment to naive search, transfer of authority to the \aaa{}, and delegation to an intermediary who is biased in the \aaa{}'s favor.}

\end{titlingpage}

\section{Introduction}

In the U.S. legal system, the decision to file charges against a suspect is usually made by the local prosecutor's office, based on an investigation conducted by the police department.\footnote{This is true of the vast majority of cases \citet{GPRS2020}. The main exception is federal cases, in which the decision to bring charges is made by a grand jury. } Imagine that an extensive police investigation has uncovered circumstantial evidence that implicates a suspect, but it has produced no direct evidence of the suspect's guilt or innocence. The prosecutor agrees with the police officer's interpretation of the circumstantial evidence, and both the prosecutor and the officer would like to see justice done. Nevertheless, the prosecutor might be less inclined to pursue a case based on circumstantial evidence than the officer would like: in an overloaded system, the prosecutor must consider the opportunity cost, which is not borne by the officer. Furthermore, because a significant investigation has failed to uncover direct evidence, the case may soon be closed. In this situation, a police officer who believes that the suspect is probably guilty based on the circumstantial evidence might be tempted to fabricate direct evidence against him.\footnote{\citet{GPRS2020} document various types of police misconduct that distort evidence in criminal cases, including but not limited to planting physical evidence, coercing a false confession, misrepresenting the results of forensic analysis, and procuring false testimony or identifications from witnesses.} When deciding whether to pursue the case, the prosecutor must therefore consider the possibility that seemingly conclusive evidence may actually be fabricated, impacting the entire investigation.

Such issues are not confined to criminal investigations. A startup firm may not fully internalize a venture capitalist's opportunity cost. Thus, the startup may believe its product to be sufficiently viable to justify additional investment, but the VC may not be willing to invest further given the information currently available. If additional information about product viability (e.g., a working prototype) takes too long to produce, the venture capitalist may withdraw her investment. The startup may therefore be tempted to fabricate such evidence, rather than lose its funding.\footnote{An example is the medical device company Theranos, whose ``pathbreaking'' product was fake.  } When deciding whether to withdraw funding, the venture capitalist must consider the possibility that evidence of a product's viability might be fake. A similar scenario might unfold when a firm decides whether to launch a new project or expand a division, or when a regulator decides whether to approve  a  product coming to market.

In this paper I study a dynamic \ppp{}-\aaa{} model of information acquisition, in which information can be faked strategically. A \ppp{} faces a single, irreversible choice between a safe action with a known payoff and a risky action with an uncertain payoff. The \ppp{} is free to choose her action at any time, and the game continues until she makes this choice. The payoff of the risky action might be higher or lower than the payoff of the safe action, but under the prior belief, the \ppp{} prefers safe. The \aaa{} has no private information about the risky payoff, and he obtains the same payoff as the \ppp{} if the risky action is chosen. But, if the safe action is chosen, then the \aaa{}'s payoff is smaller than the \ppp{}'s. The \ppp{} and \aaa{} agree on the best action in the absence of uncertainty. However, the \aaa{}'s safe payoff is small enough that he prefers the risky action under the prior belief, generating a conflict of interest between the parties that can only be resolved with additional information. The \aaa{} has a privately known type, \textit{\faker{}} or \textit{normal}, which is described in more detail below.

The \ppp{} can delay her choice in order to acquire additional information about the risky payoff. Information arrives through a news process, which generates a single arrival of news at some time.  A \textit{real news} arrival accurately reports the payoff of the risky action, either supporting the risky action or opposing it. The arrival time of real news is uncertain, and it is drawn from a continuous prior distribution that is independent of the risky payoff. If this arrival time is reached, then real news is instantly produced. \textit{Fake news} is generated by the \faker{} \aaa{} in order to influence the \ppp{}'s decision. At any time he likes, the \faker{} \aaa{} can fabricate a news arrival that supports the risky action. Though it is fake, such an arrival \textit{appears real}. Thus, when news arrives, the \ppp{} cannot directly observe whether it is real or fake; she must infer this from the \faker{} \aaa{}'s strategy. The normal \aaa{} does not have the ability to generate fake news, and he waits passively for the game to end. 

The conflict of interest distorts the search for information. In each player's ``first best,'' he or she has full control over the timing and action and only real news arrives. Each player decides how long to wait for real news before stopping the search and selecting the action that he or she prefers under the prior belief. Because the arrival of real news could reveal that each player's favorite action under the prior is actually wrong, news is valuable for both players. Provided its arrival rate is sufficiently high, each player waits for some time before terminating the search and taking an action. Furthermore, when the preference misalignment is not too large, the \aaa{}'s first best search duration is longer than the \ppp{}'s---the paper focuses on this case.\footnote{In the context of the introductory examples, the prosecutor's and VC's opportunity costs are not too high.} In the game, the \ppp{} must consider the possibility that a news arrival that conveys favorable information about the risky action is actually fake. She may therefore treat such news with skepticism, selecting the safe action instead of risky. Furthermore, the possibility that future news is manipulated undermines its informativeness and reduces its value. Consequently, the \ppp{} may terminate her search early, selecting the safe action before news arrives. Compared to the \aaa{}'s first best, these distortions in the \ppp{}'s search strategy make it more difficult for him to implement the risky action, which he prefers under the prior. Thus, if the \ppp{}'s skepticism is currently low and the chance of early termination is high, then the \faker{} \aaa{} may try to implement the risky action by faking favorable news.

When the probability of the \faker{} type is not too large, the game has a unique equilibrium with some intriguing features. The equilibrium incorporates two deadlines, an early ``soft deadline'' and a later ``hard deadline.'' Before the soft deadline, the \aaa{} does not fabricate news. The \ppp{} follows the advice implied by any news arrival and only selects an action if news arrives. In other words, before the soft deadline, equilibrium search is identical to each player's first best. Once the soft deadline is reached, the search for news becomes distorted. At each time between the soft and hard deadlines, the \aaa{} randomly fakes a news arrival, and the \ppp{} randomly terminates the search by choosing the safe action. Despite the positive likelihood of fabrication by the \faker{} type,  news that supports the risky action is sufficiently informative to overturn the \ppp{}'s prior in favor of safe, convincing her to select the risky action. A \faker{} \aaa{} fakes a news arrival with probability one before the hard deadline. Thus, reaching the hard deadline with no arrival reveals that the \aaa{} is normal and all future arrivals are real. Nevertheless, the \ppp{} stops the search at the hard deadline, which coincides with her first best search duration. 

Though equilibrium search lasts until the hard deadline with positive probability, from a normative perspective, it is identical to a first best search that is terminated at the soft deadline, too soon for either player's liking. Before the soft deadline, equilibrium search is undistorted, but once the soft deadline is reached, each player randomizes. Each player is therefore indifferent over all times between the soft and hard deadline. From a normative perspective, the \ppp{} may as well choose safe at the soft deadline, and the \aaa{} may as well fake, inducing risky. Thus, each player's equilibrium payoff is as if he or she waits until the soft deadline and then selects his or her favorite action under the prior if no real news has arrived.

The equilibrium can be interpreted as a ``stochastic unraveling'' of the first best. Imagine that the \ppp{} searches ``naively,'' simply following her strategy from the first best benchmark without considering the possibility of fake news. In this case, she would act on all news that arrives before (or at) the hard deadline, picking the safe action if no news comes. In this case, the \faker{} \aaa{} would not fake before the hard deadline. Before it is reached, the \aaa{} is always able to induce the risky action by faking a bit later, and waiting allows more time for valuable real news to arrive. At the hard deadline, however, waiting is no longer an option: the \aaa{} must fake an arrival immediately, in order to preempt the \ppp{}'s imminent choice of safe, resulting in masses of faking and search termination concentrated at the hard deadline.\footnote{These masses are inconsistent with equilibrium. Anticipating such a faking strategy, the \ppp{} would ignore news that arrives at the deadline, incentivizing the \aaa{} the fake a bit sooner. } In the unique equilibrium of the game, these probability masses are spread backwards in time, resulting in an interval with continuous mixing, starting from the soft deadline and ending at the hard one. In this sense, the naive search stochastically unravels into the unique equilibrium. Furthermore, a higher prior probability of a \faker{} \aaa{} intensifies the \aaa{}'s preemption motive, producing a longer unraveling and an earlier soft deadline. As the probability of the \faker{} \aaa{} approaches a threshold less than 1, the soft deadline approaches 0. If the probability of the \faker{} type exceeds this threshold, then the initial phase of first best search collapses. In this case (and only in this case) a multiplicity of equilibria arises. In all of them, the \ppp{} gains nothing from search: her payoff is as if she implements safe immediately.

Motivated by these normative features, I consider three simple forms of commitment that mitigate the \aaa{}'s preemption motive. First, I allow the \ppp{} to commit to implement the ``naive search'' strategy described above. By guaranteeing that she will select risky if favorable news arrives at the hard deadline, the \ppp{} eliminates the \aaa{}'s preemption motive and the resulting unraveling. The \faker{} \aaa{} simply waits until the hard deadline to fake without worrying that he ``tips his hand'' by being  too predictable. From the \ppp{}'s perspective, such a commitment mimics the first best, except for one scenario: when no real news comes before the hard deadline and the \aaa{} is \faker{}, she selects risky instead of safe. When the probability of a \faker{} \aaa{} is not too high, the \ppp{} benefits from a commitment to naive search. Thus, when the probability of manipulated news is relatively small, the \ppp{} might want to pretend that it does not exist. The \ppp{} might therefore wish to introduce institutional features that allow her to ``plausibly deny'' the existence of fake news.

A more straightforward commitment by the \ppp{} transfers authority over the search and action to the \aaa{}. The \ppp{} need not pretend, she can simply step back and allow the \aaa{} to do as he sees fit. In this case, there is clearly no preemption motive for the \aaa{}. The \ppp{}'s search is nevertheless distorted, because the \aaa{} selects the wrong action if he terminates the search without news, and he does so at the wrong time. The cost of these distortions may be small; if so, the \ppp{} would like to transfer all authority to the \aaa{}. Thus, if delegation is feasible, then the credible threat of news manipulation allows the \aaa{} to command outsized influence within the organization, despite having no private information that is relevant to the choice of action.

Finally, I consider a more nuanced form of delegation in which the \ppp{} allocates the decision authority to an intermediary, whose payoff from the safe action is between hers and the \aaa{}'s. On one hand, the conflict of interest between intermediary and \aaa{} is smaller, which shifts the soft deadline toward a later time and extends the duration of the initial ``first best'' phase. On the other hand, once the soft deadline is reached, the intermediary is indifferent between continuing the search and selecting the safe action with no news. The \ppp{}'s safe payoff is larger than the intermediary's, and therefore the \ppp{} would like to terminate the search at all times after the soft deadline. Thus, delegation is harmful if the search extends for too long and is (weakly) helpful otherwise. When the equilibrium of the main model has beneficial search, such delegation benefits the \ppp{}. For example, the prosecutor benefits when the decision to indict a suspect is made by a grand jury that places a bit less weight on the opportunity cost of trial.

\noindent\textbf{Related Literature.} In my paper, information acquisition is modeled as a stopping problem in the spirit of \citet{W1947} and \citet{ABG1949}. While the classic literature focuses on a single searcher,\footnote{Recent contributions by \citet{CM2019} and \citet{LMS2022} endogenize a single searcher's information by allowing her to choose which source(s) to sample among a set of available signals. In my analysis the information structure is endogenous, but it arises from a strategic interaction.} 
 \citet{BC2007} and \citet{HO2019} study Wald stopping in settings where authority over the final decision and search duration is split between a \ppp{} and \aaa{}. These authors show that even without the ability to manipulate the realization of the news process, controlling the stopping decision can benefit the \aaa{}. In my analysis, the \aaa{} has the ability to stop the search: by faking, he shuts off the news process, which forces the \ppp{} to choose an action immediately. However, the \aaa{} does not benefit from this ability. In particular, if the \aaa{} cannot fabricate a news arrival but can take an action that stops the search, in equilibrium he would never do so.

 A number of related papers consider stopping problems with \textit{observable} manipulation of the news process \citep{AJ2016,KL2020, OSZ}. In this strand, \citet{CKM2022} (CKM) is closest to my analysis. These authors consider a dynamic persuasion model in which an uninformed sender designs a flow experiment at each instant but cannot commit to future experiments. Similar to my model, the receiver's expectation of future informativeness affects her stopping strategy, which itself affects sender's incentive to manipulate the flow experiment. Unlike my model, which often has a unique equilibrium, in CKM a range of mutual expectations can be sustained, resulting in multiple equilibria. More broadly, in my model the \aaa{}'s manipulation of the news manipulation is unobservable. This  difference is crucial: if manipulation were observable in my model, then the \aaa{} would not manipulate in equilibrium.\footnote{My analysis also connects to an emerging literature on dynamic information design,  whereby a sender commits to a long run communication protocol in order to influence a receiver's decision(s) over time \citep{E2017, HKL2016, RSV2017, ES2020}. The pioneering contribution of \citet{KG2011}, who study information design in a static setting, is also related. Juxtaposed with this literature, the \aaa{}'s lack of commitment in my analysis is even more pronounced.  }

I am aware of only one other paper that studies a similar stopping problem with unobservable news manipulation.\footnote{\citet{D2019} studies a dynamic signaling model in which an informed sender's unobserved effort affects a public news process. Unlike my analysis, the sender decides when to stop; see also \cite{K2015,K2019}.}  \citet{EGMSW2022} consider a drift-diffusion model in which an informed sender would like to persuade the receiver to wait by secretly inflating an informative performance measure, showing that observed performance tends to spike just before receiver stops. This model differs from mine in two key respects: information arrives continuously, and the \aaa{} is informed about his type. These modeling differences imply a distinct set of applications and generate different results.\footnote{For example, in the context of venture capital, the setting of \citet{EGMSW2022} suggests that the firm is relatively established. Its current operations generate a continuous performance measure, and its past has endowed management with private information about its potential for success. In contrast, my setting suggests that the firm is just getting started: a breakthrough (in R\&D, perhaps) is still needed for the firm to be viable, and it has not operated for long enough to acquire private information about its ability. }

Finally, my paper is related to \citet{BT2023} (BT), which studies fraud in the context of production. With respect to this paper, there are a number of substantive differences in focus, model, and results. First, in BT the \aaa{}'s incentive to cheat comes from impatience---he wants to obtain a reward quickly. In contrast, in my paper the \aaa{} fakes news in order to influence the \ppp{}'s action.\footnote{In my model even a perfectly patient \aaa{} would manipulate news in equilibrium, provided the players prefer different actions under the prior belief.} Second, in BT the \ppp{} does not have an outside option, and thus, the stopping decision at the heart of my analysis is absent. Third, the equilibrium structure differs significantly. In BT, the \aaa{} fakes a project randomly in an early phase, and the \ppp{} mixes between approving and rejecting such a project when it arrives. In the late phase, all distortions vanish. In my paper, this structure is reversed: the search for news is distorted in the late phase of the equilibrium. Furthermore, the \ppp{}'s response to faking is also different. In BT the \ppp{} randomizes between approval and rejection after a project arrives. In my paper, \ppp{} randomizes between stopping and waiting before news arrives; if it arrives, she always ``approves'' it by selecting the recommended action.

\section{Model}
An organization consists of a \ppp{} and an \aaa{}, who play a game in continuous time. At any time she likes, the \ppp{} makes a single irreversible decision, choosing either a safe or risky action. When the \ppp{} chooses an action, the game ends and the players' payoffs are realized.  The safe action ($S$) delivers known lump-sum payoffs $\psafe\in(0,1)$ and $\asafe\in(0,\psafe)$ to principal and agent respectively, while the risky action ($R$) delivers an uncertain lump-sum payoff $\omega\in\{0,1\}$, identical for both players. Thus, the preference misalignment stems from the safe action, which gives the \aaa{} a smaller payoff than the \ppp{}.\footnote{A similar preference structure, featuring a misaligned payoff with respect to a single action, also appears in the cheap talk model of \citet{CDK2013}.} Although the players would select the same action if the payoff of the risky action were known, uncertainty about it creates disagreement. In particular, the players have a common prior belief $\mu\equiv \Pr(\omega=1)$, where $\asafe<\mu<\psafe$. In this case, the \ppp{} prefers the safe action under the prior, but the \aaa{} prefers the risky action. The \aaa{} and \ppp{} both discount future payoffs at rate $\rho$ and have zero flow payoffs. In addition, the \aaa{} has a privately known type, either \textit{\faker{}} or \textit{normal}, whose significance is described below. At the outset, the \ppp{} believes that the \aaa{} is \faker{} with probability $\sigma\in(0,1)$. 

Though she prefers the safe action given the prior belief, the \ppp{} may delay her decision in order to acquire additional information about $\omega$. Both players observe a news process which generates a single arrival of news---either real or fake---at some time. A \textit{real news} arrival accurately reports the payoff of the risky action. Thus, real news comes in two varieties ``type-0'' and ``type-1,'' and a real news arrival is type-$\omega$ if and only if the risky payoff is $\omega\in\{0,1\}$. The arrival time of real news is uncertain, and both players believe it is distributed according to continuous cumulative distribution function $\rcdf{}(\cdot)$, with associated density $\rpdf{}(\cdot)$, and decreasing hazard rate $\rhaz{}(\cdot)$, supported on $\mathbb{R}_+$.\footnote{The decreasing hazard rate implies that pessimism about the likelihood of a real news arrival grows over time, strengthening the incentive to stop. The analysis is unchanged if there is a positive probability that no news arrives, $\lim_{t\rightarrow\infty}G(t)<1$. One example is an exponential bandit: if real news exists, then its arrival time is exponentially distributed;  otherwise, it never arrives.} If the arrival time of real news is reached, then real news is instantly produced. Note that the arrival time of real news is  statistically independent of the payoff of the risky action (and the \aaa{}'s type). Thus, the content of real news is informative about the risky payoff, but its arrival time is not.

 \textit{Fake news} is produced strategically  in order to influence the \ppp{}'s decision. In particular, the \faker{} type of \aaa{} has the ability to produce an instant arrival of fake type-1 news, which is indistinguishable from real type-1 news that could have arrived at the same moment.\footnote{I focus the case where generating a fake is costless. A small faking cost can be incorporated into the analysis with minimal changes to the results. In particular, only the \ppp{}'s equilibrium strategy adjusts when a small cost is included.} The normal type of \aaa{} does not have the ability to fabricate news and waits passively for the \ppp{}'s decision. In light of this, ``the \aaa{}'' should be understood as the \faker{} type.

Three observations streamline the subsequent exposition and analysis. First, recall that the news process delivers only one arrival. Thus, the \ppp{} chooses between safe and risky immediately when news arrives. Second, recall that type-0 news is never faked.\footnote{If the \aaa{} had the ability to fake or suppress an arrival of type-0 news, he would not choose to do so in equilibrium. Because the \aaa{} is uninformed and a type-0 arrival might be real, type-0 news is always weakly bad news about the risky action and induces safe. Thus, he would never fake a type-0 arrival, since the \aaa{} prefers risky without a real type-0 arrival. If a real type-0 arrival does occur, however, then the \aaa{} prefers safe, and he would want it to be implemented immediately. Unlike \citet{SB2015} and \citet{S2023}, the \aaa{} has no incentive to suppress or censor type-0 news.} Therefore, type-0 news is fully revealing, and the \ppp{} immediately selects the safe action when it arrives. Third, the non-arrival of news before a given time conveys no information about the payoff of the risky action. In particular, the arrival time of real news is drawn from $\rcdf(\cdot)$, and this arrival time is statistically independent of $\omega$. Furthermore, because the \aaa{} has no private information about $\omega$, the arrival time of fake news is also statistically independent of $\omega$, regardless of the \aaa{}'s strategy. The arrival time of news is the smaller of the real and fake arrival times. Because both of these times are statistically independent of the risky payoff, their minimum is also. Thus, players' beliefs about the risky payoff do not change in the absence of news (see also footnote \ref{fnbel}). By implication, at all times before news arrives, the \ppp{} prefers the safe action and the \aaa{} risky.

\noindent\textbf{First-Best Benchmarks.} To clarify the broader implications of the conflict of interest between the players (and in anticipation of the analysis to follow), it is helpful to understand how each player would search for real news in isolation. We therefore consider two ``first-best'' benchmarks, one for the \ppp{} and one for the \aaa{}. In both benchmarks, the relevant player has full authority over the decision and search, and the news is always real. Thus, in either benchmark, the player selects safe following a type-0 arrival and risky following a type-1. Furthermore, as described above, the passage of time does not affect players belief about the risky payoff. Thus, if acting without news, the \ppp{} defaults to safe, and the \aaa{} risky.

Therefore, the crucial decision for each player in the first-best benchmark is the  search duration---how long to wait for news before selecting the default action. For $i\in\{A,P\}$, let $u^{FB}_i(t)$ be player $i$'s expected payoff of selecting search duration $t$ in the first-best benchmark. In addition, let 
\begin{align*}
	\phi_P\equiv \frac{\rho\psafe}{\mu(1-\psafe)}\quad\quad 	\phi_A\equiv \frac{\rho\mu}{\asafe(1-\mu)}.
\end{align*}
The numerator of $\phi_i$ represents player $i$'s marginal cost of delaying his or her default acton by an instant, while the denominator represents the gains from selecting the correct action (instead of the default action) if news arrives. I introduce two assumptions on these parameters that streamline subsequent analysis,
\begin{align*}
\text{A1:}\quad \rhaz(0)>\phi_P>\rhaz(\infty)\quad\quad\text{A2:}\quad\phi_P>\phi_A, \text{ or equivalently, } \beta>\underline{\beta}\equiv\mu(\frac{\mu}{\theta})(\frac{1-\theta}{1-\mu}).
\end{align*}

The significance and interpretation of (A1) and (A2) will be discussed below. For now, note that $\mu<\theta\Rightarrow\underline{\beta}<\mu$. Thus, the set of $\beta$ that is consistent with (A2) is non-empty.

The following result characterizes each player's optimal search duration in his or her first best benchmark.

\begin{lemma}\label{fb}(First-Best Benchmarks). In player $i\in\{A,P\}$ first-best benchmark,
\begin{enumerate}
\item[(i)] $u_i^{FB}(\cdot)$ is a continuous, differentiable, single-peaked function.
	\item [(ii)] if {\normalfont(A1)}, then \ppp{}'s optimal first best search duration $\tau_P$ is strictly positive and finite, and it is the unique value such that $\rhaz(\tau_P)=\phi_P$.

 \item[(iii)] if {\normalfont(A2)}, then \aaa{}'s optimal first best search duration $\tau_A$ is longer than the \ppp{}'s, and it is the unique value such that $\rhaz(\tau_A)=\phi_A$
\end{enumerate}
\end{lemma}

In the first-best benchmark, player $i\in\{A,P\}$ waits for news until time $\tau_i$, at which the hazard rate of real news arrivals reaches $\phi_i$.\footnote{These are essentially the first order conditions for each player. By marginally extending her search, the \ppp{} allows a bit of extra time for news to arrive, but imposes costly delay if it does not.} Despite their disagreement over the default action, both players may wish to wait for news, since the arrival of real news could reveal that their default action was actually the wrong choice. However, because the players disagree over the default action, they may also disagree over the duration of first-best search. A1 ensures that the \ppp{} selects a positive and finite duration in her first best benchmark. A2 ensures that the \aaa{}'s first best search duration is longer than the \ppp{}'s. This assumption softens (but does not eliminate) the disagreement over search duration, streamlining the analysis and exposition.\footnote{Suppose the \aaa{} prefers a longer first best duration and the \ppp{} follows her first best strategy. If no news comes by the \ppp{}'s chosen deadline, the \aaa{} would like the \ppp{} to extend the search, but he has no action that would allow him to do so. Thus, he fakes at or just before the deadline only to preempt the \ppp{}'s imminent choice of safe. In contrast, if the \aaa{} prefers a shorter duration and the \ppp{} follows her first best strategy, then he would fake when his preferred duration is reached, which also induces his preferred action. In this case, the equilibrium strategy is shaped both by the \aaa{}'s incentive to preempt the safe action and to terminate search early. This case is substantially more complicated and is left for future study.} It holds whenever the preference misalignment is not too large.

\section{Equilibrium Analysis}

\subsection{Preliminaries.}

\paragraph{Strategies.}  A pure strategy for the \faker{} \aaa{} is a ``faking time,'' $t\in\{\mathbb{R}^+\cup\infty\}$, at which the \aaa{} will fake type-1 news if the game has not yet ended. A mixed strategy is a probability measure over ``faking times,'' $t\in\{\mathbb{R}^+\cup\infty\}$ represented by cumulative distribution function $\acdf(\cdot)$.

The \ppp{}'s strategy consists of two parts. The first component of the \ppp{}'s strategy describes the search duration.  The \ppp{} selects a stopping time, $t\in\{\mathbb{R}^+\cup \infty\}$. If the \ppp{}'s stopping time is reached with no news, then she stops the search and selects the safe action (which is optimal in the absence of news).  A mixed stopping strategy specifies a probability measure over stopping times, represented by cumulative distribution function $\pcdf{}(\cdot)$. The second component of the \ppp{}'s strategy specifies her decision when news arrives.  Because type-0 news is always real, its arrival reveals that $\omega=0$, and the \ppp{} selects safe. If type-1 news arrives at $t\in\mathbb{R}^+$, then the \ppp{} selects the risky action with probability $a(t)$.\footnote{An ambiguity arises if type-1 news arrives at the exact same time that the \ppp{} intends to stop the search and select safe---does she select risky according the news response profile, or does she select safe? The analysis assumes that the \ppp{} selects safe in this case. In other words, in the small period $[t,t+dt]$ the \ppp{} decides whether to terminate the search before observing whether news arrives. This assumption can be adjusted with minimal impact on the results, though some of the notation becomes more cumbersome.} If $a(t)=1$, then the \ppp{} ``acts on'' or ``follows'' the news at time $t$.

The analysis focuses on the case in which the players' mixed strategies $(\acdf(\cdot),\pcdf(\cdot))$ have an absolutely continuous component and a discrete component. In particular, for $i\in\{A,P\}$,  
\begin{align*}
	F_i(t)=F^*_i(t)+\sum_{m\in M}\mathcal{I}(t\geq t_m)f_m,
\end{align*}
where $M$ is a countable set indexing the set of atoms, $t_m$ is the time of atom $m$, $f_m$ is the probability assigned to atom $m$, $\mathcal{I}(\cdot)$ is the standard indicator function, and $F_i^*(\cdot)$ is absolutely continuous. In the standard way, 
\begin{align*}
	F_i(t)=\int_0^tf_i^*(s)ds+\sum_{m\in M}\mathcal{I}(t\geq t_m)f_m,
\end{align*}
where $f_{i}^*(\cdot)$ is the derivative $F_i^*(\cdot)$ almost everywhere. To economize on notation, it is helpful to represent the atoms using the Dirac $\delta(\cdot)$. In particular, let $f_i(s)\equiv f_i^*(s)+\sum_{m\in M}f_m\delta(s-t_m)$. Using the notation convention for the Dirac $\delta(\cdot)$,\footnote{In particular, $\int_{t_0}^{t_1}\delta(s)ds=\mathcal I\{t_0\leq 0\leq t_1)$.} 
\begin{align*}
	F_i(t)=\int_0^tf_i(s)ds.
\end{align*}
In the text, $f_i(t)$ is sometimes loosely referred to as the ``probability mass'' on $t$.

\paragraph{Payoffs.} From the perspective of each player, the other's strategy distorts the first best. These distortions are described below and are used to derive each player's payoff function.

\noindent\textit{\PPP{}.} From the \ppp{}'s perspective the \aaa{}'s faking strategy distorts the news process. In particular, when the payoff of the risky arm is low, the \aaa{} may fake type-1 news before real (type-0) news arrives. This has implications for the arrival times of type-0 and type-1 news, as well as the information content of type-1 news.

Suppose that the \ppp{} intends to wait for news until time $t$ and then select the safe action. Her search for news can end in three ways---an arrival of type-0 news before $t$, type-1 news before $t$, or no news before $t$.

 Because the \aaa{} never fakes type-0 news, it can only arrive from the real news process. Thus, such an arrival occurs at some time $s<t$ if the payoff of the risky arm is low, real news arrives at $s$, and fake news does not arrive before $s$. In turn, fake news takes longer than $s$ unless the \aaa{} is \faker{} and his chosen faking time is smaller than $s$. Combining these observations, the probability mass of  a type-0 arrival at $s$ is
\begin{align*}
	w_0^P(s)=(1-\mu)\rpdf(s)(1-\apri\acdf(s)).
\end{align*}

In contrast, type-1 news can be faked by the \faker{} \aaa{}, and it can therefore arrive even when the payoff of the risky action is zero. In particular, an arrival of type-1 news at time $s<t$ can occur in three mutually exclusive ways: (1) $\omega=1$,  real news arrives at $s$, fake news takes longer, (2) $\omega=1$, fake news arrives at $s$,  real news takes longer, (3) $\omega=0$, fake news arrives at $s$, real news takes longer. The probability mass of a  type-1 news arrival at $s$ is 
\begin{align*}
 w_1^P(s)&=\underbrace{\mu\rpdf(s)(1-\apri\acdf(s))}_{(1)}+\underbrace{\mu\apri\apdf(s)(1-\rcdf(s))}_{(2)}+\underbrace{(1-\mu)\apri\apdf(s)(1-\rcdf(s))}_{(3)}\\
 &=\mu\rpdf(s)(1-\apri\acdf(s))+\apri\apdf(s)(1-\rcdf(s)).
\end{align*}

No news arrives at or before time $t$ if and only if the arrival times of both  real news and fake news exceed $t$,\footnote{\label{fnbel}Note that the joint probability $\Pr(\text{no news at or before }t|\omega=1)=\Pr(\text{no news at or before }t|\omega=0)=W_\phi^P(t)$. That is, conditional on $\omega$, no news arrives by $t$ if and only if real news takes longer than $t$ (probability $1-\rcdf(t)$) and fake news takes longer than $t$ (probability $1-\apri\acdf(t)$). By implication, the posterior belief that $\omega=1$ given no news arrival by $t$ is the prior, $\mu$. Similar logic applies for the \aaa{}.}
\begin{align*}
W^P_\phi(t)\equiv\Pr(\text{no news at or before }t)=(1-\rcdf(t))(1-\apri\acdf(t)).
\end{align*}

The information content of type-1 news is also affected by the \aaa{}'s faking---it is not necessarily fully revealing. In particular, of the three ways that type-1 news could arrive at time $s$ (described above), only the first two have $\omega=1$. Thus, the \ppp{}'s belief that $\omega=1$, given an arrival of type-1 news at time $s$ is
\begin{align}\label{belief}
	\bl(s)\equiv\Pr(\omega=1|\text{type-1 news at }s)=\frac{\mu\rpdf(s)(1-\apri\acdf(s))+\mu\apri\apdf(s)(1-\rcdf(s))}{\mu\rpdf(s)(1-\apri\acdf(s))+\apri\apdf(s)(1-\rcdf(s))}.
\end{align}
Note that the denominator of $\bl(s)$ is always strictly positive, and therefore this belief can be derived using Bayes' rule at all times $s\geq 0$.

Building on these observations, if the \ppp{} waits for news until time $t$ and acts on news according to $a(\cdot)$, her expected payoff is
\begin{align}\label{uP}
	u_P(t)=\int_0^t &\exp(-\rho s)\{w^P_0(s)\psafe+w^P_1(s)((1-a(s))\theta+a(s)\bl(s))\}ds\\
	+&\exp(-\rho t)W^P_\phi(t))\psafe.\nonumber
\end{align}

\noindent\textit{\AAA.} Compared to his single-player benchmark, the \ppp{}'s stopping strategy and action profile restrict the \aaa{}'s ability to implement the risky action. In his single-player problem, the \aaa{} could stop waiting for real news and select risky at any time. However, in the game with the \ppp{}, he faces two distortions: the \ppp{} may stop the search and select the safe action before news arrives, and even if type-1 news arrives at $t$ the \ppp{} selects the risky action with probability $a(t)$, which may be less than 1.

Suppose that the \aaa{} intends to fake at some time $t$ if it reached. The search for news can end in three ways. First, there could be an arrival of real news at time $s<t$, which occurs before the \ppp{} stops her search. The probability mass of such an arrival is
\begin{align*}
	&w_0^A(s)=(1-\mu)\rpdf(s)(1-\pcdf(s))\\
	&w_1^A(s)=\mu\rpdf(s)(1-\pcdf(s)),
\end{align*}
where the subscript denotes the type of news arrival. Second, the \ppp{} could stop her search at $s<t$, before news arrives. The probability mass of such an arrival is  
\begin{align*}
	&w_S^A(s)=\ppdf(s)(1-\rcdf(s)).
\end{align*}
Finally, the \aaa{}'s intended faking time $t$ could be reached without a news arrival. In other words, real news takes longer than $t$, and the \ppp{}'s stopping time exceeds $t$, which occurs with probability
\begin{align*}
	W_\phi^A(t)=(1-\rcdf(t))(1-\pcdf(t)).
\end{align*}
Thus, if the \aaa{} waits until time $t$ to generate a fake, his expected payoff is
\begin{align}\label{uA}
	u_A(t)=\int_0^t &\exp(-\rho s)\{w^A_0(s)\asafe+w^A_1(s)((1-a(s))\beta+a(s))+w_S^A(s)\asafe\}ds\\+&\exp(-\rho t)W_\phi^A(t)((1-a(t))\asafe+\mu a(t)).\nonumber
%	u_A(t|\pcdf(\cdot),a(\cdot))=\int_0^t &\exp(-\rho s)\{w^A_0(s)\asafe+w^A_1(s)((1-a(s))\beta+a(s))+w_Q^A(s)\asafe\}ds\\+&\exp(-\rho t)W_\phi^A(t)(\mu a(t)+(1-a(t))\asafe).\nonumber
\end{align}
In the previous expression, if the search ends with the arrival of type-1 news, the \aaa{} is confident that $\omega=1$, and his payoff is 1 if the \ppp{} selects the risky action. In contrast, if the search ends at $t$ with no news, then the \aaa{}'s belief that $\omega=1$ is still the prior $\mu$. Thus, if the risky action is chosen at $t$ (in response to the fake type-1 arrival),  the \aaa{}'s expected payoff is $\mu$. 

\paragraph{Equilibrium.} 
A profile of strategies $(a(\cdot),\pcdf(\cdot),\acdf(\cdot))$ is an equilibrium if and only if
 
\begin{enumerate}
	\item[(i)] The \aaa{}'s faking strategy $\acdf(\cdot)$ is optimal given the \ppp{}'s strategy $(a(\cdot),\pcdf(\cdot))$,
	
	\item[(ii)] For all times $t\geq 0$, the \ppp{}'s response to type-1 news, $a(t)$ is optimal given $\bl(t)$, her posterior belief that $\omega=1$ following the arrival of type-1 news at $t$ (derived in (\ref{belief})).
	\item[(iii)] The \ppp{}'s stopping strategy $\pcdf(\cdot)$ is optimal given the \aaa{}'s strategy $\acdf(\cdot)$ and her response to type-1 news, $a(\cdot)$.
\end{enumerate}

Condition (i) holds if and only if 
\begin{align}\label{AIC}
	\apdf(t)>0\Rightarrow u_A(t)\geq u_A(t')\quad\text{for all}\quad t,t'.
\end{align}
That is, the \aaa{}'s mixed strategy assigns positive mass (or probability) only to faking times that generate the highest possible payoff.

Condition (ii) holds if and only if 
\begin{equation}\label{PICa}
a(t) =
    \begin{cases}
      1 & \text{if}\quad  \bl(t)>\psafe \\
      [0,1] & \text{if}\quad  \bl(t)=\psafe\\
      0 & \text{if}\quad  \bl(t)<\psafe,
    \end{cases}
\end{equation}
for all $t\geq 0$. By selecting safe, the \ppp{} obtains guaranteed payoff $\psafe$, but following type-1 news at time $t$ her expected payoff of risky is $\bl(t)$, her posterior belief that $\omega=1$ given type-1 news at $t$, which is derived from Bayes' rule at all times (see (\ref{belief})). The \ppp{} selects the action with the larger expected payoff, mixing only if the payoffs are equal. 

Condition (iii) holds if and only if 
\begin{equation}\label{PICt}
	\ppdf(t)>0\Rightarrow u_P(t)\geq u_P(t')\quad\text{for all}\quad t,t'.
\end{equation}
In other words, optimal stopping for the \ppp{} requires that her strategy assigns positive mass (or probability) only to times that generate the highest possible payoff.\footnote{Each player's mixed strategy can be optimal even if it assigns non-zero probability mass to a set of $t$ at which $u_i(t)$ is suboptimal, as long as the probability of a realization in this set is zero.  However, if the realization of the player's faking time happened to be in this set, the player would be tempted to redraw it.} 

\subsection{Equilibrium With Beneficial Search}

I first focus on equilibria in which the \ppp{} benefits from the ability to search for news despite the \aaa{}'s manipulation. In such an equilibrium, the \ppp{}'s equilibrium payoff exceeds $\psafe$, her payoff from making a decision immediately based only on the prior belief.

To develop intuition for the structure of such an equilibrium, first consider a ``naive search'' by the \ppp{}. In particular, suppose that the \ppp{} adopts her strategy from the first-best benchmark: wait for news until stopping time $\tau_P$ and then select safe, while acting on all news that arrives ($a(\cdot)=1$). As long as the \ppp{}'s stopping time $\tau_P$ has not passed, the \aaa{}'s payoff is identical to his first-best benchmark---the distortions associated with the \ppp{}'s stopping are absent. Therefore, the \aaa{}'s payoff is increasing in his faking time if $t<\tau_P$ (his first-best payoff has a single peak at $\tau_A>\tau_P$). Thus, the \aaa{} would not want to fake at such times, which reinforces the \ppp{}'s incentive to wait for news (\ppp{}'s first-best payoff has a single peak at $\tau_P$). However, if time $\tau_P$ passes without a real news arrival, the \ppp{} stops on the safe action, which the \aaa{} dislikes under the prior belief. Thus, the \aaa{} would prefer to preempt the \ppp{}'s stopping by generating a fake arrival of type-1 news just before $\tau_P$, thereby inducing the \ppp{} to select  risky. Anticipating that it is fake, \ppp{} would ignore the type-1 arrival at time $\tau_P$, and the \aaa{} would like to preempt earlier, ``unraveling'' the search for news backwards from time $\tau_P$.

To explore this unraveling in more detail, it is helpful to consider incentives at two small adjacent time periods, an early period and a late period, both of which precede $\tau_P$. For ease of discussion, suppose that within-period, the news arrival occurs before the \ppp{}'s stopping decision.\footnote{In other words, if the \ppp{} plans to stop her search in a period, and type-1 news arrives in the period, then she selects risky according to her strategy $a(\cdot)$.} As described above, unraveling is triggered by the \aaa{}'s expectation that the \ppp{} will stop at the end of the late period in the absence of news, which incentivizes the \aaa{} to fake at the beginning of the late period in order to preempt the \ppp{}'s imminent choice of safe. Anticipating such a faking strategy, the \ppp{} would be skeptical about such a convenient arrival of type-1 news. Indeed, with this faking strategy the probability of a real arrival in the late period is low (the period is small), but the probability of a fake is high. Thus, the \ppp{} identifies the arrival of type-1 news in the late period as a likely fake and ignores it, selecting safe in the late period whether or not type-1 news arrives. But if she always selects safe in the late period, why wait? It is better for the \ppp{} to avoid the cost of discounting by selecting safe at the end of the early period. Of course, if the \aaa{} anticipates such early stopping by the \ppp{}, he would want to preempt it by faking in the early period. However, if the \faker{} \aaa{} always fakes early, then the \ppp{} does \textit{not} want to stop her search  at the end of the early period. Indeed, if no news comes in the early period, the \ppp{} concludes that the \aaa{} is \textit{not} the \faker{} type and that any future news is real. She therefore wants to wait for news until $\tau_P$. Thus, the \aaa{} wants to fake in the same period that the \ppp{} stops, but the \ppp{} wants to stop in the other period. This cyclical structure of best responses suggests that the search unravels stochastically: the \ppp{} mixes between stopping in the early period and the late period, and the \aaa{} mixes between faking early and late.

Extended to continuous time,  the preceding intuition suggests that in equilibrium, the atom of stopping and faking that would be concentrated at $\tau_P$ under a ``naive search''  unravels into an interval of times with smooth, stochastic stopping and faking. But, if search is beneficial for the \ppp{}, the stochastic unraveling cannot extend all the way to time 0. If it does, then the \ppp{}'s payoff is the same as if she immediately terminates the search for news by selecting safe, i.e., $\psafe$. Furthermore, the unraveling cannot leave an atom of fakes at $\tau_P$; the existence of such an atom triggered the unraveling in the first place. Provided there is no atom of fakes at $\tau_P$, an atom of stopping by the \ppp{} could occur at this time, since the \aaa{}'s faking strategy preempts the entire atom for certain.

The following proposition formalizes this intuition by characterizing the structure of an equilibrium with beneficial search.

\begin{prop}\label{strucval}(Structure of Beneficial Search Equilibrium.) If an equilibrium with beneficial search exists, then it has the following structure. There exists $\mint>0$ such that
\begin{enumerate}
	\item[(i)] the \aaa{}'s faking time is drawn from a mixed strategy with no mass points or gaps, supported on interval $[\mint,\tau_P]$.
	\item[(ii)] the \ppp{}'s stopping time is drawn from a mixed strategy with no gaps, supported on interval $[\mint,\tau_P]$; $\tau_P$ is the only possible mass point in the \ppp{}'s stopping strategy.
	\item[(iii)] the \ppp{} selects risky following type-1 news at all times, $a(t)=1$ for all $t\geq 0$.
\end{enumerate}
\end{prop}

Consistent with the ``stochastic unraveling'' intuition described above,  an equilibrium with beneficial search time features two deadlines: a soft deadline $\mint$, and a hard deadline $\tau_P$. If the hard deadline is reached with no news, the \ppp{} stops the search and selects safe.  Before the soft deadline, the players incentives to wait for news are mutually reinforcing: during this phase there is no stopping by the \ppp{} or faking by the \aaa{}, and equilibrium search is identical to each player's first best. However, past $\mint$, the search becomes distorted. At the soft deadline, the \ppp{} becomes ``barely engaged'' in search---she acts on a news arrival if one occurs, but she may terminate the search at any moment if one does not. The \faker{} \aaa{} becomes ``anxious'' about the arrival rate of real news (which has decreased) and about the \ppp{} abandoning the search by selecting the safe action before news arrives. Consequently, he begins to fake randomly, which undermines the value of information and reinforces the \ppp{}'s minimal engagement in the search.  Although type-1 news is randomly fabricated during this phase, it remains sufficiently informative that the \ppp{} prefers to select the risky action when such news arrives. Maintaining the \ppp{}'s ``bare engagement'' with the search requires that a future type-1 news arrival is expected to be sufficiently informative to offset her cost of delay. By implication, such an arrival must also be informative enough to induce the risky action.

The equilibrium structure also reveals an interesting normative feature. The support of each player's mixed strategy begins at $\mint$, the soft deadline. Below $\mint$, there is no stopping or faking, and each player's payoff coincides with the first best. Furthermore, optimality of player strategies (\ref{AIC}, \ref{PICt}) imply that each player's equilibrium payoff is equal to the payoff of stopping the search or faking at $\mint$. Because there are no atoms at $\mint$, this payoff is simply the first-best payoff evaluated at $\mint$. Thus,  each player's payoff in the beneficial search equilibrium equal to the payoff he or she would receive in the first-best benchmark, but with a premature stopping time, $\mint<\tau_P<\tau_A$. In other words, if real news arrives in equilibrium before $\mint$, then both players get their first-best payoffs. But if the equilibrium search extends past $\mint$, players' payoffs are as if they simultaneously implemented their default actions at $\mint$, even though these default actions are different.

The following proposition completes the characterization of equilibria with beneficial search, deriving necessary and sufficient conditions for existence, establishing uniqueness, and providing closed-form expressions for the players' strategies.

\begin{prop}\label{charval}(Equilibrium With Beneficial Search). An equilibrium with beneficial search exists if and only if $\sigma<\bar{\sigma}$, where 
\begin{align*}
	\bar{\sigma}\equiv 1-\exp\{-\int_0^{\tau_P}\frac{\mu(1-\psafe)\rhaz(s)-\rho\psafe}{\psafe-\mu}ds\}.
\end{align*}
In such an equilibrium, then the players' faking and stopping strategies are
\begin{align*}
		\acdf(t)&=\frac{1}{\sigma}\Big(1-\exp\{-\int_{\mint}^{t}\frac{\mu(1-\psafe)\rhaz(s)-\rho\psafe}{\psafe-\mu}ds\}\Big),\\
				\pcdf(t)&=1-\exp\{-\int_{\mint}^{t}\frac{\asafe(1-\mu)\rhaz(s)-\rho\mu}{\mu-\asafe}ds\}\quad\text{for}\quad t\in[\mint,\tau_P)\quad\text{and}\quad\pcdf(\tau_P)=1,
\end{align*}
 both supported on $[\mint,\tau_P]$, where $\mint$ is the unique value for which 
\begin{align*}
1-\exp\{-\int_{\mint}^{\tau_P}\frac{\mu(1-\psafe)\rhaz(s)-\rho\psafe}{\psafe-\mu}ds\}=\sigma.
\end{align*}
The equilibrium payoff of player $i\in\{A,P\}$ is $u_i^{FB}(\mint)$. No other equilibrium with beneficial search exists.
\end{prop}

\begin{figure}
\begin{minipage}{0.5\textwidth}
%LEFT PANEL
\begin{center}
\begin{tikzpicture}[scale=0.8]

%AXES
\draw[-stealth] (0,0) -- (0,8); 
\draw[-stealth] (0,0) -- (8,0); 

%AXES LABELS
\fill (8,0) node[right] {\footnotesize{$t$}};

%support
\fill (7,0) node[below] {\footnotesize{$\tau_P$}};
\filldraw[color=black, fill=black] (7,0) circle (2.5pt);
\fill (3,0) node[below] {\footnotesize{$\mint$}};
\filldraw[color=black, fill=black] (3,0) circle (2.5pt);
\draw[line width=0.5pt, black, dashed] (7,8) -- (7,0);

%ACDF
\draw[line width=1pt, red] (0,0) -- (3,0); 
\draw[line width=1pt, red] (7,8) -- (8,8); 
\draw[line width=0.5pt, red] plot[smooth] coordinates {(3,0) 
(4,5) (5,7) (6,7.8) (7,8)};
\fill (7,7.2) node[left] {\footnotesize{$\acdf(\cdot)$}};

%PCDF
\draw[line width=0.5pt, blue] (0,0) -- (3,0); 
\draw[line width=0.5pt, blue] (7,8) -- (8,8); 
\filldraw[color=blue, fill=blue] (7,8) circle (2.5pt);
\filldraw[color=blue, fill=white] (7,5.1) circle (2.5pt);
\draw[line width=0.5pt, blue] plot[smooth] coordinates {(3,0) (4,3) (5,4.25) (6,4.8) (7,5.1)};
\fill (7,4.2) node[left] {\footnotesize{$\pcdf(\cdot)$}};

%\draw[blue,line width=1pt] plot[variable=\x,domain=0:8,smooth]
%({\x},{(1-\gam+\gam*exp(-\lam*\x))^((1-\th)*\m/(\th-\minutesdate)});
%

%COMMANDS
%\fill (0,0) node[left] {\footnotesize{$\mu_a^{PS}$}};
%\draw[line width=0.5pt] (0,-1.5) -- (0,7); 
%\filldraw[color=black, fill=black] (1.5,0) circle (3pt);
%\draw[line width=0.5pt,red,dashed] plot[smooth] (3.5,0.7956)--(3.85,0.9636)--(4.2,1.1556)--
%(4.55,1.3716)--(4.9,1.6116)--(5.25,1.8756)--(5.6,2.1636)--(5.95,2.4756)--
%(6.3,2.8116)--(6.65,3.1716)--(7,3.5556);
\end{tikzpicture} \captionsetup{font=footnotesize}
\captionof{figure}{Equilibrium Strategies $\acdf(\cdot),\pcdf(\cdot)$\label{strat}}
\end{center}
\end{minipage}
\begin{minipage}{0.5\textwidth}
\begin{center}
\begin{tikzpicture}[scale=0.8]

%AXES
\draw[-stealth] (0,0) -- (0,8); 
\draw[-stealth] (0,0) -- (8,0); 

%AXES LABELS
\fill (8,0) node[right] {\footnotesize{$t$}};

\fill (7,0) node[below] {\footnotesize{$\tau_P$}};
\filldraw[color=black, fill=black] (7,0) circle (2pt);
\fill (3,0) node[below] {\footnotesize{$\mint$}};
\filldraw[color=black, fill=black] (3,0) circle (2pt);

%COMMANDS
\fill (0,1) node[left] {\footnotesize{$\mu$}};
\fill (0,4) node[left] {\footnotesize{$\psafe$}};
\draw[line width=0.5pt,dashed] (0,3) -- (3,3); 
\draw[line width=0.5pt,dashed] (0,6.3) -- (3,6.3); 
\fill (0,6.3) node[left] {\footnotesize{$u_P^{FB}(\mint)$}};

%\draw[line width=0.5pt] (0,-1.5) -- (0,7); 
%\filldraw[color=black, fill=black] (1.5,0) circle (3pt);
\draw[line width=0.5pt,red] (0,1) .. controls (0,1) and (1,2.5) .. (3,3);
\draw[line width=0.5pt,red] (3,3) -- (7,3);
\fill (7,3.5) node[right] {\footnotesize{$u_A(\cdot)$}};
\fill (0,3) node[left] {\footnotesize{$u_A^{FB}(\mint)$}};

%\draw[line width=0.5pt,red, dashed] (3,3) .. controls  (4,3.3) and (5,3.7)  .. (8,4);

\draw[line width=0.5pt,blue] (0,4) .. controls (0,4) and (1,5.5) .. (3,6.3);

\fill (7,6.7) node[right] {\footnotesize{$u_P(\cdot)$}};

%\draw[line width=0.5pt,blue,dashed] (3,6.3) .. controls (4,6.75) and (7,7) .. (8,6.75);
\draw[line width=0.5pt,blue] (3,6.3) -- (7,6.3);
\draw[line width=0.5pt,blue] (7,6.3) .. controls (7.25,6.29) and (7.75,6.15) .. (8,5.95);

\draw[line width=0.5pt,black,dashed] (3,0) -- (3,6.3);
\draw[line width=0.5pt,black,dashed] (7,0) -- (7,6.3);
\filldraw[color=red, fill=white] (7,3) circle (2.5pt);
\filldraw[color=red] (7,2.1) circle (2.5pt);
\draw[line width=0.5pt,red] (7,2.1) -- (8,2.1);
%plot[smooth, tension=0.1] (0,1)--(1,2.5)--(2,3)--(3,3.4);
% (3.5,0.7956)--(3.85,0.9636)--(4.2,1.1556)--
%(4.55,1.3716)--(4.9,1.6116)--(5.25,1.8756)--(5.6,2.1636)--(5.95,2.4756)--
%(6.3,2.8116)--(6.65,3.1716)--(7,3.5556);
\end{tikzpicture}
\captionsetup{font=footnotesize}
\captionof{figure}{Equilibrium payoffs $u_A(\cdot)$, $u_P(\cdot)$\label{pay}}
\end{center}
\end{minipage}
\end{figure}

%\vspace{5mm}

The  equilibrium is illustrated graphically in Figures \ref{strat} and \ref{pay}, which depict the equilibrium strategies and payoff functions of the players. As illustrated in Figure \ref{strat}, the \aaa{} smoothly mixes on the equilibrium support $[\mint,\tau_P]$, but the \ppp{}'s strategy has an atom at $\tau_P$. Because of this atom, the \aaa{}'s payoff of waiting to fake until $\tau_P$ has a downward jump discontinuity. The jump in the \aaa{}'s payoff does not disrupt the equilibrium, because the \faker{} \aaa{} fakes the news with probability 1 before time $\tau_P$ under his strategy. In other words, time $\tau_P$ is only ever reached if the \aaa{} is \textit{not} the \faker{} type; therefore, the \ppp{}'s atom of stopping at this time does not disrupt the incentives of the \faker{} agent. This atom has another interesting implication. Together with $a(\cdot)=1$ (see Proposition \ref{strucval}), the atom's existence implies that the \ppp{} engages in a ``naive search'' with positive probability in equilibrium. To an outside observer, it may appear that the \ppp{} is completely unaware that the \aaa{} may have the ability to produce fake news, even though her equilibrium strategy accounts fully for this possibility.

Next, consider the informativeness of news in a beneficial search equilibrium. Although the underlying real news process is fully revealing, the equilibrium news observed by the \ppp{} is distorted by the \aaa{}'s faking. Consistent with the intuition underlying stochastic unraveling and Proposition \ref{strucval}, type-1 news is always sufficiently informative to induce the \ppp{} to select the risky action when it arrives. However, the informativeness of type-1 news is non-monotone in the arrival time. At early times, before the soft deadline, type-1 arrivals are fully revealing. At the soft deadline, $\bl(t)$ drops discontinuously, and it increases gradually over the equilibrium support. At the hard deadline, type-1 news is fully revealing once again. The \ppp{} nevertheless prefers to exit at the hard deadline---even though a future news arrival reveals the risky payoff, it comes too slowly to justify additional delay. Thus, the informativeness of type-1 news exhibits a pattern of collapse and recovery, as formalized in the following corollary.

\begin{corollary}\label{corbel}(Beneficial Search, Beliefs). In an equilibrium with beneficial search, 
\begin{enumerate}
\item [(i)] the posterior belief that $\omega=1$ given type-1 news at $t$ is
\begin{align*}
	\mu_1(t)=\begin{cases}
		1 & \text{if }t\in[0,\mint)\\
		\psafe+(\psafe-\mu)\frac{\rho\psafe}{\rhaz(t)\mu(1-\mu)-\rho\theta} &\text{if }t\in[\mint,\tau_P]\\
		1 & \text{if }t\in[\tau_P,\infty).
	\end{cases}
\end{align*}
\item[(ii)] belief $\mu_1(t)$ has a downward jump at $\mint$, is strictly increasing on $[\mint,\tau_P)$, and is continuous at $\tau_P$. In addition, an $\epsilon>0$ exists such that $\mu_1(t)\geq \psafe+\epsilon$ at all $t\geq 0$.
\item[(iii)] type-0 news reveals that $\omega=0$ at all times.	
\end{enumerate}
\end{corollary}
 The informativeness of the equilibrium news process is undermined most significantly near the soft deadline, at the earliest times in the support of the \aaa{}'s distribution of faking times. This might be surprising, given that the \faker{} \aaa{} ``panics'' as the hard deadline approaches and the hazard rate of faking approaches infinity at the hard deadline. This effect arises from a ``feedback loop'' in the \ppp{}'s beliefs about the \aaa{}'s type. Reaching a time near the hard deadline without a type-1 arrival is, in and of itself, a strong signal that the \aaa{} is not manipulative, suggesting that a type-1 arrival is likely real.  The \faker{} \aaa{} exploits the \ppp{}'s optimism by faking aggressively. If type-1 news does not arrive, it is an  even stronger signal that the \aaa{} is not \faker{}, further reinforcing the \faker{} type's incentive to fake an instant later, and so on. The informativeness of type-1 news increases over time because the \aaa{} is less likely to be \faker{} as time passes, which counteracts the increased aggressiveness of the \faker{} \aaa{}'s faking.

The comparative statics with respect to $\{\sigma,\psafe\}$ are helpful in subsequent analysis. Intuitively, an increase in $\psafe$ intensifies the preference misalignment between players, as well as the disagreement over default action and search duration.\footnote{When the safe payoff is higher, the \ppp{} prefers safe more strongly under the prior, and in addition, prefers an even shorter duration under the first best, shifting it further away from the \aaa{}'s.} Thus, we might expect that the equilibrium distortions might be stronger, shifting the soft and hard deadlines, as well as the players' strategies toward earlier times. Meanwhile, $\sigma$ determines the salience of the misalignment: if it is low, it is unlikely that the \aaa{} can affect the \ppp{}'s action, even if he would gain more by doing so. Effectively, increases in $\sigma$ intensify the conflict of interest and reductions weaken it. Thus, we might expect increases in $\sigma$ to shift the soft deadline and the player's strategies toward earlier times, resulting in a stronger distortion from the first best.

\begin{corollary}\label{corcs}(Beneficial Search, Comparative Statics). In an equilibrium with beneficial search, changes in model parameters have the following effects:
\begin{enumerate}
	\item[(i)] An increase in $\sigma$ reduces $\tau_M$, has no effect on $\tau_P$, and generates a first order stochastic dominance shift in the stopping and faking strategies toward earlier times. 
	\item[(ii)] An increase in $\psafe$ reduces $\tau_M$ and $\tau_P$, and generates a first order stochastic dominance shift in the stopping and faking strategies toward earlier times.
\end{enumerate}
\end{corollary}

The normative properties of the equilibrium are also worth discussing further. As the probability of the \aaa{} being the faker increases, the stochastic unraveling intensifies and the soft deadline $\mint$ shifts toward zero, thereby reducing the value of search and the equilibrium payoffs of the \ppp{} and \faker{} \aaa{}.  Furthermore, if the probability of the \faker{} type is too large, then the unraveling completely undermines the value of search for the \ppp{}: no equilibrium with beneficial search exists. Perhaps the most striking normative feature of Proposition \ref{charval} is that the stochastic unraveling can completely undermine the value of search, even if the degree of preference misalignment between \ppp{} and \aaa{} is arbitrarily small. Indeed, provided that the players disagree over the default action ($\mu\in(\asafe,\psafe)$), (A2) holds, and the \aaa{} is sufficiently likely to be \faker{}, search has no value for the \ppp{} in equilibrium regardless of the degree of preference misalignment.

\subsection{Equilibrium Without Beneficial Search}

The preceding analysis implies that equilibrium with beneficial search exists only for a subset of parameters, in particular, when the prior probability of the \faker{} \aaa{} is not too high. In the remaining situations, if an equilibrium exists, it must be such that search is not beneficial. Does such an equilibrium indeed exist when $\sigma>\bar{\sigma}$? Furthermore, the logic of strategic unraveling described above suggests a strategic complementarity between the \ppp{}'s stopping and the \aaa{}'s faking. Could it be that, because of this strategic complementarity, multiple equilibria exist for $\sigma<\bar{\sigma}$, one with beneficial search and another (or others) without? 

This section answers the preceding questions, deriving necessary and sufficient conditions for the existence of an equilibrium in which the \ppp{} does not benefit from search. In particular, Proposition \ref{charnot} establishes that an equilibrium without beneficial search exists if and only if $\sigma>\bar{\sigma}$, and provides an explicit characterization of one such equilibrium.

If the \ppp{} does not benefit from search in equilibrium, then her equilibrium payoff cannot be higher than $\psafe$, her payoff of choosing an action based on the prior alone. Simultaneously, $u_P(0)=\psafe$, regardless of the \aaa{}'s strategy. Therefore, in such an equilibrium the \ppp{}'s equilibrium payoff must be $\psafe$. From (\ref{PICt}), we have $u_P(t)\leq \psafe$ for all $t\geq 0$, with equality for all $t$ at which $\ppdf(t)>0$. Therefore, in an equilibrium without beneficial search, the \ppp{} is indifferent between stopping at any such time and stopping at time $0$. 

Building on the preceding observation, suppose that  an equilibrium without beneficial search exists. Without changing the \aaa{}'s strategy $\acdf(\cdot)$, suppose the \ppp{} stops her search at time 0 with probability 1. As described above, such a stopping decision is optimal given the \aaa{} strategy. Furthermore, if the \ppp{} stops her search at time 0, then the \aaa{}'s payoff is $\asafe$, regardless of the faking time he chooses. Thus, the \aaa{} is indifferent over all possible faking times. In particular, the \aaa{}'s original faking strategy $\acdf(\cdot)$ is optimal. By implication, immediate stopping for the \ppp{}, coupled with the \aaa{}'s original strategy is also an equilibrium.

\begin{lemma}\label{IS}(Immediate Stopping). If an equilibrium exists in which the \ppp{} does not benefit from search, then an equilibrium exists in which the \ppp{} does not benefit from search and selects the safe action at time 0 with probability 1.
\end{lemma}

Lemma \ref{IS} simplifies the characterization of equilibrium without beneficial search. It allows us to focus on equilibria in which the \ppp{} immediately stops, so that the \aaa{} is indifferent over all possible faking strategies. Thus, an equilibrium without beneficial search exists if and only if a faking strategy exists such that immediate stopping by the \ppp{} is a best response. In other words, even if she selects her action profile optimally (as in (\ref{PICa})), $u_P(t)\leq\psafe$ for all $t\geq 0$. Necessary conditions for existence of such an equilibrium, along with one such equilibrium are presented below.

\begin{prop}\label{charnot}(Equilibrium Without Beneficial Search). An equilibrium without beneficial search exists if and only if $\sigma\geq \bar{\sigma}$. If this condition holds, then many such equilibria exist. In one such equilibrium, the \ppp{} terminates her search immediately, and the \aaa{} uses continuous faking strategy 
	\begin{align*}
		\acdf(t)&=\frac{1}{\sigma}\Big(1-\exp\{-\int_{0}^{t}\frac{\mu(1-\psafe)\rhaz(s)-\rho\psafe}{\psafe-\mu}ds\}\Big),
\end{align*}
supported on $[0,\tau_X]$, where $\tau_X>\tau_P$.
\end{prop}
Proposition \ref{charnot} illustrates that equilibria without beneficial search exist if and only if the equilibrium without beneficial search, characterized in Proposition \ref{charval} does not. Thus, for all parameters an equilibrium exists. Furthermore, the stochastic unraveling completely undermines the value of the \ppp{}'s search for news if and only if $\sigma>\bar{\sigma}$. When $\sigma$ is smaller, the  equilibrium with beneficial search (from Proposition \ref{charval}) is also the unique equilibrium of the game. This also allows for a clean normative implication, untainted by issues of equilibrium selection: for $\sigma=0$ the \ppp{}'s equilibrium payoff coincides with the first-best, is decreases in $\sigma$ until $\bar{\sigma}$, at which point the value of search is eliminated, and further increases in $\sigma$ have no effect on the \ppp{}'s payoff.

\section{Normative Remedies}\label{rem} 
The fundamental incentive problem in this game comes from the \aaa{}'s preemption motive, which triggers the stochastic unraveling of the search for news. Here, I describe three simple remedies that mitigate or eliminate the \aaa{}'s preemption motive, deriving conditions under which they improve the \ppp{}'s payoff.

\paragraph{Commitment to Naive Search.} Consider a commitment by the  \ppp{} to carry out a naive search, in which she waits until $\tau_P$ to stop and acts on all news that arrives at or before this time. Effectively, the \ppp{} blinds herself to the possibility that the \aaa{} might be the \faker{} type and behaves as if all news is real. Such a commitment eliminates the \aaa{}'s preemption motive. In particular, because the \ppp{} always waits at $t<\tau_P$, and the \aaa{} can induce the risky action by faking type-1 news, the search is identical to the \aaa{}'s first best at $t<\tau_P$. Thus, the \aaa{} has no incentive to fake at such times since he prefers a longer search duration in a first best search, i.e., $t<\tau_P<\tau_A$. Given that the \ppp{} acts on it, the \aaa{} would therefore wait until $\tau_P$ to fake a type-1 arrival. Without a binding commitment to act naively, the \ppp{} would recognize the fake and select the safe action if type-1 news arrives at $\tau_P$, triggering the \aaa{} to preempt. In contrast, with this binding commitment, the \aaa{} has no incentive to preempt: he can fake at $\tau_P$ and still induce the risky action. 

Though it eliminates the \aaa{}'s preemption motive and the stochastic unraveling, such a commitment does not fully restore the \ppp{}'s first-best search. In particular, if $\tau_P$ is reached without a real news arrival, in the first-best the \ppp{} would like to stop and select the safe action. However, if she commits to a naive search, she may be forced to select risky. Indeed a \faker{} \aaa{} will fake type-1 news at $\tau_P$. Following through on her commitment, the \ppp{} must select risky in this instance. In particular, with a commitment to naive search, the \ppp{}'s expected payoff is
\begin{align*}
	u_P^{N}&=\int_0^{\tau_P}\exp(-\rho s)\rpdf(s)(\mu+(1-\mu)\psafe)ds+\exp(-\rho\tau_P)(1-\rcdf(\tau_P))((1-\sigma)\psafe+\sigma\mu)\\
	&=u_P^{FB}(\tau_P)-\exp(-\rho \tau_P)(1-\rcdf(\tau_P))\sigma(\psafe-\mu).
\end{align*}

The next proposition shows that such a commitment is always valuable for the \ppp{} if the equilibrium (without commitment) has beneficial search. More broadly, such a commitment is valuable unless the \aaa{} is very likely to be \faker{}, so that the probability of selecting risky at the deadline is high.

\begin{prop}\label{naive} (Naive Search). If the prior probability of the \faker{} \aaa{} is not too high, then the \ppp{}'s payoff with a commitment to naive search is higher than her equilibrium payoff in the main model. If $\sigma<\bar{\sigma}$, then commitment to naive search delivers a higher payoff to \ppp{} than the unique equilibrium without commitment, which has beneficial search. If $\sigma\in(\bar{\sigma},\widetilde{\sigma}_{N})$ for some $\bar{\sigma}<\widetilde{\sigma}_{N}\leq 1$, then commitment to naive search delivers a higher payoff to the \ppp{} than any equilibrium without commitment, all of which feature non-beneficial search. 
\end{prop}

When the probability of manipulation is relatively low, the \ppp{} would like to pretend that manipulation is impossible. To support this pretense, the \ppp{} could develop an institution that allows her to plausibly deny the possibility of fabricated news. As awareness of such manipulation grows, it becomes more difficult for the \ppp{} to maintain this facade, threatening her commitment to naive search. To the extent that the \ppp{} is aligned with society more broadly, increased awareness of fabricated or manipulated information can therefore be counterproductive, unless it is also accompanied by increased accountability. 

\paragraph{Delegating Authority to the \AAA{}.} Next, suppose that the \ppp{} can transfer the authority over the search and action directly to the \aaa{}. Doing so completely eliminates the \aaa{}'s preemption motive, but it introduces an additional distortion in the decision process. In particular, if the \aaa{} controls the search, then he waits for news until $\tau_A>\tau_P$ and if this time is reached without real news, then he selects risky: compared to the \ppp{}'s first best, the \aaa{} waits too long for news and selects the wrong default action. Furthermore, unlike the naive search, where the risky action was selected at the deadline only if the \aaa{} is \faker{}, with delegation the risky action is always chosen if the \aaa{}'s deadline is reached. Thus, the \ppp{}'s payoff of delegating authority to the \aaa{} is 
\begin{align*}
	u_P^D&=\int_0^{\tau_A}\exp(-\rho s)\rpdf(s)(\mu+(1-\mu)\psafe)ds+\exp(-\rho\tau_A)(1-\rcdf(\tau_A))\mu\\
	&=u_P^{FB}(\tau_A)-\exp(-\rho\tau_A)(1-\rcdf(\tau_A))(\psafe-\mu).
\end{align*}
As a first step, consider a simple choice for the \ppp{}: either she delegates the search to the \aaa{}, or she terminates the search immediately and selects the safe action. If the \ppp{} would rather terminate than delegate ($u_P^D<\psafe$), then the conflict of interest between the players is \textit{severe}; otherwise it is \textit{mild}. Obviously, with a severe conflict of interest, delegation could never be optimal---the \ppp{}'s equilibrium payoff always exceeds $\psafe$ when she retains authority. Note that  $\psafe$ and $\asafe$ can be arbitrarily close without violating any parametric assumption of the model, and in this case, $u_P^D\approx u_P^{FB}(\tau_P)>\psafe$.\footnote{$\rhaz(0)>\rho/(1-\psafe)$ guarantees (A1). Furthermore, as discussed in footnote 2, (A2) holds if $\asafe\in(\underline{\asafe},\mu)$. Thus,  $\psafe$ and $\asafe$ can be arbitrarily close to each other. It follows that $\psafe\approx \mu$ and $\tau_P\approx\tau_A$.} In other words, when the $\psafe$ and $\asafe$ are close, the conflict of interest is mild. Next, note that as long as $\mu>\psafe$, we have $\bar{\sigma}<1$. By implication, there is a range of large $\sigma$ where the \ppp{}'s payoff in the equilibrium is close to (or equal to) $\psafe$. Continuity of the \ppp{}'s equilibrium payoff as a function of $\sigma$ delivers the following result.

\begin{prop}\label{del}(Delegating to \AAA{}).  If the conflict of interest is mild and the prior probability of the \faker{} \aaa{} is sufficiently high, then the \ppp{}'s payoff from delegating authority to the \aaa{} is higher than her equilibrium payoff in the main model. With a mild conflict and $\sigma>\bar{\sigma}$, the \ppp{}'s payoff with delegation is higher than her payoff in any equilibrium of the main model. In addition, if $\sigma\in(\widetilde{\sigma}_D,\overline{\sigma})$ for some $0<\widetilde{\sigma}_D<\bar{\sigma}$, then the \ppp{}'s payoff with delegation is higher than her payoff in the unique equilibrium of the main model, which features beneficial search.
\end{prop}

Proposition \ref{del} suggests that the credible threat of manipulation of the news process allows the \aaa{} to gain outsized influence over the decision. Indeed, if the conflict of interest is relatively mild and the \aaa{} can probably manipulate the news, the \ppp{} would rather give the \aaa{} full authority over the search and decision than initiate a (possibly) manipulated search. Even though the \aaa{} has no private information about the best action, he is nevertheless granted complete authority over the decision. This contrasts with the classic logic supporting delegation, which is based on the \aaa{}'s ability to access or produce superior payoff-relevant information \citep{AT97, D2002}.

\paragraph{Delegating Authority to an Intermediary.} In this section, a milder form of delegation is analyzed: the \ppp{} delegates the decision authority to a third party whose payoff from the safe action is different than her own. This section focuses on the more interesting case $\sigma<\bar{\sigma}$, which implies that the unique  equilibrium has beneficial search. In addition, we focus on the \ppp{}'s incentive to  delegate ``locally,'' so that the intermediary's payoff from the safe of option,  $\psafe_I$, is not too far from $\psafe$. In this way, we can ensure that (A1) and (A2) are satisfied by the intermediary and that the equilibrium between the intermediary and \aaa{} features beneficial search.

By delegating to an intermediary, \ppp{} alters the equilibrium search for news: the incentives of the \aaa{} and intermediary determine the equilibrium, rather than the \ppp{}. Thus, with delegation the equilibrium is characterized by Proposition \ref{charval} with the intermediary standing in for the \ppp{}.  Building on the comparative statics derived in Corollary \ref{corbel}, consider the effects of delegating authority to an intermediary with a smaller payoff from the safe action. On one hand, such an intermediary has a smaller preference conflict with the \aaa{}. By delegating to such an intermediary, the \ppp{} weakens the incentives for stochastic unraveling, shifting the soft deadline to a later time. Thus by delegating, the \ppp{} extends the duration of the first best search, which increases her expected payoff. On the other hand, delegating authority to an intermediary also affect the \aaa{}'s faking strategy. In equilibrium, the intermediary is indifferent between continuing the search and stopping at all times between the soft and hard deadlines if there has been no news. But the \ppp{}'s payoff of stopping the search and selecting safe is higher, and she therefore prefers to stop at all such times. In other words, when she delegates,  the \ppp{}'s payoff function is decreasing between the soft and hard deadlines. Although her payoff is higher at the soft deadline, if the search extends for too long beyond it, the \ppp{} is harmed. This cost is particularly salient: by delegating, the \ppp{} shifts both the intermediary's stopping strategy and the \aaa{} faking strategy toward later times, increasing the likely duration of the search.

To make this discussion precise, consider the equilibrium of Proposition \ref{charval}, with intermediary playing the role of \ppp{}. To distinguish the equilibrium with delegation, the  intermediary's preference parameter is included explicitly in the notation. Thus, $\tau_P(\psafe_I)$ is the intermediary's first best search duration (and the hard deadline when it has the authority), $\mint(\theta_I)$ is the soft deadline,  $\acdf(\cdot|\psafe_I),\pcdf(\cdot|\psafe_I)$ are the equilibrium strategies of the \aaa{} and intermediary when it plays the role of \ppp{}, and $u_P(t|\psafe_I,\psafe)$ is the \ppp{}'s equilibrium expected payoff conditional on the intermediary's  stopping time $t$. Note that the equilibrium strategies and the hard and soft deadline depend only on the intermediary's safe payoff $\psafe_I$, but the \ppp{}'s safe payoff enters explicitly into $u_P(t|\psafe_I,\psafe)$. Indeed, whenever the intermediary selects the safe action (either because it chose to stop the search or because type-0 news arrived), the \ppp{}'s payoff is $\psafe$. In particular, we have
\begin{align*}
	u_P(t|\psafe_I,\psafe)=
u_P^{FB}(\mint(\psafe_I))+\int_{\mint(\psafe_I)}^t\exp(-\rho s)\{w_0^P(s|\psafe_I)\psafe+w_1^P(s|\psafe_I)\mu_1(s|\psafe_I)\}ds+\exp(-\rho t)W_\phi^P(t|\psafe_I)\psafe,
%
%	u_P(t|\psafe_I,\psafe)=u_P^{FB}(\mint(\psafe_I))+\int_{\mint(\psafe_I)}^t\exp(-\rho s)\{(\mu+(1-\mu)\psafe)\rpdf(s)(1-\sigma \acdf(s|\psafe_I)+\mu\sigma\apdf(s|\psafe_I)(1-\rcdf(s))\}ds+\exp(-\rho t)(1-\rcdf(t))(1-\sigma\acdf(t|\psafe_I))\psafe
\end{align*}
where use has been made of the fact that the intermediary always selects the risky action following type-1 news in equilibrium (see Proposition 1). Recalling that the intermediary's equilibrium strategy has a mass point on $\tau_P(\psafe_I)$ but is otherwise differentiable, we have that the \ppp{}'s expected payoff of delegating to intermediary with preference parameter $\theta_I$ is
\begin{align*}
	U(\psafe_I|\psafe)&=\int_{\mint(\psafe_I)}^{\tau_P(\psafe_I)}\ppdf(s|\psafe_I)u_P(s|\psafe_I,\psafe)ds+(1-\pcdf(\tau_P^-(\psafe_I)|\psafe_I))u_P(\tau_P(\psafe_I)|\psafe_I,\psafe)\\
	&=\int_{\mint(\psafe_I)}^{\tau_P(\psafe_I)}(1-\pcdf(s|\psafe_I))u'_P(s|\psafe_I,\psafe)ds+u^{FB}_P(\mint(\psafe_I)),
\end{align*}
where the last line follows from an integration by parts.\footnote{See Proof of Proposition \ref{delint}.} In this form, tradeoff is apparent: by delegating to an intermediary with a smaller payoff of the safe action, the \ppp{} extends the soft deadline and the duration of first best search, but if the soft deadline is reached with no news continued search is strictly harmful. 
\begin{figure}
\begin{minipage}{0.5\textwidth}
%LEFT PANEL
\begin{center}
\begin{tikzpicture}[scale=0.8]

%AXES
\draw[-stealth] (0,0) -- (0,9); 
\draw[-stealth] (0,0) -- (9,0); 

%AXES LABELS
\fill (9,0) node[right] {\footnotesize{$t$}};

%support
\fill (7,-0.13) node[below] {\footnotesize{$\tau_P$}};
\fill (8,0) node[below] {\footnotesize{$\tau_P(\psafe_I)$}};
\filldraw[color=black, fill=black] (7,0) circle (2.5pt);
\filldraw[color=black, fill=black] (8,0) circle (2.5pt);
\fill (3,-0.13) node[below] {\footnotesize{$\mint$}};
\filldraw[color=black, fill=black] (4.5,0) circle (2.5pt);
\fill (4.5,0) node[below] {\footnotesize{$\mint(\psafe_I)$}};
\filldraw[color=black, fill=black] (3,0) circle (2.5pt);
\draw[line width=0.5pt, black, dashed] (7,8) -- (7,0); 
\draw[line width=0.5pt, black, dashed] (8,8) -- (8,0);

%ACDF
\draw[line width=0.5pt, red,dashed] (0,0.03) -- (3,0.03); 
\draw[line width=0.5pt, red] (8,8.03) -- (9,8.03); 
\draw[line width=0.5pt, blue] (8,8) -- (9,8); 
\draw[line width=0.5pt, red,dashed] plot[smooth] coordinates {(3,0) 
(4,5) (5,7) (6,7.8) (7,8+0.03)};
\fill (7.15,3.8) node[left] {\footnotesize{$\acdf(\cdot|\psafe_I)$}};

%PCDF
\draw[line width=0.5pt, blue] (0,-0.03) -- (4.5,-0.03); 
\draw[line width=0.5pt, red] (0,0.03) -- (4.5,0.03); 
\draw[line width=0.5pt, blue,dashed] (7,8) -- (8,8); 
\filldraw[color=blue, fill=blue,opacity=0.3] (7,8) circle (2pt);
\filldraw[color=blue, fill=white,opacity=0.3] (7,5.1) circle (2pt);
\draw[line width=0.5pt, blue,dashed] plot[smooth] coordinates {(3,0) (4,3) (5,4.25) (6,4.8) (7,5.1)};
\fill (7,1.2) node[left] {\footnotesize{$\pcdf(\cdot|\psafe_I)$}};

\draw[line width=0.5pt, blue] (4.5,0) .. controls (5,2) and (6,3) .. (8,3.2);
\filldraw[color=blue, fill=blue] (8,8) circle (2.5pt);
\filldraw[color=blue, fill=white] (8,3.2) circle (2.5pt);
\draw[line width=0.5pt, red] (4.5,0) .. controls (5,3) and (6,8) .. (8,8);

%\draw[blue,line width=1pt] plot[variable=\x,domain=0:8,smooth]
%({\x},{(1-\gam+\gam*exp(-\lam*\x))^((1-\th)*\m/(\th-\minutesdate)});
%

%COMMANDS
%\fill (0,0) node[left] {\footnotesize{$\mu_a^{PS}$}};
%\draw[line width=0.5pt] (0,-1.5) -- (0,7); 
%\filldraw[color=black, fill=black] (1.5,0) circle (3pt);
%\draw[line width=0.5pt,red,dashed] plot[smooth] (3.5,0.7956)--(3.85,0.9636)--(4.2,1.1556)--
%(4.55,1.3716)--(4.9,1.6116)--(5.25,1.8756)--(5.6,2.1636)--(5.95,2.4756)--
%(6.3,2.8116)--(6.65,3.1716)--(7,3.5556);
\end{tikzpicture} \captionsetup{font=footnotesize}
\captionof{figure}{Strategies, Delegation $\acdf(\cdot|\psafe_I),\pcdf(\cdot|\psafe_I)$\label{stratdel}}
\end{center}
\end{minipage}
\begin{minipage}{0.5\textwidth}
\begin{center}
\begin{tikzpicture}[scale=0.8]

%AXES
\draw[-stealth] (0,0) -- (0,9); 
\draw[-stealth] (0,0) -- (8,0); 

%AXES LABELS
\fill (8,0) node[right] {\footnotesize{$t$}};

\fill (7,-0.13) node[below] {\footnotesize{$\tau_P$}};
\fill (8,0) node[below] {\footnotesize{$\tau_P(\psafe_I)$}};
\filldraw[color=black, fill=black] (7,0) circle (2.5pt);
\filldraw[color=black, fill=black] (8,0) circle (2.5pt);
\fill (3,-0.13) node[below] {\footnotesize{$\mint$}};
\filldraw[color=black, fill=black] (4.5,0) circle (2.5pt);
\fill (4.5,0) node[below] {\footnotesize{$\mint(\psafe_I)$}};
\filldraw[color=black, fill=black] (3,0) circle (2.5pt);
\draw[line width=0.5pt,dashed] (3,0) -- (3,6.3); 
\draw[line width=0.5pt,dashed] (4.5,0) -- (4.5,6.7); 
\draw[line width=0.5pt,dashed] (7,0) -- (7,6.3); 
\draw[line width=0.5pt,dashed] (8,0) -- (8,5.5); 
%COMMANDS
%\fill (0,1) node[left] {\footnotesize{$\mu$}};
\fill (0,4) node[left] {\footnotesize{$\psafe$}};
%\draw[line width=0.5pt,dashed] (0,3) -- (3,3); 
\draw[line width=0.5pt,dashed,black] (0,6.3) -- (3,6.3); 
\fill (0,6.2) node[left] {\footnotesize{$u_P^{FB}(\mint)$}};

%\fill (0,6.8) node[left] {\footnotesize{$u_P^{FB}(\mint(\psafe))$}};

%\draw[line width=0.5pt] (0,-1.5) -- (0,7); 
%\filldraw[color=black, fill=black] (1.5,0) circle (3pt);
%\draw[line width=0.5pt,red] (0,1) .. controls (0,1) and (1,2.5) .. (3,3);
%\draw[line width=0.5pt,red] (3,3) -- (7,3);
%\fill (7,3.5) node[right] {\footnotesize{$u_A(\cdot)$}};
%\fill (0,3) node[left] {\footnotesize{$u_A^{FB}(\mint)$}};

%\draw[line width=0.5pt,red, dashed] (3,3) .. controls  (4,3.3) and (5,3.7)  .. (8,4);

\draw[line width=0.5pt,blue] (0,4) .. controls (0,4) and (1,5.5) .. (3,6.3);

\draw[line width=0.5pt,blue] (3,6.3) .. controls (3.5,6.55) and (4.25,6.65) .. (4.5,6.7);
\draw[line width=0.5pt,blue] (4.5,6.7) .. controls (6,6) and (7,5.6) .. (8,5.5);

%\draw[line width=0.5pt,dashed] (0,6.7) -- (4.5,6.7); 
%\fill (7,6.7) node[right] {\footnotesize{$u_P(\cdot|\psafe_I,\psafe)$}};

%\draw[line width=0.5pt,blue,dashed] (3,6.3) .. controls (4,6.75) and (7,7) .. (8,6.75);
\draw[line width=0.5pt,dashed,blue] (3,6.3) -- (7,6.3);
\draw[line width=0.5pt,blue,dashed] (7,6.3) .. controls (7.25,6.29) and (7.75,6.15) .. (8,5.95);

%\draw[line width=0.5pt,black,dashed] (3,0) -- (3,6.3);
%\draw[line width=0.5pt,black,dashed] (7,0) -- (7,6.3);
%\filldraw[color=red, fill=white] (7,3) circle (2.5pt);
%\filldraw[color=red] (7,2.1) circle (2.5pt);
%\draw[line width=0.5pt,red] (7,2.1) -- (8,2.1);
%plot[smooth, tension=0.1] (0,1)--(1,2.5)--(2,3)--(3,3.4);
% (3.5,0.7956)--(3.85,0.9636)--(4.2,1.1556)--
%(4.55,1.3716)--(4.9,1.6116)--(5.25,1.8756)--(5.6,2.1636)--(5.95,2.4756)--
%(6.3,2.8116)--(6.65,3.1716)--(7,3.5556);
\end{tikzpicture}
\captionsetup{font=footnotesize}
\captionof{figure}{Delegation payoff  $u_P(\cdot|\psafe_I,\psafe)$\label{paydel}}
\end{center}
\end{minipage}
\end{figure}

The tradeoff of delegating is illustrated graphically in Figures \ref{stratdel} and \ref{paydel}. Figure \ref{stratdel} illustrates the shift toward later stopping and faking times induced by delegation to an intermediary with a smaller payoff from the safe option. The dashed curves represent the equilibrium strategies when the \ppp{} keeps authority for herself, while the solid curves are the strategies under delegation. Figure \ref{paydel} illustrates the \ppp{}'s payoff when she delegates to an intermediary with safe payoff $\psafe_I$, and when the intermediary select stopping time $t$. The solid blue curve is this payoff, $u_P(\cdot|\psafe_I,\psafe)$, while the dashed blue curve is the \ppp{}'s payoff from keeping authority herself.

 The following proposition shows that the benefits of extending the first best search outweigh the harm.

\begin{prop}\label{delint}(Delegating to an Intermediary.) If $\sigma<\bar{\sigma}$, then the \ppp{} benefits by delegating decision authority to an intermediary whose payoff from the safe option is smaller than her own. In particular, a $\psafe_I^*<\psafe$ exists, such that the \ppp{} benefits from delegation if $\psafe_I\in(\psafe^*_I,\psafe)$. Furthermore, if the \ppp{} delegates authority to such an intermediary, then in equilibrium she is never tempted to overrule the intermediary's chosen action.
\end{prop}

The second part of Proposition \ref{delint} highlights another difference between delegating to an intermediary and the  \aaa{}. When delegating authority to the \aaa{}, the \ppp{} is tempted to overrule the \aaa{} if he tries to select risky at time $\tau_A$. When delegating to an intermediary with a small bias, the \ppp{} will not be tempted to overrule the intermediary's choice of action. Indeed, whenever the intermediary selects risky, the \ppp{} also prefers risky. As shown in Corollary \ref{corbel}, the posterior belief that $\omega=1$ given an arrival of type-1 news exceeds $\psafe_I$, and furthermore, the gap is strictly positive. If $\psafe_I$ and $\psafe$ are close, then the same is also true of $\psafe$.

\section{Conclusion}

I have studied a model of information acquisition, in which information can be faked strategically.  A \ppp{} must choose between a safe action with a known payoff and a risky action with an uncertain payoff. An arrival of real news reveals the payoff of the risky action, arriving at an uncertain future time. The \ppp{} prefers the safe action under the prior, but she can wait for news to arrive. An \aaa{} has a smaller payoff of the safe action and prefers risky under the prior. The \aaa{} may have the ability to fabricate a news arrival that is indistinguishable from a real one.

The disagreement over the best action under the prior belief incentivizes preemption by the \aaa{}. whenever the \ppp{} stops before news arrives, she selects risky. Thus, if the \aaa{} expects that the \ppp{} is about to stop, he has an incentive to fake a news arrival, in order to induce the risky action, which he prefers under the prior belief. The \aaa{}'s faking undermines the value of future news, which creates an incentive for her to stop before such news arrives. 

From a positive perspective, two distortions arise in equilibrium. When the probability that the \aaa{} can manipulate news is not too high, the equilibrium features an early phase of undistorted search, followed by a phase in which \aaa{} randomly fakes news and the \ppp{} randomly stops before news arrives. Though it is sometimes manipulated, all news is sufficiently informative for the \ppp{} to follow. From a normative perspective, the possibility of faking undermines the value of information and reduces the \ppp{}'s payoff. If the probability that the \aaa{} can fabricate news is sufficiently high, the benefits of news are undermined completely. I study three straightforward remedies that mitigate the \aaa{}'s preemption motive and increase the value of search.

The model can be extended in a number of interesting directions. The current version studies only a single dynamic interaction between the players. It might be interesting to consider incentives of an \aaa{} who interacts with \ppp{}s repeatedly. In this case, the \aaa{} may be concerned about revealing his ability to fake to future \ppp{}s, which may temper his incentive to manipulate. Another possibility is to consider stronger commitment power for the \ppp{}. While I have focused on relatively simple forms of commitment that dampen the \aaa{}'s preemption motive, it would be interesting to know how the \aaa{}'s ability to fake news allows him to obtain rents, even when the \ppp{} can commit to her entire strategy. I plan to address these and other related questions in future work.

\bibliography{fake}
\bibliographystyle{rfs}

\section{Appendix}
\setlength{\parindent}{0pt}

\begin{proof}[Proof of Lemma \ref{fb}]

Consider the first-best benchmark for player $i$. Let $(d_i,s_i)$ be player $i$'s expected payoff of the default action, and safe action respectively. That is $(d_P,s_P)=(\psafe,\psafe)$ and $(d_A,s_A)=(\mu,\asafe)$. We have
\begin{align*}
	u_i^{FB}(t)=\int_0^t\exp(-\rho x)\rpdf(x)(\mu+(1-\mu)s_i)dx+\exp(-\rho t)(1-\rcdf(t))d_i.
\end{align*}
$\rcdf(\cdot)$ is differentiable, and therefore
\begin{align*}
\frac{d}{dt}u_i^{FB}(t)&=\exp(-\rho t)\rpdf(t)(\mu+(1-\mu)s_i)-d_i\exp(-\rho t)(\rho(1-\rcdf(t)+\rpdf(t))\\
&=\exp(-\rho t)(1-\rcdf(t))(\mu+(1-\mu)s_i-d_i)\{\rhaz(t)-\frac{d_i\rho}{\mu+(1-\mu)s_i-d_i}\}.
\end{align*}
Note that $d_i<\mu+(1-\mu)s_i-d_i$ for $i\in\{A,P\}$. It follows that the first three terms are positive. The sign of the derivative is therefore determined by the sign of the fourth term. Because $\rhaz(\cdot)$ is decreasing, the derivative of $u_i^{FB}(\cdot)$ crosses zero at most once, and if it exists, this crossing is from above. It follows that $u_i^{FB}(\cdot)$ is single-peaked. Under A1 and A2, the peak is at a strictly positive $t$. Substituting $(d_i,s_i)$, for $i\in\{A,P\}$, we find that the crossing occurs at $\tau_i$ such that $\rhaz(\tau_i)=\phi_i$. A1 and A2 imply $\tau_A>\tau_P>0$.\end{proof}

\subsection{Proof of Proposition \ref{strucval}}
The proof makes use of a number of lemmas, stated in the next section. To simplify the exposition, the proofs of the lemmas are presented in the  Supplemental Appendix. 
\subsection{Lemmas}
\noindent \textbf{Observation 1.} Note that 
\begin{align*}
	\int_0^tw^P_0(s)+w^P_1(s)ds+W^P_\phi(t)=1.
\end{align*}
Indeed, the integral is the probability that a type-0 or type-1 arrival occurs at times less than or equal to $t$, while $W_\phi^P(t)$ is the probability that an arrival of news takes longer than $t$. Rearranging, for $t_1\leq t_2$ we have
\begin{align*}
	\int_{t_1}^{t_2}w^P_0(s)+w^P_1(s)ds=W^P_\phi(t_1)-W^P_\phi(t_2).
\end{align*}
Similarly, for the \aaa{} we have,
\begin{align*}
	\int_{t_1}^{t_2}w^A_0(s)+w^A_1(s)+w_S^A(s)ds=W^A_\phi(t_1)-W^A_\phi(t_2).
\end{align*}

\begin{lemma}\label{pbounds} If (\ref{PICa}), then inequalities $(i),(ii),(iii)$ hold,
\begin{enumerate}
	\item[(i)] $u_P(t')-u_P(t)\leq (1-\sigma \acdf(t))\{u_P^{FB}(t')-u_P^{FB}(t)\}$ for $t<t'$. 
	\item[(ii)] $u_P(t)<u_P(\tau_P)$ for $\tau_P<t\leq\infty$.

	\item[(iii)] $u_P(t')-u_P(t)\geq -(\exp(-\rho t)-(\exp(-\rho t'))W_\phi^P(t)\psafe$ for $t<t'$.
	\item[(iv)] Part $(i)$ holds as an equality if $\acdf(t)=\acdf(t')$.
\end{enumerate}
\end{lemma}

\begin{lemma}\label{pcont} If (\ref{PICa}), then $u_P(\cdot)$ is continuous, regardless of the \aaa{}'s strategy.
\end{lemma}
%\begin{proof}[Proof of Lemma \ref{pcont}] 
%
%Using Parts (i) and (iii) of Lemma \ref{pbounds}, for $t'>t$
%\begin{align*}
%-(\exp(-\rho t)-(\exp(-\rho t'))W_\phi^P(t)\psafe	\leq u_P(t')-u_P(t)\leq (1-\sigma\acdf(t))\{u_P^{FB}(t')-u_P^{FB}(t)\}.
%\end{align*}
%By continuity of $u_P^{FB}(\cdot)$ and $\exp( \cdot)$, we have $\lim_{t'\rightarrow t^+}u_P(t')=u_P(t)$.
%Similarly, for $t'<t$
%\begin{align*}
%-(\exp(-\rho t')-(\exp(-\rho t))W_\phi^P(t')\psafe	\leq u_P(t)-u_P(t')\leq (1-\sigma \acdf(t'))\{u_P^{FB}(t)-u_P^{FB}(t')\},
%\end{align*}
%Note that (1) $W_\phi^P(t')\leq W_\phi^P(t)$,  (2) $(\exp(-\rho t')-(\exp(-\rho t))>0$, and (3) $\acdf(t')\geq 0$. It follows that
%\begin{align*}
%-(\exp(-\rho t')-(\exp(-\rho t))W_\phi^P(t)\psafe	\leq u_P(t)-u_P(t')\leq u_P^{FB}(t)-u_P^{FB}(t').
%\end{align*}
%Continuity of $u^{FB}(\cdot)$ and $\exp(\cdot)$ imply $\lim_{t'\rightarrow t^-}u_P(t')=u_P(t)$.\end{proof}

\begin{lemma}\label{ptop} If (\ref{PICa}) and (\ref{PICt}), then $\pcdf(\tau_P)=1$.
\end{lemma}
%\begin{proof}[Proof of Lemma \ref{ptop}] From Lemma \ref{pbounds}, we have $u_P(t)<u_P(\tau_P)$ for $\tau_P<t\leq \infty$. From (\ref{PICt}), we have $\ppdf(t)=0$ for such $t$.  By implication, $\pcdf(\tau_P)=1$.
%\end{proof}
\noindent \textbf{Definition:} Let
\begin{align*}
&v(t)\equiv \int_0^t\exp(-\rho x)\rpdf(x)(\mu+(1-\mu)\asafe)dx+\exp(-\rho t)(1-\rcdf(t))\asafe.
%=\\
%&u_A^{FB}(t)-\exp(-\rho t)(1-\rcdf(t))(\mu-\asafe).
\end{align*}
%By a calculation similar to Lemma \ref{fb}, 
%\begin{align*}
%	v'(t)> 0\iff \rhaz(t)>\frac{\rho\asafe}{\mu(1-\asafe)}.
%\end{align*}
%Note that $\asafe<\psafe\Rightarrow \frac{\rho\asafe}{\mu(1-\asafe)}<\phi_P$. Therefore, $v(\cdot)$ is strictly increasing for $t\leq\tau_P$.

\begin{lemma}\label{abounds} Inequalities $(i),(ii)$ hold for $t<t'\leq \tau_P$, regardless of the \ppp{}'s strategy. 
	\begin{align*}
(i)\quad &	u_A(t')-u_A(t)\leq\\
&	(1-\pcdf(t))\{v(t')-v(t)\}+
\exp(-\rho t')W_\phi^A(t')(\mu-\asafe) a(t')-\exp(-\rho t)W_\phi^A(t)(\mu-\asafe) a(t).\\
(ii)\quad	&u_A(t')-u_A(t)\geq\\
	&	\{\exp(-\rho t')W_\phi^A(t')a(t')-\exp(-\rho t)W_\phi^A(t)a(t)\}(\mu-\asafe)
-\{\exp(-\rho t)-\exp(-\rho t')\}W_\phi^A(t)\asafe
\end{align*}
 $(iii)$: if $\pcdf(t)=\pcdf(t')$,  $a(t)=a(t')=1$, and $a(s)=1$ for almost all $s\in(t,t')$, then
\begin{align*}
u_A(t')-u_A(t)=(1-\pcdf(t))(u_A^{FB}(t')-u_A^{FB}(t)).
\end{align*}
\end{lemma}

\begin{lemma}\label{aball} Suppose (\ref{AIC}). If $a(t)<1$ and $\pcdf(t)<1$ for some $t$, then
\begin{enumerate}
	\item[(i)] an $\epsilon'>0$ such that $a(t')<1$ for $t'\in[t,t+\epsilon')$.
	\item[(ii)] an $\epsilon''>0$ exists such that for $a(t'')<1$ for $t''\in(t-\epsilon'',t]$, provided $t>0$.
\end{enumerate}
\end{lemma}

\begin{lemma}\label{comp1} Suppose an interval $0\leq t_1<t_2\leq \tau_P$ exist such that $\acdf(t_1)=\acdf(t_2)$. If (\ref{PICa}) and (\ref{PICt}), then
\begin{enumerate}
\item[(i)] $\ppdf(t)=0$ for all $t\in[t_1,t_2)$.
\item[(ii)] If $t_1>0$, then there exists $t_0<t_1$ such that $\ppdf(t)=0$ for all $t\in[t_0,t_2)$.
\end{enumerate}

\end{lemma}
%\begin{proof}[Proof of Lemma \ref{comp1}] 
%
%\textit{Part (i)}. Note that if $\acdf(t_1)=\acdf(t_2)$, then $\acdf(t_1)=\acdf(t)=\acdf(t_2)$ for all $t\in[t_1,t_2)$. In addition we have assumed (\ref{PICa}). From Lemma \ref{pbounds} part (iv), we have $u_P(t_2)-u_P(t)=(1-\sigma\acdf(t_1))(u_P^{FB}(t_2)-u_P^{FB}(t))$. We have assumed $t_2\leq \tau_P$. From Lemma \ref{fb}, $u^{FB}_P(t_2)>u^{FB}_P(t)$. Therefore, $u_P(t_2)>u_P(t)$. From (\ref{PICt}), we have $\ppdf(t)=0$.
%
%\textit{Part (ii)}. In the proof of (i) we have shown that $u_P(t_1)<u_P(t_2)$. From continuity of $u_P(\cdot)$ (Lemma \ref{pcont}), we have $u_P(t)<u_P(t_2)$ for all $t\in[t_1-\epsilon,t_1)$, for sufficiently small $\epsilon>0$. From (\ref{PICt}), we have $\ppdf(t)=0$ for such $t$. Together with part (i) we have $\ppdf(t)=0$ for $t\in[t_1-\epsilon,t_2)$.
%\end{proof}

\begin{lemma}\label{comp2} Suppose an interval $[t_1,t_2)$ exists such that $\pcdf(t_1)=\pcdf(t_2)<1$. In equilibrium, $\apdf(t)=0\,$ for all $t\in[t_1,t_2)$.	
\end{lemma}

\begin{lemma}\label{gapexp} Suppose $0<t_1<t_2\leq \tau_P$ and $\pcdf(t_1)=\pcdf(t_2)<1$. In equilibrium, $\pcdf(t)=\pcdf(t_2)$ for all $t\in[0,t_2]$.
\end{lemma}
%\begin{proof}[Proof of Lemma \ref{gapexp}] 
%
%Consider $T_P=\{t|\pcdf(t)=\pcdf(t_2)\}$. Because $T_P$ is bounded from below by 0 and $\pcdf(\cdot)$ is right-continuous, $T_P$ has a minimum, $\underline{t}_1\in[0,t_1]$. To prove the result, it will be shown that $\underline{t}_1=0$. 
%
%Suppose that $\underline{t}_1>0$. Consider interval $[\underline{t}_1,t_2')$ where $t_1<t_2'<t_2$. Applying Lemma \ref{comp2}, we have $\apdf(t)=0$ for all $t\in[\underline{t}_1,t_2')$. Next, consider interval $[\underline{t}_1,t_2'']$ where $t_1<t_2''<t_2'$. Because $\apdf(t)=0$ for $t\in[t_1,t_2'']$, we have $\acdf(t_1)=\acdf(t_2'')$. Furthermore, $t_2''<t_2<\tau_P$, where the last  inequality follows from $\pcdf(t_2)<1$ and Lemma \ref{ptop}. Using Lemma \ref{comp1} part (ii), there exists $t_0<\underline{t}_1$ such that $\ppdf(t)=0$ for all $t\in[t_0,t_2'')$.  Consider interval $[t_0,t_2''']$ where $t_1<t_2'''<t_2''$. Because $\ppdf(t)=0$ for all such $t$, we have $\pcdf(t_0)=\pcdf(t_2''')$. Furthermore, $t_2'''\in[t_1,t_2]$, and therefore $\pcdf(t_2''')=\pcdf(t_2)$. Hence, $\pcdf(t_0)=\pcdf(t_2)$. Because $t_0<\underline{t}_1$, we have that $\underline{t}_1$ is not the minimum of $T_P$.
%\end{proof}

\begin{lemma}\label{gapexa} Suppose $0<t_1<t_2$, $\pcdf(t_2)<1$, and $\acdf(t_1)=\acdf(t_2)$. In equilibrium, $\acdf(t)=\acdf(t_2)$ for all $t\in[0,t_2]$.
\end{lemma}

%\begin{proof}[Proof of Lemma \ref{gapexa}]
%
%Consider $T_A=\{t|\acdf(t)=\acdf(t_2)\}$. Because $T_A$ is bounded from below by 0 and $\acdf(\cdot)$ is right-continuous, $T_A$ has a minimum, $\underline{t}_1\in[0,t_1]$. To prove the result, it will be shown that $\underline{t}_1=0$. 
%
%Suppose that $\underline{t}_1>0$. Consider interval $[\underline{t}_1,t_2')$ where $t_1<t_2'<t_2$. Applying Lemma \ref{comp1} part (ii), we have $\ppdf(t)=0$ for all $t\in[t_0,t_2')$, for some $t_0<\underline{t}_1$. By implication, $\ppdf(t)=0$ for all $t\in[t_0,t_2'']$, where $t_1<t_2''<t_2'$. Hence, $\pcdf(t_0)=\pcdf(t_2'')$. Furthermore, because $t_2''<t_2$, we have $\pcdf(t_2'')\leq \pcdf(t_2)<1$. Applying Lemma \ref{comp2}, we have $\apdf(t)=0$ for all $t\in[t_0,t_2'')$. By implication, $\apdf(t)=0$ for all $t\in[t_0,t_2''']$, where $t_1<t_2'''<t_2''$. Therefore, $\acdf(t_0)=\acdf(t_2''')=\acdf(t_2)$, where the last equality follows from $t_2'''\in(t_1,t_2)$. Because $t_0<\underline{t}_1$ and $\acdf(t_0)=\acdf(t_2)$, we have that $\underline{t}_1$ is not the minimum of $T_A$.	
%\end{proof}

\subsection{Proof of Proposition \ref{strucval}}

The Proposition is proved by establishing a number of claims. Each claim assumes the existence of an equilibrium with beneficial search, and its conclusion describes a property of such an equilibrium.

 Let $\bar{t}_i$, for $i\in\{P,A\}$ denote the maximum of the support of the \ppp{}'s and \aaa{}'s stopping strategies, respectively.

\begin{claim}\label{tsupp} $\bar{t}_A=\bar{t}_P=\tau_P$.
\end{claim}
\begin{proof}[Proof of Claim \ref{tsupp}] 

\textit{Step 1}: $\bar{t}_P>0$. Suppose $\bar{t}_P=0$. It follows immediately that $\pcdf(t)=1$ for all $t\geq 0$. That is, the \ppp{} terminates the search at time 0. Regardless of the \aaa{}'s strategy, $u_P(0)=\psafe$. Therefore, the \ppp{} does not benefit from search.

\textit{Step 2}: $\bar{t}_P\leq \tau_P$. Follows immediately from Lemma \ref{ptop}.

\textit{Step 3}: $\bar{t}_A>0$. Suppose $\bar{t}_A=0$, i.e., $\acdf(t)=1$ for $t\geq 0$. It follows immediately that $\apdf(t)=0$ for all $t>0$. From Lemma \ref{comp1}, we have $\ppdf(t)=0$ for all $t\in[0,\tau_P)$. From Lemma \ref{ptop}, we have $\pcdf(t)=\mathcal{I}(t\geq \tau_P)$. From Lemma \ref{abounds} part (iii), the $u_A(t)-u_A(0)=u_A^{FB}(t)-u_A^{FB}(0)$ for $0<t<\tau_P$. Because $t<\tau_P<\tau_A$, Lemma \ref{fb} implies $u_A^{FB}(t)-u_A^{FB}(0)>0$. Thus $\apdf(0)=\delta(0)$ violates \ref{AIC}.

\textit{Step 4:} $\bar{t}_A\geq \bar{t}_P$. Suppose $\bar{t}_A<\bar{t}_P$. For $t\in(\bar{t}_A,\bar{t}_P)$, we have $\apdf(t)=0$, and therefore $\acdf(t)=\acdf(\bar{t}_A)$. Furthermore, $\bar{t}_A<\bar{t}_P$ implies $\pcdf(\bar{t}_A)<1$. Applying Lemma \ref{gapexa}, we have $\acdf(t)=\acdf(\bar{t}_A)$ for all $t\in[0,\bar{t}_A]$. In other words, the CDF of the \aaa{}'s strategy is constant from 0 up to and including the top of the support $\bar{t}_A$. It follows that $\bar{t}_A=0$, contradicting Step 3.

\textit{Step 5:} For any $\epsilon>0$, there exists $t\in(\bar{t}_P-\epsilon,\bar{t}_P)$ such that $a(t)=1$. 

Suppose not. An $\epsilon>0$ exists such that $a(t)<1$ for all $t\in(\bar{t}_P-\epsilon,\bar{t}_P)$. It follows that for all $t\in(\bar{t}_P-\epsilon,\bar{t}_P]$, we have $\apdf(t)>0$ and $\bl(t)=\psafe$. Therefore, for $\bar{t}_P-\epsilon\leq t''<t\leq \bar{t}_P$ we have 
\begin{align*}
	u_P(t)-u_P(t'')=
	&\int_{t''}^t\exp(-\rho s)(w^P_0(s)+w^P_1(s))\psafe ds+\exp(-\rho t)W_\phi^P(t)\psafe-\exp(-\rho t'')W_\phi^P(t'')\psafe\\
\leq\exp(-\rho t'')&\int_{t''}^t(w^P_0(s)+w^P_1(s))\psafe ds+\exp(-\rho t)W_\phi^P(t)\psafe-\exp(-\rho t'')W_\phi^P(t'')\psafe.
\end{align*}
Using Observation 1,  we have
\begin{align*}
u_P(t)-u_P(t'')\leq 
&\exp(-\rho t'')(W_\phi^P(t'')-W_\phi^P(t))\psafe+\exp(-\rho t)W_\phi^P(t)\psafe-\exp(-\rho (t'')W_\phi^P(t'')\psafe\\
=&-\{\exp(-\rho t'')-\exp(-\rho t)\}W_\phi^P(t)\psafe<0,
\end{align*}
where the last inequality follows from $W_\phi^P(t)=(1-\sigma\acdf(t))(1-\rcdf(t))>0$ for all $t\geq 0$.
It follows from (\ref{PICt}) that $\ppdf(t)=0$ for all $t\in(\bar{t}_P-\epsilon,\bar{t}_P]$, contradicting the assumption that $\bar{t}_P$ is the maximum of the support of the \ppp{}'s stopping strategy.

\textit{Step 6}: $\bar{t}_A\leq\bar{t}_P$. Consider any $t'>\bar{t}_P$. From Step 5 we have that for any $\epsilon$, there exists some $t_\epsilon\in(\bar{t}_P-\epsilon,\bar{t}_P)$ such that $a(t_\epsilon)=1$. It follows that
\begin{align*}
u_A(t')-u_A(t_\epsilon)=\underbrace{u_A(t')-u_A(\bar{t}_P)}_{A}+\underbrace{u_A(\bar{t}_P)-u_A(t_\epsilon)}_{B}.
\end{align*}
Focus first on $A$. For $t\in[\bar{t}_P,t']$ we have $\pcdf(t)=1$, and for $t\in(\bar{t}_P,t']$ we have $\ppdf(t)=0$. Thus,
\begin{align*}
	u_A(t')-u_A(\bar{t}_P)=0.
\end{align*}
Next, focus on $B$. From Lemma \ref{abounds}, we have
\begin{align*}
	u_A(\bar{t}_P)-u_A(t_\epsilon)&\leq	(1-\pcdf(t_\epsilon))\{v(\bar{t}_P)-v(t_\epsilon)\}-\exp(-\rho t_\epsilon)W_\phi^A(t_\epsilon)(\mu-\asafe)\Rightarrow\\
\frac{u_A(\bar{t}_P)-u_A(t_\epsilon)}{1-\pcdf(t_\epsilon)}&\leq v(\bar{t}_P)-v(t_\epsilon)-\exp(-\rho t_\epsilon)(1-\rcdf(t_\epsilon))(\mu-\asafe).
\end{align*}
Because $v(\cdot)$ is continuous, for $\epsilon$ sufficiently small, we have $v(\bar{t}_P)-v(t_\epsilon)\approx 0$. Simultaneously, $\exp(-\rho t_\epsilon)(1-\rcdf(t_\epsilon))(\mu-\asafe)>\exp(-\rho \tau_P)(1-\rcdf(\tau_P))(\mu-\asafe)>0$, where the last inequality follows because $\exp(-\rho\cdot)(1-\rcdf(\cdot))$ is decreasing and Step 2 ($t_\epsilon<\bar{t}_P<\tau_P$). Therefore, for $\epsilon$ sufficiently small, there exists $t_\epsilon$ such that RHS is strictly negative. It follows that $B$ is strictly negative. Recalling that $A$ is 0, we have $u_A(t')-u_A(t_\epsilon)<0$. Therefore, a $t_\epsilon$ exists such that $u_A(t')-u_A(t_\epsilon)<0$ for all $t'>\bar{t}_P$. From (\ref{AIC}), we have $\apdf(t')=0$ for all such $t'$. It follows that $\bar{t}_A\leq \bar{t}_P$.

\textit{Step 7:} $\bar{t}_P=\bar{t}_A=\tau_P$. Combining Steps  2, 4, and 6 we have $\bar{t}_P=\bar{t}_A\leq \tau_P$. Suppose that the inequality is strict, i.e., $\bar{t}_P<\tau_P$. Because $\bar{t}_A=\bar{t}_P$, for all $t\in[\bar{t}_P,\tau_P]$ we have $\acdf(t)=1$. Using Lemma \ref{pbounds} part (iv), $u_P(t)-u_P(\bar{t}_P)=(1-\sigma\acdf(\bar{t}_P))(u^{FB}(t)-u^{FB}(\bar{t}_P))$ for all $t\in[\bar{t}_P,\tau_P]$. Note that  $u^{FB}(t)-u^{FB}(\bar{t}_P)>0$ for such $t$ (see Lemma \ref{fb}). Therefore, $u_P(t)>u_P(\bar{t}_P)$. Note further that continuity of $u_P(\cdot)$ (see Lemma \ref{pcont}) implies that $u_P(\bar{t}_P)=u_P^*$, the \ppp{}'s equilibrium payoff. Thus, $u_P(t)>u_P^*$, which contradicts (\ref{PICt}).\end{proof}

 Let $\underline{t}_i$, $i\in\{P,A\}$, denote the bottom of the support of the \ppp{} and \aaa{} stopping strategy, respectively.

\begin{claim}\label{bsupp} $\underline{t}_P=\underline{t}_A=\tau_M$, for some $\tau_M>0$.
\end{claim}	
\begin{proof}[Proof of Claim \ref{bsupp}]

\textit{Step 1}: $\underline{t}_P>0$. From Lemma \ref{pcont}, in equilibrium $u_P(\cdot)$ is continuous. It follows that $u_P(\underline{t}_P)=u_P^*$, the \ppp{}'s equilibrium payoff. Note that $u_P(0)=\psafe$, regardless of the \aaa{}'s strategy. Therefore, if $\underline{t}_P=0$, then $u_P^*=\psafe$. In this case, the \ppp{} does not benefit from search.

\textit{Step 2}: $\underline{t}_A\leq\underline{t}_P$. First, suppose $\underline{t}_A=0$. From Step 1, we have $\underline{t}_P>0$; thus Step 2 follows immediately. Second, consider  $\underline{t}_A>0$. Note that $\underline{t}_A\leq\bar{t}_A=\tau_P$, where the equality is Claim \ref{tsupp}. Next, consider any $t'\in(0,\underline{t}_A)$. Because $t'<\underline{t}_A$, we have $\acdf(t')=\acdf(0)=0$. From Lemma \ref{comp1}, we have $\ppdf(t)=0$ for all $t\in[0,t')$. Because $t'$ can be chosen arbitrarily close to $\underline{t}_A$, we have $\ppdf(t)=0$ for all $t\in[0,\underline{t}_A)$. Therefore, $\underline{t}_P\geq\underline{t}_A$.

\textit{Step 3}: $\underline{t}_A\geq\underline{t}_P$. Consider any $t'\in[0,\underline{t}_P)$. For such $t'$, we have $\pcdf(0)=\pcdf(t')=0$. Using Lemma \ref{comp2}, we have $\apdf(t)=0$ for all $t\in[0,t')$. Because $t'$ can be chosen arbitrarily close to $\underline{t}_P$, we have $\apdf(t)=0$ for all $t\in[0,\underline{t}_P)$. Therefore, $\underline{t}_A\geq \underline{t}_P$.

\textit{Step 4}: $\underline{t}_A=\underline{t}_P=\mint>0$. Combining Steps 2 and 3 we have $\underline{t}_A=\underline{t}_P=\mint$, for some $\mint\geq 0$. Using Step 1, we have $\mint>0$.\end{proof}

\begin{claim}\label{asupport}  The support of $\acdf(\cdot)$ is the interval $[\tau_M,\tau_P]$.
\end{claim}
\begin{proof}[Proof of Claim \ref{asupport}]

 It must be shown that $\acdf(t_2)-\acdf(t_1)>0$ for any interval $(t_1,t_2)$ where $\tau_M<t_1< t_2< \tau_P$. Suppose that there exists an interval $(t_1,t_2)$ satisfying the required bounds such that $\acdf(t_1)=\acdf(t_2)$. Furthermore, because $\tau_M<t_2<\tau_P$, the max and min of the support, we have $0<\acdf(t_2)<1$. From Lemma \ref{gapexa}, we have $\acdf(t)=\acdf(t_2)$ for all $t\in[0,t_2)$.  Therefore, $\acdf(0)>0$. By implication $\mint=0$,  which contradicts Step 4.
 \end{proof}

\begin{claim}\label{psupport}  The support of $\pcdf(\cdot)$ is the interval $[\tau_M,\tau_P]$.
\end{claim}
\begin{proof}[Proof of Claim \ref{psupport}]

 It must be shown that $\pcdf(t_2)-\pcdf(t_1)>0$ for any interval $(t_1,t_2)$ where $\tau_M<t_1< t_2< \tau_P$. Suppose that there exists an interval $(t_1,t_2)$ satisfying the required bounds such that $\pcdf(t_1)=\pcdf(t_2)$. Furthermore, because $\tau_M<t_2<\tau_P$, the max and min of the support, we have $0<\pcdf(t_2)<1$. From Lemma \ref{gapexp}, we have $\pcdf(t)=\pcdf(t_2)$ for all $t\in[0,t_2)$.  Therefore, $\pcdf(0)>0$. By implication $\mint=0$,  which contradicts Step 4.
\end{proof}

\begin{claim}\label{a1} The \ppp{} acts on type-1 news at all times, i.e. $a(t)=1$ for all $t\geq 0$.
\end{claim}
\begin{proof}[Proof of Claim \ref{a1}]

First, consider $t<\tau_M$ and suppose that $a(t)<1$ for such $t$. From Lemma \ref{aball}, an interval $(t_1,t_2)$ exists such that $t_1\leq t<t_2$ and $a(t')<1$ for all $t'\in(t_1,t_2)$.  Furthermore,  $a(t')<1$ implies $\bl(t')\leq\psafe$ (consult \ref{PICa}). In turn, $\bl(t')\leq\psafe$ implies $\apdf(t')>0$ (consult \ref{belief}). Hence, for $t'\in(t_1,t_2)$, we have $\apdf(t')>0$. Recall $t_1\leq t<\tau_M$. Therefore, $\acdf(\cdot)$ is strictly increasing on interval $(t_1,\min\{\tau_M,t_2\})$, which contradicts Claim \ref{asupport}. The argument for $t>\tau_P$ is similar.

Next, consider $t\in[\tau_M,\tau_P]$, and suppose $a(t)<1$. From  Lemma \ref{aball}, an interval $(t_1,t_2)\subset[\tau_M,\tau_P]$ exists such that $a(t)<1$ for $t\in(t_1,t_2)$. Furthermore, $a(t)<1\Rightarrow \bl(t)=\psafe$ for $t\in(t_1,t_2)$ (see \ref{PICa}). It follows that for $t_1<t''<t'<t_2$ we have 

\begin{align*}
	u_P(t')-u_P(t'')=\int_{t''}^{t'}\exp(-\rho s)\{w_0^P(s)\psafe+w_1^P(s)\psafe\} ds+\\
	\exp(-\rho t')W^P_\phi(t')\psafe-\exp(-\rho (t'')W^P_\phi(t'')\psafe\leq \\
\exp(-\rho t') \int_{t''}^{t'}\{w_0^P(s)\psafe+w_1^P(s)\psafe\}ds+
	\exp(-\rho t')W^P_\phi(t')\psafe-\exp(-\rho (t'')W^P_\phi(t'')\psafe=\\
	\exp(-\rho t') (W_\phi^P(t'')-W_\phi^P(t'))\psafe+
	\exp(-\rho t')W^P_\phi(t')\psafe-\exp(-\rho t'')W^P_\phi(t'')\psafe=\\
-(\exp(-\rho t'')-\exp(-\rho t'))W_\phi^P(t'')\psafe<0,
\end{align*}
where the last inequality uses $t''<t'$ and $W_\phi^P(\cdot)>0$. It immediately follows for all such $t'$, we must have $\ppdf(t')=0$, contradicting Claim \ref{psupport}.
\end{proof}

%\begin{claim}\label{constant} The \aaa{}'s payoff function $u_A(\cdot)$ is constant on $(\mint,\tau_P)$.
%\end{claim}
%\begin{proof}[Proof of Claim \ref{constant}] \textit{\ppp}
%From Claim \ref{psupport} the support of the \ppp{}'s mixed strategy is $[\tau_M,\tau_P]$ for some $\tau_M>0$. It follows that the set of $t$ such that $\ppdf(t)>0$ is dense in $(\tau_M,\tau_P)$. From Lemma \ref{pcont}, $u_P(\cdot)$ is continuous in equilibrium. From (\ref{PICt}) $u_P(t)=u_P^*$ for all $t$ such that $\ppdf(t)>0$. It follows that $u_P(t)=u_P^*$ for all $t\in[\tau_M,\tau_P]$, where $u_P^*$ is the \ppp{}'s equilibrium payoff. Furthermore, according to Proposition \ref{strucval}, $a(\cdot)=1$.
%%(https://math.stackexchange.com/questions/2984533/if-a-continuous-function-on-the-reals-is-constant-on-a-dense-subset-of-a-topolog)
%	
%\end{proof}

\begin{claim}\label{nomassa} The \aaa{}'s strategy has no mass points.
\end{claim}
\begin{proof}[Proof of Claim \ref{nomassa}]

Suppose the \aaa{}'s strategy has a mass point at $t\in[\mint,\tau_P]$. Because distribution of the arrival time of real news has no mass points, an arrival of type-1 news at $t$ has no effect on the \ppp{}'s about the payoff of the risky arm, i.e. $\bl(t)=\mu$. It follows that $a(t)=0$, which contradicts Claim \ref{a1}.
\end{proof}

\begin{claim}\label{nomassp} The \ppp{}'s stopping strategy has no mass points for $t\in[\mint,\tau_P)$.
\end{claim}
 \begin{proof}[Proof of Claim \ref{nomassp}] 
Consider $\mint\leq t<t'<\tau_P$. From Lemma \ref{abounds}, we have 
\begin{align*}
&	u_A(t')-u_A(t)\leq\\
&	(1-\pcdf(t))\{v(t')-v(t)\}+
\exp(-\rho t')W_\phi^A(t')(\mu-\asafe) a(t')-\exp(-\rho t)W_\phi^A(t)(\mu-\asafe) a(t).
\end{align*}
Because $\mint\leq t<t'<\tau_P$, using Claim \ref{asupport} and condition (\ref{AIC}) we have $u_A(t)=u_A(t')$. From Claim $\ref{a1}$, we have $a(t)=a(t')=1$. Substituting, we have
\begin{align*}
&	(1-\pcdf(t))\{v(t')-v(t)\}+
\exp(-\rho t')W_\phi^A(t')(\mu-\asafe) -\exp(-\rho t)W_\phi^A(t)(\mu-\asafe)\geq 0\Rightarrow\\
&W_\phi^A(t')\geq \exp(\rho (t'-t))W_\phi^A(t)-
	\exp(\rho t')(1-\pcdf(t))\frac{v(t')-v(t)}{\mu-\asafe}\\
&1-\pcdf(t')\geq (1-\pcdf(t))\{\underbrace{\exp(\rho (t'-t))\frac{1-\rcdf(t)}{1-\rcdf(t')}}_{A}-
	\underbrace{\exp(\rho t')\frac{v(t')-v(t)}{\mu-\asafe}}_{B}\}.
\end{align*}
Consider the limit of both sides as $t'\rightarrow t^+$. Because $\pcdf(\cdot)$ is right-continuous, the limit on the left hand side exists. Next, note that as $t'\rightarrow t^+$ the limit of term $A$ is 1, and the limit of term $B$ is 0, which follows from the continuity of function $v(\cdot)$. Therefore, the limit of the right hand side exists, and is equal to $1-\pcdf(t)$. Combining these observations, we have $\lim_{t'\rightarrow t^+}(1-\pcdf(t'))\geq 1-\pcdf(t)$. Simplifying, we have $\lim_{t'\rightarrow t^+}\pcdf(t')\leq \pcdf(t)$. Because  $\pcdf(\cdot)$ is increasing, we also have $\lim_{t'\rightarrow t^+}\pcdf(t')\geq \pcdf(t)$. Therefore $\lim_{t'\rightarrow t^+}\pcdf(t')=\pcdf(t)$. Thus, $\pcdf(\cdot)$ is continuous at $t$. 
 \end{proof}

 \subsection{Proof of Proposition \ref{charval} and Corollary 1}
 
\textbf{Corollary \ref{corcs}} is proved in the Supplemental Appendix.

 \begin{proof}[Proof of Proposition \ref{charval}]
 I first show that the distributions presented in the proposition are the only ones that could constitute an equilibrium with beneficial search. That is, if an equilibrium with beneficial search exists, then it must take the form presented in the statement of the proposition. I then show that these distributions, together with $a(\cdot)=1$, are an equilibrium.
 
\textit{Step 1.} If an equilibrium with beneficial search exists, then the distributions must be the ones from the statement Proposition \ref{charval}.

\textit{\PPP{} Strategy.} From Proposition \ref{strucval}, the support of the \ppp{}'s mixed strategy is $[\tau_M,\tau_P]$ for some $\tau_M>0$. It follows that the set of $t$ such that $\ppdf(t)>0$ is dense in $(\tau_M,\tau_P)$. From Lemma \ref{pcont}, $u_P(\cdot)$ is continuous in equilibrium. From (\ref{PICt}) $u_P(t)=u_P^*$ for all $t$ such that $\ppdf(t)>0$. It follows that $u_P(t)=u_P^*$ for all $t\in[\tau_M,\tau_P]$, where $u_P^*$ is the \ppp{}'s equilibrium payoff. Furthermore, according to Proposition \ref{strucval}, $a(\cdot)=1$.
%(https://math.stackexchange.com/questions/2984533/if-a-continuous-function-on-the-reals-is-constant-on-a-dense-subset-of-a-topolog)
Therefore, for all $t\in[\tau_M,\tau_P]$ we have 
\begin{align*}
u_P(t)-u_P(\tau_M)&=\int^t_{\tau_M}\exp(-\rho s)\{w_0^P(s)\psafe+w_1^P(s)\bl(s)\}+\exp(-\rho t)W_\phi^P(t)\psafe-\exp(-\rho \tau_M)W_\phi^P(\tau_M)\psafe\\
&=\int^t_{\tau_M}\exp(-\rho s)\{((1-\mu)\psafe +\mu)\rpdf(s)(1-\sigma \acdf(s))+\mu\sigma \apdf(s)(1-\rcdf(s))\}ds\\
&+\exp(-\rho t)(1-\rcdf(t))(1-\sigma\acdf(t))\psafe-\exp(-\rho \tau_M)(1-\rcdf(\tau_M))\psafe=0.
\end{align*}
From Proposition \ref{strucval} $\acdf(\cdot)$ has no mass points. It follows that $\acdf(\tau_M)=0$ and  $\acdf(\cdot)$ is absolutely continuous. By implication, for $t\in(\tau_M,\tau_P)$ we have 
\begin{align*}
\exp(-\rho t)(1-\rcdf(t))(1-\sigma\acdf(t))\psafe-\exp(-\rho \tau_M)(1-\rcdf(\tau_M)\psafe=\\
-\int_{\tau_M}^t\exp(-\rho s)\{\rpdf(s)(1-\sigma \acdf(s))+\sigma \apdf(s)(1-\rcdf(s))+\rho(1-\rcdf(s))(1-\sigma \acdf(s)) \}\psafe ds.%	\exp(-\rho t)W_\phi^P(t)\psafe-\exp(-\rho \tau_M)\psafe=-\int_{\tau_M}^t\exp(-\rho s)(w_0^P(s)+w^P_1(s)+\rho W_\phi^P(s))\psafe ds
\end{align*}
To see this, note that the integrand is equal to the derivative of $\exp(-\rho t)(1-\rcdf(t))(1-\sigma\acdf(t))\psafe$ wherever the derivative exists, and the equality above follows from absolute continuity.
Substituting the previous equality and simplifying, 
\begin{align*}
&\int^t_{\tau_M}\exp(-\rho s)\{\mu(1-\psafe)\rpdf(s)(1-\sigma \acdf(s))+(\mu-\theta)\sigma \apdf(s)(1-\rcdf(s))-\rho\theta(1-\sigma\acdf(s))(1-\rcdf(s))\}ds=\\
&\int^t_{\tau_M}\exp(-\rho s)(1-\sigma\acdf(s))(1-\rcdf(s))(\psafe-\mu)\{\frac{\mu(1-\psafe)\rhaz(s)-\rho\theta}{\psafe-\mu}-\ahaz(s)\}ds=0,
\end{align*}
where $\ahaz{}(s)\equiv\sigma\apdf(s)/(1-\sigma\acdf(s))$. Because the first four terms in the integrand are strictly positive, we have 
\begin{align*}
&\int^t_{\tau_M}\{\frac{\mu(1-\psafe)\rhaz(s)-\rho\theta}{\psafe-\mu}-\ahaz(s)\}ds=0.
\end{align*}
Note that $\ln(1-\sigma\acdf(t))$ is absolutely continuous on $[\tau_M,\tau_P]$ and its derivative is $-\ahaz(t)$ (wherever $\acdf(\cdot)$ is differentiable). By implication
\begin{align*}
\ln(1-\sigma\acdf(t))-\ln(1-\sigma\acdf(\tau_M))=-\int^t_{\tau_M}\frac{\mu(1-\psafe)\rhaz(s)-\rho\theta}{\psafe-\mu}ds.
\end{align*}
Recall that $\acdf(\cdot)$ has no mass points, and the bottom of the support is $\tau_M$. It follows that $\acdf(\tau_M)=0$, and hence
\begin{align*}
	\acdf(t)=\frac{1}{\sigma}(1-\exp\{-\int^t_{\tau_M}\frac{\mu(1-\psafe)\rhaz(s)-\rho\theta}{\psafe-\mu}ds\}).
\end{align*}

\textit{Determining $\tau_M$}. From Proposition \ref{strucval}, in an equilibrium with valuable search (i) the \aaa{}'s strategy has no atoms and (ii), the the top of the support is $\tau_P$. It follows that
\begin{align*}
\frac{1}{\sigma}(1-&\exp\{-\int^{\tau_P}_{\tau_M}\frac{\mu(1-\psafe)\rhaz(s)-\rho\theta}{\psafe-\mu}ds\})=1\Leftrightarrow
1-&\exp\{-\int^{\tau_P}_{\tau_M}\frac{\mu(1-\psafe)\rhaz(s)-\rho\theta}{\psafe-\mu}ds\}=\sigma.
\end{align*}
Next, note that the LHS is a strictly decreasing function of $\mint$ and its value is 0 at $\mint=\tau_P$. Therefore, the previous equation is satisfied by at most one $\mint$. Furthermore, from Proposition \ref{strucval}, in a beneficial search equilibrium $\mint>0$. Such a solution exists if and only if
\begin{equation*}
\exp\{-\int^{\tau_P}_{0}\frac{\mu(1-\psafe)\rhaz(s)-\rho\theta}{\mu-\psafe}ds\}<1-\sigma\Leftrightarrow\sigma<\bar{\sigma}.
\end{equation*}

\textit{Agent Strategy.} As a preliminary step, I show that $u_A(\cdot)$ is continuous on $[\mint,\tau_P)$. Consider $t''\leq t'<\tau_P$. From Proposition \ref{strucval}, we have that $a(\cdot)=1$. Therefore, for such $t'',t'$ we have
\begin{align*}
	u_A(t')-u_A(t'')=\int_{t''}^{t'}\exp(-\rho s)\{(w_0^A(s)+w_S^A(s))\asafe+w_1^A(s)\}ds\\+\exp(-\rho t')W^A_\phi (t')\mu-\exp(-\rho t)W^A_\phi (t)\mu.
\end{align*}
Furthermore, according to Proposition \ref{strucval}, $\pcdf(\cdot)$ has no atoms in $[\tau_M,\tau_P)$. Therefore, $w_0^A(\cdot)$,$w_1^A(\cdot)$, $w_S^A(\cdot)$, $W_\phi^A(\cdot)$ are all continuous. Therefore $u_A(\cdot)$ is continuous on $[\mint,\tau_P)$.

 From Proposition, the support of the \aaa{}'s mixed strategy is $[\tau_M,\tau_P]$; that is, the set of $t$ such that $\apdf(t)>0$ is dense in $(\tau_M,\tau_P)$. As shown above $u_A(\cdot)$ is continuous on $[\tau_M,\tau_P)$. From (\ref{PICt}) $u_A(t)=u_A^*$ for all $t$ such that $\apdf(t)>0$. It follows that $u_A(t)=u_A^*$ for all $t\in[\tau_M,\tau_P)$, where $u_A^*$ is the \aaa{}'s equilibrium payoff.
%(https://math.stackexchange.com/questions/2984533/if-a-continuous-function-on-the-reals-is-constant-on-a-dense-subset-of-a-topolog)
Furthermore, Proposition \ref{strucval}, $a(t)=1$ for all $t$. By implication, for $\tau_M<t<\tau_P$ we have 
\begin{align*}
	u_A(t)&-u_A(\tau_M)=\\
&	\int_{t}^{\tau_M}\exp(-\rho s)(w_0^A(s)+w_S^A(s))\asafe+w_1^A(s)ds+\exp(-\rho t)W^A_\phi (t)\mu-\exp(-\rho \tau_M)W_\phi^A(\tau_M)\mu\\
&	\int_{t}^{\tau_M}\exp(-\rho s)\{(\mu+(1-\mu)\asafe)\rpdf(s)(1-\pcdf(s))+\ppdf(s)(1-\rcdf(s))\asafe\}ds+\\
&\exp(-\rho t)(1-\rcdf(t))(1-\pcdf(t))\mu-\exp(-\rho \tau_M)(1-\rcdf(\tau_M))\mu=0.
\end{align*}
Because $\pcdf(\cdot)$ has no mass points in $[\tau_M,\tau_P)$, (1) $\pcdf(\tau_M)=0$, and (2) $\pcdf(\cdot)$ is absolutely continuous on this interval. By implication, for $t\in[\tau_M,\tau_P)$ we have 
\begin{align*}
	&\exp(-\rho t)(1-\rcdf(t))(1-\pcdf(t))\mu-\exp(-\rho \tau_M)(1-\rcdf(\tau_M))\mu=\\
	&-\int^{\tau_M}_t\exp(-\rho s)\{\rpdf(s)(1-\pcdf(s))+\ppdf(s)(1-\rcdf(s))+\rho(1-\rcdf(s))(1-\pcdf(s))\}\mu ds.
\end{align*}
To see this, note that the integrand is equal to the derivative of $\exp(-\rho t)(1-\rcdf(t))(1-\pcdf(t))\mu$ wherever the derivative exists, and the equality above follows from absolute continuity.

Combining and simplifying, we have 
\begin{align*}
&	\int^{t}_{\tau_M}\exp(-\rho s)\{\asafe(1-\mu)\rpdf(s)(1-\pcdf(s))-\ppdf(s)(1-\rcdf(s))(\mu-\asafe)-\rho \mu (1-\rcdf(s))(1-\pcdf(s))\}ds=\\
&	\int^{t}_{\tau_M}\exp(-\rho s)(1-\rcdf(s))(1-\pcdf(s))(\mu-\beta)\{\frac{\asafe(1-\mu)\rhaz(s)-\rho\mu}{\mu-\beta}-\phaz(s)\}ds=0.
\end{align*}
Where $\phaz{}(s)\equiv\apdf(s)/(1-\acdf(s))$. Because the first four terms in the integrand are strictly positive, we have 
\begin{align*}
&\int^t_{\tau_M}\frac{\asafe(1-\mu)\rhaz(s)-\rho\mu}{\mu-\beta}-\phaz(s)ds=0.
\end{align*}
Note that $\ln(1-\pcdf(t))$ is absolutely continuous on $[\tau_M,\tau_P)$ and its derivative is $-\phaz(t)$ (wherever $\acdf(\cdot)$ is differentiable). By implication
\begin{align*}
\ln(1-\pcdf(t))-\ln(1-\pcdf(\tau_M))=-\int^t_{\tau_M}\frac{\asafe(1-\mu)\rhaz(s)-\rho\mu}{\mu-\beta}ds.
\end{align*}
Recall that $\pcdf(\cdot)$ has no mass points on $[\mint,\tau_P)$, and the bottom of the support is $\tau_M$. It follows that $\pcdf(\tau_M)=0$, and hence
\begin{align*}
	\pcdf(t)=1-\exp\{-\int^t_{\tau_M}\frac{\asafe(1-\mu)\rhaz(s)-\rho\mu}{\mu-\asafe}ds\},
\end{align*}
for $t\in[\tau_M,\tau_P)$. Proposition \ref{strucval} requires that the support of the \ppp{}'s mixed strategy is $[\tau_M,\tau_P]$, and a mass point on $\tau_P$ is allowed. In particular, because $\pcdf(\cdot)$ has a mass point at $\tau_P$, \aaa{}'s payoff $u_A(\cdot)$ has a downward jump at $\tau_P$. It it is straightforward to verify that $\apdf(\tau_P)=0$, and therefore the \aaa{}'s strategy does not violate (\ref{AIC}).

\textit{Sufficiency.} The preceding calculations establish indifference for the \ppp{} in $[\mint,\tau_P]$ and for the \aaa{} in $[\mint, \tau_P)$. To establish sufficiency, it remains to show (1) the proposed distributions are valid CDFs, (2) for $t$ at which $f_i(t)=0$, each player's payoff does not exceed $u_i^*$, and (3) under the proposed distribution for $\acdf(\cdot)$, we have $a(\cdot)=1$ in equilibrium. (1) follows from elementary calculations. (2) is immediate. (3) is proved as part of Corollary \ref{corbel}, below.
\end{proof}

\begin{proof}[Proof of Corollary \ref{corbel}] Any type-1 arrival at time $t\notin[\mint,\tau_P]$ reveals $\omega=1$, because the likelihood of a fake is 0. Consider times $t\in[\mint,\tau_P]$. At such a time, the \ppp{}'s posterior belief is given by (\ref{belief}),
\begin{align*}
	\bl(t)=\frac{\mu\rpdf(t)(1-\apri\acdf(t)+\mu\apri\apdf(t)(1-\rcdf(t))}{\mu\rpdf(t)(1-\apri\acdf(t))+\apri\apdf(t)(1-\rcdf(t))}=\frac{\mu[\rhaz(t)+\ahaz(t)]}{\mu\rhaz(t)+\ahaz(t)},
\end{align*}
where $\ahaz(t)\equiv \sigma\apdf(t)/(1-\sigma\acdf(t))$. Using Proposition \ref{charval}, it is straightforward to show that in an equilibrium with beneficial search,
\begin{align*}
	\ahaz(t)=\frac{\mu(1-\psafe)\rhaz(s)-\rho\psafe}{\psafe-\mu}.
\end{align*}
Routine simplification yields the expression in the proposition. Direct substitution establishes continuity at $\tau_P$. That the belief is increasing follows from the assumption that $\rhaz(\cdot)$ is decreasing. Finally, it must be demonstrated that an $\epsilon>0$ exists such that $\mu_1(t)\geq \psafe+\epsilon$ for $t\in[\mint,\tau_P]$.  Note that $\rhaz(t)\geq \rho\theta/(\mu(1-\theta))$ for $t\in[\mint,\tau_P]$, and hence, 
\begin{align*}
\rhaz(t)\mu(1-\mu)-\rho\psafe\geq \rho\psafe(\frac{1-\mu}{1-\psafe}-1)>0.
\end{align*}
Let
\begin{align*}
\epsilon\equiv (\psafe-\mu)\frac{\rho\psafe}{\rhaz(\tau_M)\mu(1-\mu)-\rho\psafe}.
\end{align*}
As shown above, $\epsilon>0$. Furthermore, for $t\in[\tau_M,\tau_P]$, we have $\mu_1(t)>\mu_1(\mint)=\psafe+\epsilon$.
\end{proof}

 \subsection{Proof of Proposition \ref{charnot}}

 As argued in the text, an equilibrium in which the \ppp{} does not benefit from search exists if and only if there exists an \aaa{} faking strategy $\acdf(\cdot)$, such that $u_P(t)\leq\psafe$ for all $t\geq 0$, with $a(\cdot)$ chosen optimally (as in (\ref{PICa})). That is such an equilibrium exists if and only if there exists $\acdf(\cdot)$, such that
 \begin{align}\label{nocon}
u_P(t)=\int_0^t &\exp(-\rho s)\{w^P_0(s)\psafe+w^P_1(s)((1-a(s))\theta+a(s)\bl(s))\}ds
\\\nonumber	&+\exp(-\rho t)W^P_\phi(t))\psafe\leq \psafe, 	
 \end{align}
  for all $t\geq 0$, where $a(t)=\mathcal{I}(\mu_1(t)\geq \psafe)$. 
  
  First, note that attention can be restricted to $t\leq\tau_P$.
  
 \textit{Step 1. An equilibrium without beneficial search exists if and only if there exists a faking strategy $\acdf(\cdot)$ such that $u_P(t)\leq\psafe$ for $t\in[0,\tau_P]$, with $a(\cdot)$ chosen optimally.} The necessary and sufficient condition (\ref{nocon}) must hold for all $t\geq 0$, and therefore must hold for $t\in[0,\tau_P]$. The proposed condition is therefore necessary. To see that it is sufficient, Recall Lemma \ref{pbounds} (ii), which establishes that if (\ref{PICa}), then $u(t)<u(\tau_P)$ for $t>\tau_P$. Therefore, if $u(\tau_P)\leq\psafe$, then for $t>\tau_P$ we have $u(t)<u(\tau_P)\leq\psafe$.

 \textit{Part (i). Show that if an equilibrium without beneficial search exists, then $\sigma\geq \bar{\sigma}$.}
  
Consider a constrained version of the \ppp{}'s objective function, where $a(\cdot)$ is chosen suboptimally. In particular, the \ppp{} is constrained to $a(t)=1$ unless she is certain that the news is fake, as would occur at an atom of the faking strategy. That is, the \ppp{} is constrained to select $a(t)=\mathcal{I}(\mu_1(t)>0)$.  
% \begin{align*}
%\underline{u}_P(t)\equiv\int_0^t &\exp(-\rho s)\{w^P_0(s)\psafe+w^P_1(s)\bl(s)\}ds
%	+\exp(-\rho t)W^P_\phi(t))\psafe\leq \psafe. 	
% \end{align*}
Let $\underline{u}_P(t)$ be the \ppp{}'s payoff in this constrained benchmark. Obviously, $\underline{u}_P(t)\leq u_P(t)$. 
 
 \begin{claim}\label{pseudo} If an equilibrium without beneficial search exists, then $\underline{u}_P(t)\leq\psafe$ for all $t\in[0,\tau_P]$.
 \end{claim}
 \begin{proof}[Proof of Claim \ref{pseudo}] If an equilibrium without beneficial search exists, then for all $t\leq \tau_P$, we have $u_P(t)\leq\theta$, and therefore, $\underline{u}_P(t)\leq u_P(t)\leq\theta$ for all such $t$.
 \end{proof}
 
 Next, I show that attention can be confined to distributions $\acdf(\cdot)$ without atoms.
 \begin{claim}\label{noatoms} If a distribution $\acdf(\cdot)$ exists such that $\underline{u}_P(t)\leq\psafe$ for all $t\in [0,\tau_P]$, then a continuous distribution $\acdf^*(\cdot)$ exists such that the same inequality holds for $t\in[0,\tau_P]$.
 \end{claim}
 \begin{proof}[Proof of Claim \ref{noatoms}] Suppose a distribution $\acdf(\cdot)$ exists such that $\underline{u}_P(t)\leq\psafe$ for all $t\in[0,\tau_P]$. Let $m=1,..,M$ where $M\leq \infty$ index the set of atoms of the \aaa{}'s strategy $\acdf(\cdot)$, with $t_m$ representing the time of atom $m$ and $f_m$ representing the associated probability. Recall that $\acdf(\cdot)$ is absolutely continuous except at atoms. We can therefore represent the \aaa{}'s faking strategy as a compound lottery. In the first stage, the \aaa{} selects atom $t_m$ with probability $f_m$, and continues to the second stage with probability $1-\sum_{m=1}^Mf_m$. If the second stage is reached, the \aaa{} selects his faking time by drawing from some absolutely continuous distribution $\acdf^*(\cdot)$. Denoting by $T_A$ the \aaa{}'s faking time, we have 
 \begin{align*}
 	T_A=\begin{cases}
 		t_1\text{ with probability }f_1\\
 		t_2\text{ with probability }f_2\\
 		 		...\\
 		t_M\text{ with probability }f_M\\
 		T_A^*\text{ with probability }1-\sum_{m=1}^Mf_m
 	\end{cases}
 \end{align*}
 where $T_A^*\sim F_A^*(\cdot)$. 
 
 Consider the \ppp{}'s expected payoff in the constrained benchmark, $\underline{u}_P(t)$, for $t\leq\tau_P$. I first consider this payoff conditional on the \aaa{}'s type and the realization of the first stage lottery. Then I apply the Law of Iterated expectations to calculate the payoff. 
 
 Suppose that the \aaa{} is manipulative and the first stage realization is atom $t_m$. By implication, all arrivals at $s\leq t_m$ are real, which implies that all arrivals at $s\leq \min\{t,t_m\}$ are real. Furthermore, for almost all such arrivals, we have $a(s)=1$ ($a(s)=0$ only at atoms, which are countable). Moreover, if time $\min\{t,t_m\}$ is reached without a real arrival, then the \aaa{} fakes at atom $t_m$, the \ppp{} infers that this arrival is fake $\mu_1(t_m)=0$, and selects safe. Thus, conditional on the \faker{} \aaa{} and realized faking time $t_m$, we have
 \begin{align*}
& 	\underline{u}_p(t|\text{\faker},t_m)=\\&\int_0^{\min\{t,t_m\}}\exp(-\rho s)\rpdf(s)(\mu+(1-\mu)\psafe)ds+\exp(-\rho\min\{t,t_m\})(1-\rcdf(\min\{t,t_m\}))\psafe\\
&=u_P^{FB}(\min\{t,t_m\}).
 \end{align*}

Next, suppose that the \aaa{} is \faker{} and the realization of the first stage lottery is $T_A^*$. In this case, the \aaa{}'s faking time is distributed according to absolutely continuous distribution $\acdf^*(\cdot)$, and the \ppp{} selects $a(\cdot)=1$ for almost all $s\leq t$. Thus,  \ppp{}'s conditional expected payoff is the same as waiting until time $t$ in the constrained benchmark, when the \aaa{} plays strategy $\acdf^*(\cdot)$ and is \faker{} with probability 1. In particular, 
\begin{align*}
&	\underline{u}_P(t|\text{\faker},T_A^*)=\\ &\int_0^t\exp(-\rho s)\{((1-\mu)\psafe+\mu)\rpdf(s)(1-\acdf^*(s))+\mu\apdf^*(s)(1-\rcdf(s))\}ds+\exp(-\rho t)(1-\rcdf(t))(1-\acdf^*(t))\psafe.
\end{align*}
%where 
%\begin{align*}
%	w_0^*(s)&=(1-\mu)\rpdf(s)(1-\acdf^*(s))\\
%	 w_1^*(s)\mu_1^*(s)&=\mu\rpdf(s)(1-\acdf^*(s))+\mu\apdf^*(s)(1-\rcdf(s))\\
%W^P_\phi(t)&=(1-\rcdf(t))(1-\acdf^*(t)).
%\end{align*}

Finally, suppose that the \aaa{} is not the \faker{} type. In this case, the conditional payoff of waiting until time $t$ is the same as in the first best benchmark, since all arrivals $s\leq t$ are real and $a(s)=1$ for almost all of them. That is $\underline{u}_P(t|\text{not \faker{}})=u_P^{FB}(t)$.

Applying the Law of Iterated Expectations, we have
\begin{align*}
&	\underline{u}_P(t)=\\&(1-\sigma)\underline{u}_P(t|\text{not \faker{}})+\sigma\sum_{m=1}^M f_m \underline{u}_p(t|\text{\faker},t_m)
	+\sigma(1-\sum_{m=1}^M f_m)\underline{u}_P(t|\text{\faker},T_A^*)\\
	&=(1-\sigma)u_P^{FB}(t)+\sigma\sum_{m=1}^M f_m u^{FB}_P(\min\{t,t_m\})+\sigma(1-\sum_{m=1}^M f_m)\underline{u}_P(t|\text{\faker},T_A^*)\\
		&\sum_{m=1}^M f_m\{(1-\sigma)u_P^{FB}(t)+\sigma u^{FB}_P(\min\{t,t_m\})\}+(1-\sum_{m=1}^M f_m)\{(1-\sigma)u_P^{FB}(t)+\sigma\underline{u}_P(t|\text{\faker},T_A^*)\}.
\end{align*}
Consider the last term. We have,
\begin{align*}
&(1-\sigma)u_P^{FB}(t)+\sigma\underline{u}_P(t|\text{\faker},T_A^*)=\\
&(1-\sigma)\Big(\int_0^t\exp(-\rho s)\{(\mu+(1-\mu)\psafe)\rpdf(s)\}ds+\exp(-\rho t)(1-\rcdf(t))\psafe\Big)+\\
&\sigma\Big(\int_0^t\exp(-\rho s)\{((1-\mu)\psafe+\mu)\rpdf(s)(1-\acdf^*(s))+\mu\apdf^*(s)(1-\rcdf(s))\}ds+\exp(-\rho t)(1-\rcdf(s))(1-\acdf^*(s))\psafe\Big)=\\
&\int_0^t\exp(-\rho s)\{((1-\mu)\psafe+\mu)\rpdf(s)(1-\sigma\acdf^*(s))+\mu\sigma\apdf^*(s)(1-\rcdf(s))\}ds+\exp(-\rho t)(1-\rcdf(t))(1-\sigma\acdf^*(t))\psafe.
\end{align*}
Note that the last line is simply the \ppp{}'s expected payoff in the constrained benchmark if she selects stopping time $t$ and the \aaa{}'s strategy is $\acdf^*(\cdot)$, which is absolutely continuous. Let this payoff be denoted $\underline{u}^*_P(t)$. Substituting into the Law of Iterated Expectations, we have 
\begin{align*}
	\underline{u}_P(t)=
		&\sum_{m=1}^M f_m\{(1-\sigma)u_P^{FB}(t)+\sigma u^{FB}_P(\min\{t,t_m\})\}+(1-\sum_{m=1}^M f_m)\underline{u}^*_P(t).
\end{align*}
Next, note that Lemma \ref{fb} and A1 together imply that, in the \ppp{}'s first-best problem, her optimal search time is $\tau_P>0$, and her payoff is single-peaked. It follows that for $t\in(0,\tau_P)$ we have $u_P^{FB}(t)>u_P^{FB}(0)=\psafe$. Using this observation to bound the first summation, we have
\begin{align*}
\underline{u}_P(t)>\sum_{m=1}^M f_m\psafe +(1-\sum_{m=1}^M f_m)\underline{u}^*_P(t).
\end{align*}
Because $\underline{u}_P(t)\leq\psafe$ for all $t\in[0,\tau_P]$, we have 
\begin{align*}
\psafe>\sum_{m=1}^M f_m\psafe +(1-\sum_{m=1}^M f_m)\underline{u}^*_P(t).
\end{align*}
Note first that this inequality immediately implies $\sum_{m=1}^M f_m<1$. Simplifying, we have
\begin{align*}
(1-\sum_{m=1}^M f_m)\underline{u}^*_P(t)<(1-\sum_{m=1}^M f_m)\psafe\Rightarrow \underline{u}^*_P(t)<\psafe.
\end{align*}
 Recalling that $u_P^*(t)$ is the \ppp{}'s expected payoff in the constrained benchmark when the \aaa{} uses absolutely continuous strategy $\acdf^*(\cdot)$, the claim is proved.\end{proof}
Combining Claims \ref{pseudo} and \ref{noatoms}, we have

\begin{claim}\label{con} If an equilibrium without beneficial search exists, then an absolutely continuous distribution $\acdf(\cdot)$ exists such that $\underline{u}_P(t)\leq \psafe$ for all $t\in[0,\tau_P]$.
\end{claim}

To complete the proof of Part (i), we show  the following result.

\begin{claim}\label{con} If an absolutely continuous distribution $\acdf(\cdot)$ exists such that $\underline{u}_P(t)\leq \psafe$ for all $t\in[0,\tau_P]$, then $\sigma\leq\bar{\sigma}$.
\end{claim}
\begin{proof}[Proof of Claim \ref{con}].Suppose an absolutely continuous distribution $\acdf(\cdot)$ exists, such that $\underline{u}_P(t)\leq \psafe$ for $t\in[0,\tau_P]$. By direct substitution, we have 
\begin{align*}
\underline{u}_P(t)=\int^t_{0}\exp(-\rho s)\{((1-\mu)\psafe +\mu)\rpdf(s)(1-\sigma \acdf(s))+\mu\sigma \apdf(s)(1-\rcdf(s))\}ds\\
+\exp(-\rho t)(1-\rcdf(t))(1-\sigma\acdf(t))\psafe\leq \psafe.
\end{align*}
Because $\apdf(\cdot)$ is integrable and $\rcdf(\cdot)$ is differentiable, we can
integrate by parts to eliminate $\sigma f(\cdot)$ from the integrand. In particular,
%
%\begin{align*}
%&\int^t_{0}\exp(-\rho s)\sigma \apdf(s)(1-\rcdf(s))ds=\int^t_{0}\exp(-\rho s)(1-\rcdf(s))\}d[\sigma \acdf(s)]=\\
%&\sigma \acdf(s)\exp(-\rho s)(1-\rcdf(s))]_{0}^t+\int_0^t\exp(-\rho s)\sigma \acdf(s)(\rho (1-\rcdf(s))+\rpdf(s))ds=\\
%&\sigma \acdf(t)\exp(-\rho t)(1-\rcdf(t))+\int_0^t\exp(-\rho s)\sigma \acdf(s)(\rho (1-\rcdf(s))+\rpdf(s))ds.
%\end{align*}
%
%Substituting into $\underline{u}_P(t)$, we have 
%\begin{align*}
%\underline{u}_P(t)=\int^t_{0}\exp(-\rho s)((1-\mu)\psafe +\mu)\rpdf(s)(1-\sigma \acdf(s))ds
%+\exp(-\rho t)(1-\rcdf(t))(1-\sigma\acdf(t))\psafe+\\
%\mu\Big(\sigma \acdf(t)\exp(-\rho t)(1-\rcdf(t))+\int_0^t\exp(-\rho s)\sigma \acdf(s)(\rho (1-\rcdf(s))+\rpdf(s))ds\Big).
%\end{align*}
%Simplifying, we have 
%\begin{align*}
%\underline{u}_P(t)=\int^t_{0}\exp(-\rho s)(1-\mu)\psafe \rpdf(s)(1-\sigma \acdf(s))ds
%+\exp(-\rho t)(1-\rcdf(t))(1-\sigma\acdf(t))\psafe+\\
%\mu\Big(-\sigma \acdf(t)\exp(-\rho t)(1-\rcdf(t))+\int_0^t\exp(-\rho s)\sigma \acdf(s)(\rho (1-\rcdf(s))+\rpdf(s))ds.\Big).
%\end{align*}

\begin{align*}
&\int^t_{0}\exp(-\rho s)\sigma \apdf(s)(1-\rcdf(s))ds=-\int^t_{0}\exp(-\rho s)(1-\rcdf(s))\}d(1-\sigma \acdf(s)))=\\
&-(1-\sigma \acdf(s)))\exp(-\rho s)(1-\rcdf(s))]_{0}^t-\int_0^t\exp(-\rho s)(1-\sigma \acdf(s))(\rho (1-\rcdf(s))+\rpdf(s))ds=\\
&1-(1-\sigma \acdf(t)))\exp(-\rho t)(1-\rcdf(t))]-\int_0^t\exp(-\rho s)(1-\sigma \acdf(s))(\rho (1-\rcdf(s))+\rpdf(s))ds.
\end{align*}
Substituting into $\underline{u}_P(t)$, we have 
\begin{align*}
\underline{u}_P(t)=\int^t_{0}\exp(-\rho s)((1-\mu)\psafe +\mu)\rpdf(s)(1-\sigma \acdf(s))ds
+\exp(-\rho t)(1-\rcdf(t))(1-\sigma\acdf(t))\psafe+\\
\mu\Big(1-(1-\sigma \acdf(t)))\exp(-\rho t)(1-\rcdf(t))-\int_0^t\exp(-\rho s)(1-\sigma \acdf(s))(\rho (1-\rcdf(s))+\rpdf(s))ds\Big).
\end{align*}
Simplifying, we have 
\begin{align*}
&\underline{u}_P(t)=\\&\int^t_{0}\exp(-\rho s)(1-\sigma \acdf(s))\{(1-\mu)\psafe\rpdf(s)-\rho\mu(1-\rcdf(s))\}ds
+\exp(-\rho t)(1-\rcdf(t))(1-\sigma\acdf(t))(\psafe-\mu)+\mu=\\
&\int^t_{0}\exp(-\rho s)(1-\rcdf(s))(1-\sigma \acdf(s))\{(1-\mu)\psafe\rhaz(s)-\rho\mu\}ds
+\exp(-\rho t)(1-\rcdf(t))(1-\sigma\acdf(t))(\psafe-\mu)+\mu.
\end{align*}

Let $w(t)\equiv \exp(-\rho s)(1-\rcdf(s))(1-\sigma \acdf(s))$. We have 
\begin{align*}
&\underline{u}_P(t)=
\int^t_{0}w(s)\{(1-\mu)\psafe\rhaz(s)-\rho\mu\}ds
+w(t)(\psafe-\mu)+\mu\leq\theta\Rightarrow\\
&w(t)\leq 1-\int_0^t w(s)\frac{(1-\mu)\psafe\rhaz(s)-\rho\mu}{\psafe-\mu}ds\Rightarrow\\
&w(t)\leq 1-\int_0^t b(s)w(s)ds,
\end{align*}
where
\begin{align*}
	b(s)=\frac{(1-\mu)\psafe\rhaz(s)-\rho\mu}{\psafe-\mu}.
\end{align*}
Note that for $t\in[0,\tau_P]$ we have $H_R(s)\geq \phi_P$, 
\begin{align*}
	b(s)\geq & \frac{(1-\mu)\psafe\phi_P-\rho\mu}{\psafe-\mu}=\frac{(1-\mu)\psafe}{\psafe-\mu}(\phi_P-\frac{\rho\mu}{\psafe(1-\mu)})=\frac{\rho(1-\mu)\psafe}{\psafe-\mu}(\frac{\psafe}{\mu(1-\psafe)}-\frac{\mu}{\psafe(1-\mu)})\\
	=&\frac{(\psafe-\mu)(\psafe+\mu(1-\psafe))}{\psafe(1-\psafe)\mu(1-\mu)}>0.
	\end{align*}
%Note that 
%\begin{align*}
%&\frac{\psafe}{\mu(1-\psafe)}-\frac{\mu}{\psafe(1-\mu)}=\frac{\psafe^2(1-\mu)-\mu^2(1-\psafe)}{\psafe(1-\psafe)\mu(1-\mu)})=\frac{\psafe^2-\mu^2-\mu\psafe(\psafe-\mu)}{\psafe(1-\psafe)\mu(1-\mu)}=\\
%&\frac{(\psafe-\mu)(\psafe+\mu(1-\psafe))}{\psafe(1-\psafe)\mu(1-\mu)}>0.
%\end{align*}

Next, consider the following result, presented in  \citet{BS1992} (Corollary 1.12).
\begin{claim}\label{Goll}(G\"ollwitzer (1969)) Let $y(t)$ and $b(t)$ be continuous functions in $J=[\alpha,\beta]$, let $q(t)$ be a nonnegative function in $J$, and suppose $b(t)$ is also non-negative in $J$. If a function $u(t)$ in $J$ is such that
\begin{align*}
u(t)\leq y(r)-q(t)\int_r^t b(s)y(s)ds
\end{align*}
then 
\begin{align*}
u(t)\leq y(r)\exp(-q(t)\int_r^t b(s)ds)	
\end{align*}
where $\alpha\leq r<t\leq \beta$.
\end{claim}
To apply this result, set $u(\cdot)=w(\cdot)$, $y(\cdot)=w(\cdot)$, $r=\alpha=0$, $\beta=\tau_P$, $q(t)=1$. Note that $y(0)=w(0)=1$. Furthermore, $b(\cdot)>0$ on interval $[0,\tau_P]$ as shown above.  Furthermore, because $\acdf(\cdot)$ is continuous, so is $w(\cdot)=y(\cdot)$. Therefore,
\begin{align*}
&w(t)\leq 1-\int_0^t b(s)w(s)ds\Rightarrow w(t)\leq \exp(-\int_{0}^tb(s)ds\}.
\end{align*}
Next, substitute $w(\cdot)$ and $b(\cdot)$.
\begin{align*}
&\exp(-\rho t)(1-\rcdf(t))(1-\sigma\acdf(t))\leq \exp(-\int_{0}^t\frac{(1-\mu)\psafe\rhaz(s)-\rho\mu}{\theta-\mu}ds).
\end{align*}
Next, substitute $1-\rcdf(t)=\exp(-\int_0^t\rhaz(s)ds)$ and simplify,
\begin{align*}
&\exp(-\int_0^t\{\rhaz(s)+\rho\} ds)(1-\sigma\acdf(t))\leq \exp(-\int_{0}^t\frac{(1-\mu)\psafe\rhaz(s)-\rho\mu}{\theta-\mu}ds)\Rightarrow\\
&1-\sigma\acdf(t)\leq \exp(-\int_{0}^t\{\frac{(1-\mu)\psafe\rhaz(s)-\rho\mu}{\theta-\mu}-\rhaz(s)-\rho\} ds)\Rightarrow\\
&1-\sigma\acdf(t)\leq \exp(-\int_{0}^t\{\frac{\mu(1-\psafe)\rhaz(s)-\rho\psafe}{\theta-\mu}\} ds).
\end{align*}
From Claim \ref{Goll}, the preceding holds for $t\in[0,\tau_P]$. Consider $t=\tau_P$.
\begin{align*}
&1-\sigma\acdf(\tau_P)\leq \exp(-\int_{0}^{\tau_P}\{\frac{\mu(1-\psafe)\rhaz(s)-\rho\psafe}{\theta-\mu}\} ds\Rightarrow\\
&1-\sigma\acdf(\tau_P)\leq 1-\bar{\sigma}\Rightarrow\\
&\bar{\sigma}\leq \sigma \acdf(\tau_P)\leq \sigma.
\end{align*}
This completes the proof of Part (i).

Part (ii) It is readily verified that the proposed CDF is increasing and that for $\sigma>\bar{\sigma}$, the proposed CDF attains 1 at some $\tau_X>\tau_P$. Furthermore, it is straightforward to show that $u_P(t)=\psafe$, for all $t\leq\tau_X$. Thus, the proposed \aaa{} strategy, coupled with immediate stopping constitute an equilibrium without beneficial search. \end{proof}

\subsection{Proofs for Section \ref{rem}}
\begin{proof}[Proof of Proposition \ref{naive}] 

\textit{Step 1. Commitment to naive search delivers \ppp{} a higher payoff than the unique equilibrium if $\sigma<\bar{\sigma}$.} If $\sigma<\bar{\sigma}$, then the unique equilibrium without commitment has beneficial search. Consider the payoff difference as a function of $\sigma$,
\begin{align*}
\Delta(\sigma)=u_P^N(\sigma)-u^{FB}_P(\mint(\sigma)).
\end{align*}
Note that for $\sigma=0$, we have $u_P^N(0)=u_P^{FB}(\mint(0))=u_P^{FB}(\tau_P)$. Thus, $\Delta(0)=0$. Furthermore,
\begin{align*}
	\frac{du_P^{N}(\sigma)}{d\sigma}=-\exp(-\rho \tau_P)(1-\rcdf(\tau_P))(\psafe-\mu).
\end{align*}
Moreover, 
\begin{align*}
	\frac{du_P^{FB}(\mint(\sigma))}{d\sigma}=\frac{du_P^{FB}(\mint)}{d\mint}\frac{d\mint}{d\sigma}.
\end{align*}
From proof of Lemma \ref{fb}, we have 
\begin{align*}
\frac{du_P^{FB}(\mint)}{d\mint}
&=\exp(-\rho \mint)(1-\rcdf(\mint))\mu(1-\psafe)\{\rhaz(\mint)-\frac{\rho\psafe}{\mu(1-\psafe)}\}.
\end{align*}
From Proposition \ref{charval}, we have the following implicit function relating $\mint$ and $\sigma$.
\begin{align*}
\exp\{-\int_{\mint}^{\tau_P}\frac{\mu(1-\psafe)\rhaz(s)-\rho\psafe}{\psafe-\mu}ds\}=1-\sigma.
\end{align*}
Differentiating, with respect to $\sigma$, we have 
\begin{align*}
\frac{d\mint}{d\sigma}\{\frac{\mu(1-\psafe)\rhaz(\mint)-\rho\psafe}{\psafe-\mu}\}\exp\{-\int_{\mint}^{\tau_P}\frac{\mu(1-\psafe)\rhaz(s)-\rho\psafe}{\psafe-\mu}\}=-1.
\end{align*}
Simplifying, we have 
\begin{align*}
\frac{d\mint}{d\sigma}\{\frac{\mu(1-\psafe)\rhaz(\mint)-\rho\psafe}{\psafe-\mu}\}(1-\sigma)=-1\Rightarrow\\
\frac{d\mint}{d\sigma}=\{\frac{\mu(1-\psafe)\rhaz(\mint)-\rho\psafe}{\psafe-\mu}\}(1-\sigma)=-\frac{\psafe-\mu}{1-\sigma}\frac{1}{\mu(1-\psafe)(\rhaz(\mint)-\frac{\rho\psafe}{\mu(1-\psafe)})}.
\end{align*}
Combining these expressions, 
\begin{align*}
&	\frac{du_P^{FB}(\mint(\sigma))}{d\sigma}=\\
&	\Big(\exp(-\rho \mint)(1-\rcdf(\mint))\mu(1-\psafe)\{\rhaz(\mint)-\frac{\rho\psafe}{\mu(1-\psafe)}\}\Big)\Big(-\frac{\psafe-\mu}{1-\sigma}\frac{1}{\mu(1-\psafe)(\rhaz(\mint)-\frac{\rho\psafe}{\mu(1-\psafe)})}\Big)=\\
	&-\frac{\psafe-\mu}{1-\sigma}\exp(-\rho \mint)(1-\rcdf(\mint)).
\end{align*}
It follows that
\begin{align*}
	\frac{d\Delta}{d\sigma}=-\exp(-\rho \tau_P)(1-\rcdf(\tau_P))(\psafe-\mu)+\frac{\psafe-\mu}{1-\sigma}\exp(-\rho \mint)(1-\rcdf(\mint)).
\end{align*}
Recalling that (1) $\mint<\tau_P$, (2) $\exp(-\rho\cdot)(1-\rcdf(\cdot))$ is a strictly decreasing function, we have 
\begin{align*}
	\frac{d\Delta}{d\sigma}>(\psafe-\mu)\exp(-\rho \mint)(1-\rcdf(\mint))(\frac{1}{1-\sigma}-1)>0.
\end{align*}
Hence, $\Delta(0)=0$ and $\Delta(\cdot)$ is increasing in $\sigma$ for $\sigma<\bar{\sigma}$. Therefore for all $\sigma\in(0,\bar{\sigma}]$, we have $\Delta(\sigma)>0$.

\textit{Step 2.  A $\widetilde{\sigma}\leq 1$ exists such that commitment to naive search delivers \ppp{} a higher payoff than any equilibrium without commitment if $\bar{\sigma}<\sigma<\widetilde{\sigma}$}. Without commitment, the \ppp{}'s payoff in any equilibrium is $u_P^{FB}(0)=\psafe$ for $\sigma>\bar{\sigma}$. Thus, 
\begin{align*}
	\Delta(\sigma)=u_P^N(\sigma)-\psafe.
\end{align*}
Note that $\Delta(\cdot)$ is continuous at $\bar{\sigma}$, because $\tau_M(\sigma)\rightarrow 0$ as $\sigma\rightarrow\bar{\sigma}^-$, and $u_P^{FB}(0)=\psafe$. It follows from Step 1 that $\Delta(\bar{\sigma})>\psafe$. Furthermore, $\Delta(\sigma)$ is linear for $\sigma>\bar{\sigma}$. The result follows.\end{proof}

\begin{proof}[Proof of Proposition \ref{delint}]

To prove the first claim of the proposition it is sufficient to show that $U'(\psafe|\psafe)<0$, where $U'(\cdot|\psafe)$ is the derivative of $U(\cdot|\cdot)$ with respect to $\psafe _I$. In this case, the \ppp{}'s payoff of delegating $U(\psafe_I|\psafe)$ is decreasing in $\psafe_I$ at $\psafe_I=\psafe$. Thus, for an interval of $\psafe_I$ just below $\psafe$, we have $U(\psafe_I|\psafe)>U(\psafe|\psafe)$.

  To simplify the notation of the proof, let $u_t^{FB}(\cdot)\equiv du_P^{FB}(t)/dt$ and $u'_{I}(t|\psafe_I,\psafe)=\partial^2 u_P(t|\psafe_I,\psafe)/\partial t\partial\psafe_I$, i.e., the derivative of $u_P'(t|\psafe_I,\psafe)$
 with respect to $\psafe_I$.
 
\textit{Step 1. Recover the expression in the text.} 
 \begin{align*}U(\psafe_I|\psafe)&=\int_{\mint(\psafe_I)}^{\tau_P(\psafe_I)}\ppdf(s|\psafe_I)u_P(s|\psafe_I,\psafe)ds+(1-\pcdf(\tau_P^-(\psafe_I)|\psafe_I))u_P(\tau_P(\psafe_I)|\psafe_I,\psafe)
 \end{align*}
 Note that $\ppdf(\cdot|\psafe_I)$ is continuous, and $u_P(\cdot|\psafe_I,\psafe)$ is continuous and  differentiable on $[\tau_M(\psafe_I),\tau_P(\psafe_I)]$. Integrating by parts, we have 
 \begin{align*}
	&=-\int_{\mint(\psafe_I)}^{\tau_P(\psafe_I)}u_P(s|\psafe_I,\psafe)d[1-\pcdf(s|\psafe_I)]+(1-\pcdf(\tau_P^-(\psafe_I)|\psafe_I))u_P(\tau_P(\psafe_I)|\psafe_I,\psafe)\\
	&=(1-\pcdf(\mint(\psafe_I))|\psafe_I)u(\mint(\psafe_I)|\psafe_I)-(1-\pcdf(\tau_P^-(\psafe_I)|\psafe_I))u_P(\tau_P(\psafe_I)|\psafe_I,\psafe)\\
	&+\int_{\mint(\psafe_I)}^{\tau_P(\psafe_I)}(1-\pcdf(s|\psafe_I))u'_P(s|\psafe_I,\psafe)ds+(1-\pcdf(\tau_P^-(\psafe_I)|\psafe_I))u_P(\tau_P(\psafe_I)|\psafe_I,\psafe)\\
	&=\int_{\mint(\psafe_I)}^{\tau_P(\psafe_I)}(1-\pcdf(s|\psafe_I))u'_P(s|\psafe_I,\psafe)ds+(1-\pcdf(\mint(\psafe_I)|\psafe_I)u_P(\mint(\psafe_I)|\psafe_I).
	\end{align*}
From Proposition \ref{charval}, we have $(1-\pcdf(\mint(\psafe_I)|\psafe_I)=1$ and $u_P(\mint(\psafe_I)|\psafe_I)=u_P^{FB}(\mint(\psafe_I))$.

 \noindent\textit{Step 2.	Differentiate $U(\psafe_I|\psafe)$ with respect to $\psafe_I$.} 
\begin{align*}
U'(\psafe_I|\psafe)= (1-\pcdf(\tau_P^-(\psafe_I)|\psafe_I))u'_P(\tau_P^-(\psafe_I)|\psafe_I,\psafe)\tau_P'(\psafe_I)-(1-\pcdf(\tau_M(\psafe_I)|\psafe_I))u'_P(\mint(\psafe_I)|\psafe_I,\psafe)\tau_M'(\psafe_I)\\
+\int_{\mint(\psafe_I)}^{\tau_P(\psafe_I)}(1-\pcdf(s|\psafe_I))u'_I(s|\psafe_I,\psafe))-\frac{d\pcdf(s|\psafe_I)}{d\psafe_I}u_P'(s|\psafe_I,\psafe)ds
+u^{FB}_t(\mint(\psafe_I))\mint'(\psafe_I).
\end{align*}

 \noindent\textit{Step 3.	Evaluate at $\psafe_I=\psafe$.} 
 
 Note that at $\psafe_I=\psafe$, we have $u_P'(\cdot|\psafe,\psafe)=0$ on $[\mint,\tau_P]$. That is, the \ppp{}'s equilibrium payoff is constant inside the support, and furthermore, all equilibrium objects are continuous in $\psafe_I$. Therefore,
\begin{align*}
U'(\psafe|\psafe)= 
\int_{\mint}^{\tau_P}(1-\pcdf(s))u'_I(s|\psafe,\psafe)ds
+u^{FB}_t(\mint)\mint'(\psafe).
\end{align*}

\textit{Step 4. Derive an expression for $u'_I(s|\psafe,\psafe)$.} 

Substituting the definitions in the text,
 \begin{align*}
u_P(t|\psafe_I,\psafe)&=u_P(\tau_M(\psafe_I)|\psafe_I,\psafe)\\
&+\int^t_{\tau_M(\psafe_I)}\exp(-\rho s)\{((1-\mu)\psafe +\mu)\rpdf(s)(1-\sigma \acdf(s|\psafe_I))+\mu\sigma \apdf(s|\psafe_I)(1-\rcdf(s))\}ds\\
&+\exp(-\rho t)(1-\rcdf(t))(1-\sigma\acdf(t|\psafe_I))\psafe-\exp(-\rho \tau_M(\psafe_I))(1-\rcdf(\tau_M(\psafe_I))\psafe.
\end{align*}
From Proposition \ref{charval} $\acdf(\cdot)$ has no mass points. It follows that $\acdf(\tau_M)=0$ and  $\acdf(\cdot)$ is absolutely continuous. By implication, for $t\in(\tau_M(\psafe_I),\tau_P(\psafe_I))$ we have 
\begin{align*}
\exp(-\rho t)(1-\rcdf(t))(1-\sigma\acdf(t|\psafe_I))\psafe-\exp(-\rho \tau_M(\psafe_I))(1-\rcdf(\tau_M(\psafe_I))\psafe=\\
-\int_{\tau_M(\psafe_I)}^t\exp(-\rho s)\{\rpdf(s)(1-\sigma \acdf(s|\psafe_I))+\sigma \apdf(s|\psafe_I))(1-\rcdf(s))+\rho(1-\rcdf(s))(1-\sigma \acdf(s|\psafe_I)) \}\psafe ds.%	\exp(-\rho t)W_\phi^P(t)\psafe-\exp(-\rho \tau_M)\psafe=-\int_{\tau_M}^t\exp(-\rho s)(w_0^P(s)+w^P_1(s)+\rho W_\phi^P(s))\psafe ds
\end{align*}
To see this, note that the integrand is equal to the derivative in $t$ of $\exp(-\rho t)(1-\rcdf(t))(1-\sigma\acdf(t|\psafe_I))\psafe$ wherever the derivative exists, and the equality above follows from absolute continuity.
Substituting the previous equality and simplifying, 
\begin{align*}
u_P(t|\psafe_I,\psafe)&=u_P(\tau_M(\psafe_I)|\psafe_I,\psafe)+\\
&\int^t_{\tau_M(\psafe_I)}\exp(-\rho s)\{\mu(1-\psafe)\rpdf(s)(1-\sigma \acdf(s|\psafe_I))+(\mu-\theta)\sigma \apdf(s|\psafe_I)(1-\rcdf(s))\\
&\quad\quad\quad\quad\quad-\rho\theta(1-\sigma\acdf(s||\psafe_I))(1-\rcdf(s))\}ds=\\
&\int^t_{\tau_M\psafe_I)}\exp(-\rho s)(1-\sigma\acdf(s|\psafe_I))(1-\rcdf(s))(\psafe-\mu)\{\frac{\mu(1-\psafe)\rhaz(s)-\rho\theta}{\psafe-\mu}-\ahaz(s|\psafe_I)\}ds,
\end{align*}
where $\ahaz(s|\psafe_I)=\sigma\apdf(s|\psafe_I)/(1-\sigma\acdf(s|\psafe_I))$. From Proposition \ref{charval}, it is readily verified that in equilibrium,
\begin{align*}
	\ahaz(s|\psafe_I)=\frac{\mu(1-\psafe_I)\rhaz(s)-\rho\theta_I}{\psafe_I-\mu}.
\end{align*}
It is also helpful to denote the derivative of $\ahaz(s|\psafe_I)$ with respect to $\psafe_I$. Let
\begin{align*}
	h_A(s|\psafe_I)\equiv&\frac{d\ahaz(s|\psafe_I)}{d\psafe_I}
=\frac{-(\mu\rhaz(s)+\rho)(\psafe_I-\mu)-\mu(1-\psafe_I)\rhaz(s)+\rho\psafe}{(\psafe_I-\mu)^2}\\
	=&-\frac{\mu(1-\mu)}{(\psafe_I-\mu)^2}(\rhaz(s)-\frac{\rho}{1-\mu}).
\end{align*}
With this notation, we have
\begin{align*}
u_P(t|\psafe_I,\psafe)&=u_P(\tau_M(\psafe_I)|\psafe_I,\psafe)+\\&
\int^t_{\tau_M(\psafe_I)}\exp(-\rho s)(1-\sigma\acdf(s|\psafe_I))(1-\rcdf(s))(\psafe-\mu)\{\ahaz(s|\psafe)-\ahaz(s|\psafe_I)\}ds.
\end{align*}
Differentiating with respect to $t$, we have 
\begin{align*}
u'_P(t|\psafe_I,\psafe)=
\exp(-\rho t)(1-\sigma\acdf(t|\psafe_I))(1-\rcdf(t))(\psafe-\mu)\{\ahaz(t|\psafe)-\ahaz(t|\psafe_I)\}.
\end{align*}
Differentiating with respect to $\psafe_I$ and then evaluating at $\psafe_I=\psafe$, we have 
\begin{align*}
u'_I(t|\psafe,\psafe)=
-\exp(-\rho t)(1-\sigma\acdf(t)(1-\rcdf(t))(\psafe-\mu)h_A(t|\psafe).
\end{align*}
Note that the term involving differentiation of $1-\sigma\acdf(t|\psafe_I)$ vanishes at $\psafe_I=\psafe$, because $\ahaz(t|\psafe)-\ahaz(t|\psafe)=0$.

\textit{Step 5. Derive an expression for $u_t^{FB}(\tau_M)\tau_M'(\psafe)$}.

Differentiating $u_P^{FB}(\cdot)$ with respect to $t$ and substituting $\mint$ gives
\begin{align*}
u_t^{FB}(\tau_M)=\frac{du_P^{FB}(\mint)}{dt}
&=\exp(-\rho \mint)(1-\rcdf(\mint))\mu(1-\psafe)\{\rhaz(\mint)-\frac{\rho\psafe}{\mu(1-\psafe)}\}\\
&=\exp(-\rho \mint)(1-\rcdf(\mint))\{\mu(1-\psafe)\rhaz(\mint)-\rho\psafe\}\\
&=\exp(-\rho \mint)(1-\rcdf(\mint))(\psafe-\mu)\ahaz(\tau_M|\psafe)
\end{align*}
(see also the proof of Lemma \ref{fb}).

To find $\mint'(\psafe)$ recall that, Proposition  \ref{charval} gives an implicit function relating $\mint$, $\tau_P$, and $\sigma$,
\begin{align*}
\exp\{-\int_{\mint}^{\tau_P}\frac{\mu(1-\psafe)\rhaz(s)-\rho\psafe}{\psafe-\mu}ds\}=1-\sigma\iff \exp\{-\int_{\mint}^{\tau_P}\ahaz(s|\psafe)ds\}=1-\sigma.
\end{align*}
Differentiating with respect to $\theta$, we have 
\begin{align*}
\exp\{-\int_{\mint}^{\tau_P}\ahaz(s|\psafe)ds\}\Big(\ahaz(\mint|\psafe)\tau_M'(\psafe)-\ahaz(\tau_P|\psafe)\tau_P'(\psafe)-\int_{\tau_M(\psafe)}^{\tau_P(\psafe)}h_A(s|\psafe)ds\Big)=0.
\end{align*}
Note that $\rhaz(\tau_P)=\rho\psafe/(\mu(1-\psafe))$, which implies $\ahaz(\tau_P|\psafe)=0$. By implication
\begin{align*}
\ahaz(\mint|\psafe)\tau_M'(\psafe)-\int_{\tau_M}^{\tau_P}h_A(s|\psafe)ds=0\Rightarrow \tau_M'(\psafe)=\frac{\int_{\tau_M}^{\tau_P}h_A(s|\psafe)ds}{\ahaz(\mint|\psafe)}.
\end{align*}
Combining these calculations, we have 
\begin{align*}
	u_t^{FB}(\tau_M)\tau_M'(\psafe)=\exp(-\rho \mint)(1-\rcdf(\mint))(\psafe-\mu).
\int_{\tau_M}^{\tau_P}h_A(s|\psafe)ds.
\end{align*}

\textit{Step 6. Show that $U'(\psafe|\psafe)<0$}.

From Step 3,
\begin{align*}
U'(\psafe|\psafe)= 
\int_{\mint}^{\tau_P}(1-\pcdf(s))u'_I(s|\psafe,\psafe)ds
+u^{FB}_t(\mint)\mint'(\psafe).
\end{align*}
Substituting the expressions derived in Steps 4 and 5, we have 
\begin{align*}
U'(\psafe|\psafe)= -&\int_{\mint}^{\tau_P}\exp(-\rho s)(1-\sigma\acdf(s)(1-\rcdf(s))(\psafe-\mu)(1-\pcdf(s))h_A(s|\psafe)ds	+\\
&\exp(-\rho \mint)(1-\rcdf(\mint))(\psafe-\mu).
\int_{\tau_M}^{\tau_P}h_A(s|\psafe)ds.
\end{align*}
Let $Q(s)\equiv \exp(-\rho s)(1-\sigma\acdf(s))(1-\rcdf(s))(1-\pcdf(s))(\psafe-\mu)$. Note that $Q(s)>0$ for all $s\geq 0$ and $Q(s)$ is strictly decreasing. With this notation, 
\begin{align*}
U'(\psafe|\psafe)=\int_{\mint}^{\tau_P}\{Q(\mint)-Q(s)\}h_A(s|\psafe)ds.
\end{align*}
Because $Q(\cdot)$ is strictly decreasing, the term in curly braces is strictly positive. To complete the proof, it suffices to show that $h_A(s|\psafe)<0$ for $s\in[\mint,\tau_P]$. As shown in Step 4,
\begin{align*}
	h_A(s|\psafe)=-\frac{\mu(1-\mu)}{(\psafe-\mu)^2}(\rhaz(s)-\frac{\rho}{1-\mu}).
\end{align*}
For $s\leq\tau_P$, we have $\rhaz(s)\geq \frac{\rho\psafe}{\mu(1-\psafe)}$. Hence, for such $s$, 
\begin{align*}
	h_A(s|\psafe)\leq -\frac{\mu(1-\mu)}{(\psafe-\mu)^2}(\frac{\rho\psafe}{\mu(1-\psafe)}-\frac{\rho}{1-\mu})=-\frac{\rho\psafe(1-\mu)-\rho\mu(1-\psafe)}{(\psafe-\mu)^2(1-\psafe)}=-\frac{\rho}{(\psafe-\mu)(1-\psafe)}<0.
\end{align*}
This completes the proof of the first claim.

To prove the second claim, suppose that the \ppp{} has delegated to an intermediary with $\psafe_I<\psafe$. It will be established that if $\psafe_I$ sufficiently close to $\psafe$, then the \ppp{} does not want to overrule the \aaa{}'s chosen action following a news arrival. Obviously, type-0 news reveals $\omega=0$ and both \ppp{} and intermediary prefer the safe action in this case. Recall from Proposition \ref{charval} that in equilibrium the intermediary selects the risky action with probability 1 following type-1 news at all times $[0,\tau_P(\psafe_I)]$. Type-1 news arrivals before $\mint(\psafe_I)$ are real with probability 1, and both \ppp{} and intermediary prefer the risky action for such arrivals. What remains is to show that the \ppp{} prefers the risky action following type-1 news arrivals that occur in $[\mint(\psafe_I),\tau_P(\psafe_I)]$. Given an arrival at such a time, the \ppp{}'s posterior belief is given by (\ref{belief}),
\begin{align*}
	\bl(s|\psafe_I)=\frac{\mu\rpdf(s)(1-\apri\acdf(s|\psafe_I))+\mu\apri\apdf(s|\psafe_I)(1-\rcdf(s))}{\mu\rpdf(s)(1-\apri\acdf(s|\psafe_I))+\apri\apdf(s|\psafe_I)(1-\rcdf(s))}=\frac{\mu[\rhaz(s)+\ahaz(s|\psafe_I)]}{\mu\rhaz(s)+\ahaz(s|\psafe_I)}.
\end{align*}
Thus, the \ppp{} strictly prefers the risky action following a type-1 news arrival if
\begin{align*}
\frac{\mu[\rhaz(s)+\ahaz(s|\psafe_I)]}{\mu\rhaz(s)+\ahaz(s|\psafe_I)}>\psafe\iff \ahaz(s|\psafe_I)<\frac{\mu(1-\psafe)\rhaz(s)}{\psafe-\mu}.
\end{align*}
It is readily verified that in the equilibrium of Proposition \ref{charval}, 
\begin{align*}
	\ahaz(s|\psafe_I)=\frac{\mu(1-\psafe_I)\rhaz(s)-\rho\theta_I}{\psafe_I-\mu}.
\end{align*}
Thus, the \ppp{} strictly prefers risky following a type-1 news arrival if
\begin{align*}
	\frac{\mu(1-\psafe_I)\rhaz(s)-\rho\theta_I}{\psafe_I-\mu}<\frac{\mu(1-\psafe)\rhaz(s)}{\psafe-\mu}.
\end{align*}
Obviously, for $\psafe_I=\psafe$, left hand side is strictly smaller. Because left hand side is continuous, this inequality holds if $\psafe_I$ and $\psafe$ are close.
\end{proof}
%
%\begin{align*}
%	\bl(s)=\frac{\mu\rpdf(s)(1-\apri\acdf(s))+\mu\apri\apdf(s)(1-\rcdf(s))}{\mu\rpdf(s)(1-\apri\acdf(s))+\apri\apdf(s)(1-\rcdf(s))}=\frac{\mu[\rhaz(s)+\ahaz(s)]}{\mu\rhaz(s)+\ahaz(s)}=\frac{}{}.
%\end{align*}

\section{Supplemental Appendix}
\setcounter{page}{1}

\subsection{Proof of Lemmas used in Proposition \ref{strucval}}
\noindent \textbf{Observation 1.} Note that 
\begin{align*}
	\int_0^tw^P_0(s)+w^P_1(s)ds+W^P_\phi(t)=1.
\end{align*}
Indeed, the integral is the probability that a type-0 or type-1 arrival occurs at times less than or equal to $t$, while $W_\phi^P(t)$ is the probability that an arrival of news takes longer than $t$. Rearranging, for $t_1\leq t_2$ we have
\begin{align*}
	\int_{t_1}^{t_2}w^P_0(s)+w^P_1(s)ds=W^P_\phi(t_1)-W^P_\phi(t_2).
\end{align*}
Similarly, for the \aaa{} we have,
\begin{align*}
	\int_{t_1}^{t_2}w^A_0(s)+w^A_1(s)+w_S^A(s)ds=W^A_\phi(t_1)-W^A_\phi(t_2).
\end{align*}

%\begin{lemma}\label{pbounds} If (\ref{PICa}), then inequalities $(i),(ii),(iii)$ hold,
%\begin{enumerate}
%	\item[(i)] $u_P(t')-u_P(t)\leq (1-\sigma \acdf(t))\{u_P^{FB}(t')-u_P^{FB}(t)\}$ for $t<t'$. 
%	\item[(ii)] $u_P(t)<u_P(\tau_P)$ for $\tau_P<t\leq\infty$.
%
%	\item[(iii)] $u_P(t')-u_P(t)\geq -(\exp(-\rho t)-(\exp(-\rho t'))W_\phi^P(t)\psafe$ for $t<t'$.
%	\item[(iv)] Part $(i)$ holds as an equality if $\acdf(t)=\acdf(t')$.
%\end{enumerate}
%\end{lemma}
\begin{proof}[Proof of Lemma \ref{pbounds}]

\textit{Proof of Parts (i) and (ii):} By definition, 
\begin{align*}
	&u_P(t')-u_P(t)=\\&\int_t^{t'}\exp(-\rho s)(w_0^P(s)\psafe+w_1^P(s)(\psafe+a(s)(\bl(s)-\psafe))ds+\exp(-\rho t')W_\phi^P(t')\psafe-\exp(-\rho t)W_\phi^P(t)\psafe.
\end{align*}
\textit{Step 1: Derive a bound on the second term of the integrand}.
\begin{align*}
&	w_1^P(s)(\psafe+a(s)(\bl(s)-\psafe))=w_1^P(s)\psafe (1-a(s))+w_1^P(s)\bl(s)a(s)=\\
&	\underbrace{(\mu\rpdf(s)(1-\sigma\acdf(s))+\sigma\apdf(s)(1-\rcdf(s))}_{w_1^P(s)}\psafe(1-a(s))+\underbrace{(\mu\rpdf(s)(1-\sigma\acdf(s))+\mu\sigma\apdf(s)(1-\rcdf(s)))}_{w_1^P(s)\bl(s)}a(s).
\end{align*}
Combining terms, we have 
\begin{align*}
&	w_1^P(s)(\psafe+a(s)(\bl(s)-\psafe))=\\
&	\mu\rpdf(s)(1-\sigma\acdf(s))(\psafe(1-a(s))+a(s))+\sigma\apdf(s)(1-\rcdf(s))(\psafe(1-a(s))+a(s)\mu).
\end{align*}
\noindent Recall that $\mu<\psafe<1$, and hence, 
\begin{align}\label{Q}
	w_1^P(s)(\psafe+a(s)(\bl(s)-\psafe))\leq
	\mu\rpdf(s)(1-\sigma\acdf(s))+\sigma\apdf(s)(1-\rcdf(s))\psafe.
\end{align}
\textit{Step 2: Apply the bound derived in Step 1}. 

\noindent Recall that $w_0^P(s)=(1-\mu)\rpdf(s)(1-\sigma\acdf(s)$. Using this and (\ref{Q}), we have 
\begin{align}\label{QQ}
	u_P(t')-u_P(t)\leq & \int_t^{t'}\exp(-\rho s)\{((1-\mu)\psafe+\mu)\rpdf(s)(1-\sigma\acdf(s))+\sigma\apdf(s)(1-\rcdf(s))\psafe\} ds+\\
	\nonumber &\exp(-\rho t')(1-\rcdf(t'))(1-\sigma\acdf(t'))\psafe-\exp(-\rho t)(1-\rcdf(t))(1-\sigma\acdf(t))\psafe.
\end{align}

\noindent \textit{Step 3: Derive an alternative expression for the RHS of inequality (\ref{QQ})}.

Consider the following integral,
\begin{align*}
	&\int_{t}^\infty\sigma\apdf(s)\min\{u_P^{FB}(s)-u_P^{FB}(t),u_P^{FB}(t')-u_P^{FB}(t)\}ds=\\
	&\underbrace{\int_{t}^{t'}\sigma\apdf(s)(u_P^{FB}(s)-u_P^{FB}(t))ds}_{A}+\underbrace{(1-\sigma\acdf(t'))(u_P^{FB}(t')-u_P^{FB}(t))}_{B}.
\end{align*}
Focus on term A. Substituting, we have
\begin{align*}
&\int_{t}^{t'}\sigma\apdf(s)[\int_t^s\exp(-\rho x)\rpdf(x)(\mu+(1-\mu)\psafe)dx+\exp(-\rho s)(1-\rcdf(s)\psafe-\exp(-\rho t)(1-\rcdf(t))\psafe]ds=\\
&\int_{t}^{t'}\int_t^{s}\sigma\apdf(s)\{\exp(-\rho x)\rpdf(x)(\mu+(1-\mu)\psafe)dxds+\int_{t}^{t'}\sigma\apdf(s)\exp(-\rho s)(1-\rcdf(s))\psafe\} ds\\
&-	\sigma(\acdf(t')-\acdf(t))(\exp(-\rho t)(1-\rcdf(t))).\end{align*}
Reversing the order of integration of the double integral, we have 
\begin{align*}
%\int_{t}^{t'}\int_x^{t'}\sigma\apdf(s)\exp(-\rho x)\rpdf(x)(\mu+(1-\mu)\psafe)dsdx+\int_{t}^{t'}\sigma\apdf(s)\exp(-\rho s)(1-\rcdf(s))\psafe ds\\
%	-\sigma(\acdf(t')-\acdf(t))(-\exp(-\rho t)(1-\rcdf(t)))=\\
%	\int_{t}^{t'}\sigma(\acdf(t')-\acdf(x))\exp(-\rho x)\rpdf(x)(\mu+(1-\mu)\psafe)dx+\int_{t}^{t'}\sigma\apdf(s)\exp(-\rho s)(1-\rcdf(s))\psafe ds\\
%	-\sigma(\acdf(t')-\acdf(t))(\exp(-\rho t)(1-\rcdf(t)))=\\
	\int_{t}^{t'}\sigma(\acdf(t')-\acdf(s))\exp(-\rho s)\rpdf(s)(\mu+(1-\mu)\psafe)+\sigma\apdf(s)\exp(-\rho s)(1-\rcdf(s))\psafe ds\\
	-\sigma(\acdf(t')-\acdf(t))(\exp(-\rho t)(1-\rcdf(t))).
\end{align*}
Now consider term $B$. We have
\begin{align*}
	&(1-\sigma\acdf(t'))\{u_P^{FB}(t')-u_P^{FB}(t)\}=\\	
	&(1-\sigma\acdf(t'))\{\int_t^{t'}\exp(-\rho s)\rpdf(s)(\mu+(1-\mu)\psafe))ds+\exp(-\rho t')(1-\rcdf(t'))\psafe-\exp(-\rho t)(1-\rcdf(t))\psafe\}.
\end{align*}
Adding terms A and B we have,
\begin{align*}
	\int_{t}^{t'}(1-\sigma\acdf(s))\exp(-\rho s)\rpdf(s)(\mu+(1-\mu)\psafe)+\sigma\apdf(s)\exp(-\rho s)(1-\rcdf(s))\psafe ds\\
	+(1-\sigma\acdf(t'))\exp(-\rho t')(1-\rcdf(t))(1-\sigma\acdf(t))\exp(-\rho t)(1-\rcdf(t)),
\end{align*}
which is the RHS of inequality (\ref{QQ}). Therefore, inequality (\ref{QQ}) can be written
\begin{align}\label{QQQ}
	u_P(t')-u_P(t)\leq \int_{t}^\infty\sigma\apdf(s)\min\{u_P^{FB}(s)-u_P^{FB}(t),u_P^{FB}(t')-u_P^{FB}(t)\}ds.
\end{align}

\textit{Proof of (i).} Notice that
\begin{align*}
	\min\{u_P^{FB}(s)-u_P^{FB}(t),u_P^{FB}(t')-u_P^{FB}(t)\}\leq u_P^{FB}(t')-u_P^{FB}(t).
\end{align*} 
Substituting into (\ref{QQQ}), we have 
\begin{align}\label{QQQQ}
	u_P(t')-u_P(t)\leq \int_{t}^\infty\sigma\apdf(s)(u_P^{FB}(t')-u_P^{FB}(t))ds=(1-\sigma\acdf(t))(u_P^{FB}(t')-u_P^{FB}(t)).
\end{align}

\textit{Proof of (ii).} Suppose $t>\tau_P$. Applying inequality (\ref{QQQQ}), we have,
\begin{align*}\
	u_P(t)-u_P(\tau_P)\leq (1-\sigma\acdf(\tau_P))(u_P^{FB}(t)-u_P^{FB}(\tau_P)).
\end{align*}
From Lemma \ref{fb} we have that the single peak of $u^{FB}(\cdot)$ is at $\tau_P$. By implication, the right hand side is strictly negative for all $\tau_P<t<\infty$. Furthermore, the right hand side is also strictly decreasing in $t$, and hence $\lim_{t\rightarrow \infty}u_P^{FB}(t)-u_P^{FB}(\tau_P)<0$.

\textit{Proof of (iii)}. By definition, 
\begin{align*}
	&u_P(t')-u_P(t)=\\&\int_t^{t'}\exp(-\rho s)(w_0^P(s)\psafe+w_1^P(s)(\psafe+a(s)(\bl(s)-\psafe))ds+\exp(-\rho t')W_\phi^P(t')\psafe-\exp(-\rho t)W_\phi^P(t)\psafe.
\end{align*}
Given (\ref{PICa}), we have $\psafe+a(s)(\bl(s)-\psafe)\geq \psafe$, and hence,
\begin{align*}
	u_P(t')-u_P(t)\geq \psafe&\int_t^{t'}\exp(-\rho s)(w_0^P(s)+w_1^P(s))ds+\exp(-\rho t')W_\phi^P(t')\psafe-\exp(-\rho t)W_\phi^P(t)\psafe\geq \\
 \exp(-\rho t') \psafe&\int_t^{t'}(w_0^P(s)+w_1^P(s))ds+\exp(-\rho t')W_\phi^P(t')\psafe-\exp(-\rho t)W_\phi^P(t)\psafe.
\end{align*}
Using Observation 1, we have 
\begin{align*}
	u_P(t')-u_P(t)\geq 
 \exp(-\rho t') \psafe(W_\phi^P(t)-W_\phi^P(t'))&+\exp(-\rho t')W_\phi^P(t')\psafe-\exp(-\rho t)W_\phi^P(t)\psafe=\\
-[\exp(-\rho t)-\exp(-\rho t')]W_\phi^P(t)\psafe.
\end{align*}

\textit{Proof of (iv).} If $\acdf(t)=\acdf(t')$, then $\apdf(s)=0$ for almost all $s\in(t,t')$. It follows that for almost all $s\in(t,t')$, $\bl(s)=1$, $a(s)=1$ (from (\ref{PICa})), $w^P_0(s)=(1-\mu)\rpdf(s)(1-\sigma\acdf(t))$, $w^P_1(s)=\mu\rpdf(s)(1-\sigma\acdf(t))$. Furthermore, $W_\phi^P(t)=(1-\rcdf(t))(1-\sigma\acdf(t))$ and $W_\phi^P(t')=(1-\rcdf(t'))(1-\sigma\acdf(t))$. The result follows from routine substitution.
%\begin{align*}
%	&u_P(t')-u_P(t)=\\&\int_t^{t'}\exp(-\rho s)(w_0^P(s)\psafe+w_1^P(s)(\psafe+a(s)(\bl(s)-\psafe))ds+\exp(-\rho t')W_\phi^P(t')\psafe-\exp(-\rho t)W_\phi^P(t)\psafe=\\
%		&(1-\sigma\acdf(t))\{\int_t^{t'}\exp(-\rho s)(\rpdf(s)(\mu+(1-\mu)\psafe)ds+\exp(-\rho t')(1-\rcdf(t'))\psafe-\exp(-\rho t)(1-\rcdf(t))\psafe=\\
%		&(1-\sigma\acdf(t))(u_P^{FB}(t')-u_P^{FB}(t)).
%\end{align*}
\end{proof}

%\begin{lemma}\label{pcont} If (\ref{PICa}), then $u_P(\cdot)$ is continuous, regardless of the \aaa{}'s strategy.
%\end{lemma}
\begin{proof}[Proof of Lemma \ref{pcont}] 

Using Parts (i) and (iii) of Lemma \ref{pbounds}, for $t'>t$
\begin{align*}
-(\exp(-\rho t)-(\exp(-\rho t'))W_\phi^P(t)\psafe	\leq u_P(t')-u_P(t)\leq (1-\sigma\acdf(t))\{u_P^{FB}(t')-u_P^{FB}(t)\}.
\end{align*}
By continuity of $u_P^{FB}(\cdot)$ and $\exp( \cdot)$, we have $\lim_{t'\rightarrow t^+}u_P(t')=u_P(t)$.
Similarly, for $t'<t$
\begin{align*}
-(\exp(-\rho t')-(\exp(-\rho t))W_\phi^P(t')\psafe	\leq u_P(t)-u_P(t')\leq (1-\sigma \acdf(t'))\{u_P^{FB}(t)-u_P^{FB}(t')\},
\end{align*}
Note that (1) $W_\phi^P(t')\leq W_\phi^P(t)$,  (2) $(\exp(-\rho t')-(\exp(-\rho t))>0$, and (3) $\acdf(t')\geq 0$. It follows that
\begin{align*}
-(\exp(-\rho t')-(\exp(-\rho t))W_\phi^P(t)\psafe	\leq u_P(t)-u_P(t')\leq u_P^{FB}(t)-u_P^{FB}(t').
\end{align*}
Continuity of $u^{FB}(\cdot)$ and $\exp(\cdot)$ imply $\lim_{t'\rightarrow t^-}u_P(t')=u_P(t)$.\end{proof}

%\begin{lemma}\label{ptop} If (\ref{PICa}) and (\ref{PICt}), then $\pcdf(\tau_P)=1$.
%\end{lemma}
\begin{proof}[Proof of Lemma \ref{ptop}] From Lemma \ref{pbounds}, we have $u_P(t)<u_P(\tau_P)$ for $\tau_P<t\leq \infty$. From (\ref{PICt}), we have $\ppdf(t)=0$ for such $t$.  By implication, $\pcdf(\tau_P)=1$.
\end{proof}
%
%\textbf{Definition:} Let
%\begin{align*}
%	v(t)\equiv \int_0^t\exp(-\rho s)\rpdf(s)(\mu+(1-\mu)\asafe) ds+\exp(-\rho t)\asafe(1-\rcdf(t)).
%\end{align*}
%\noindent \textbf{Definition:} Let
%\begin{align*}
%&v(t)\equiv \int_0^t\exp(-\rho x)\rpdf(x)(\mu+(1-\mu)\asafe)dx+\exp(-\rho t)(1-\rcdf(t))\asafe.
%%=\\
%%&u_A^{FB}(t)-\exp(-\rho t)(1-\rcdf(t))(\mu-\asafe).
%\end{align*}
%By a calculation similar to Lemma \ref{fb}, 
%\begin{align*}
%	v'(t)> 0\iff \rhaz(t)>\frac{\rho\asafe}{\mu(1-\asafe)}.
%\end{align*}
%Note that $\asafe<\psafe\Rightarrow \frac{\rho\asafe}{\mu(1-\asafe)}<\phi_P$. Therefore, $v(\cdot)$ is strictly increasing for $t\leq\tau_P$.

%\begin{lemma}\label{abounds} Inequalities $(i),(ii)$ hold for $t<t'\leq \tau_P$, regardless of the \ppp{}'s strategy. 
%	\begin{align*}
%(i)\quad &	u_A(t')-u_A(t)\leq\\
%&	(1-\pcdf(t))\{v(t')-v(t)\}+
%\exp(-\rho t')W_\phi^A(t')(\mu-\asafe) a(t')-\exp(-\rho t)W_\phi^A(t)(\mu-\asafe) a(t).\\
%(ii)\quad	&u_A(t')-u_A(t)\geq\\
%	&	\{\exp(-\rho t')W_\phi^A(t')a(t')-\exp(-\rho t)W_\phi^A(t)a(t)\}(\mu-\asafe)
%-\{\exp(-\rho t)-\exp(-\rho t')\}W_\phi^A(t)\asafe
%\end{align*}
% $(iii)$: if $\pcdf(t)=\pcdf(t')$,  $a(t)=a(t')=1$, and $a(s)=1$ for almost all $s\in(t,t')$, then
%\begin{align*}
%u_A(t')-u_A(t)=(1-\pcdf(t))(u_A^{FB}(t')-u_A^{FB}(t)).
%\end{align*}
%\end{lemma}

\noindent \textbf{Definition:} Let
\begin{align*}
&v(t)\equiv \int_0^t\exp(-\rho x)\rpdf(x)(\mu+(1-\mu)\asafe)dx+\exp(-\rho t)(1-\rcdf(t))\asafe.
%=\\
%&u_A^{FB}(t)-\exp(-\rho t)(1-\rcdf(t))(\mu-\asafe).
\end{align*}

\begin{proof}[Proof of Lemma \ref{abounds}]
Suppose $t<t'\leq \tau_P$. By definition,
\begin{align*}
&	u_A(t')-u_A(t)=\\
	&\int_{t}^{t'}\exp(-\rho s)\{(w_0^A(s)+w_S^A(s))\beta+w_1^A(s)((1-a(s))\asafe+a(s))\}ds+\nonumber\\
	&\exp(-\rho t')W_\phi^A(t')(\mu a(t')+(1-a(t'))\asafe)-\exp(-\rho t)W_\phi^A(t)(\mu a(t)+(1-a(t))\asafe)\nonumber.
	\end{align*}
\textit{Proof of (i)}. 
First, rearrange the previous expression:
\begin{align*}
&	u_A(t')-u_A(t)=\\
	&\int_{t}^{t'}\exp(-\rho s)\{(w_0^A(s)+w_S^A(s))\beta+w_1^A(s)((1-a(s))\asafe+a(s))\}ds+\exp(-\rho t')W_\phi^A(t')\asafe-\exp(-\rho t)W_\phi^A(t)\asafe+\nonumber\\
	&\exp(-\rho t')W_\phi^A(t')(\mu-\asafe) a(t'))-\exp(-\rho t)W_\phi^A(t)(\mu-\asafe)a(t)\nonumber.
	\end{align*}
Because $a(s)\leq 1$ and $\asafe<\mu<1$,we have 
\begin{align}\label{AQ}
&	u_A(t')-u_A(t)\leq\\
	&\int_{t}^{t'}\exp(-\rho s)\{(w_0^A(s)+w_S^A(s))\beta+w_1^A(s)\}ds+\exp(-\rho t')W_\phi^A(t')\asafe-\exp(-\rho t)W_\phi^A(t)\asafe+\nonumber\\
	&\exp(-\rho t')W_\phi^A(t')(\mu-\asafe) a(t')-\exp(-\rho t)W_\phi^A(t)(\mu-\asafe)a(t)\nonumber.
	\end{align}
\textit{Step 1: Provide a different expression the first line of the RHS of (\ref{AQ})}.

The expression to be rewritten is 
	\begin{align}\label{AQ1}
	\int_{t}^{t'}&\exp(-\rho s)\{(w_0^A(s)+w_S^A(s))\beta+w_1^A(s)\}ds+\exp(-\rho t')W_\phi^A(t')\asafe-\exp(-\rho t)W_\phi^A(t)\asafe=\nonumber\\		
		\int_{t}^{t'}&\exp(-\rho s)\{\rpdf(s)(1-\pcdf(s))(\mu+(1-\mu)\asafe)+\asafe\ppdf(s)(1-\rcdf(s))\}ds+\\
		&\exp(-\rho t')(1-\pcdf(t'))(1-\rcdf(t'))\asafe-\exp(-\rho t)\exp(-\rho t)(1-\pcdf(t))(1-\rcdf(t))\asafe\nonumber.
	\end{align}

%To do so, consider the function,
%\begin{align*}
%&v(t)\equiv \int_0^t\exp(-\rho x)\rpdf(x)(\mu+(1-\mu)\asafe)dx+\exp(-\rho t)(1-\rcdf(t))\asafe=\\
%&u_A^{FB}(t)-\exp(-\rho t)(1-\rcdf(t))(\mu-\asafe).
%\end{align*}
%By a calculation similar to Lemma \ref{fb}, 
%\begin{align*}
%	v'(t)> 0\iff \rhaz(t)>\frac{\rho\asafe}{\mu(1-\asafe)}.
%\end{align*}
For $t\leq \tau_P$, we have $\rhaz(t)\geq\phi_P>\rho\asafe/(\mu(1-\asafe))$. Thus, for $t\leq\tau_P$, function $v(\cdot)$ is strictly increasing.

Consider the integral
\begin{align*}
&	\int_t^{\infty}\ppdf(s)\min\{v(s)-v(t),v(t')-v(t)\}ds=\underbrace{\int_t^{\infty}\ppdf(s)\{v(s)-v(t)\}ds}_{A}+\underbrace{\strut (1-\pcdf(t'))(v(t')-v(t))}_{B}.
\end{align*}
which follows because $v(\cdot)$ is increasing for $t\leq \tau_P$.

First, consider term A.
\begin{align*}
&		\int_t^{t'}\ppdf(s)\{\int_t^s\exp(-\rho x)\rpdf(x)(\mu+(1-\mu)\asafe)dx+\exp(-\rho s)(1-\rcdf(s))\asafe-\exp(-\rho t)(1-\rcdf(t))\asafe\} ds,
\end{align*}
which is equal to
\begin{align*}
&	\int_t^{t'}\int_t^s\ppdf(s)\{\exp(-\rho x)\rpdf(x)(\mu+(1-\mu)\asafe)\}dxds+\int_t^{t'}\ppdf(s)\exp(-\rho s)(1-\rcdf(s))\asafe ds\\
	&-(\pcdf(t')-\pcdf(t))(\exp(-\rho t)(1-\rcdf(t))\asafe.
\end{align*}
Reversing the order of integration in the double integral, we have 
\begin{align*}
		&\int_t^{t'}\int_x^{t'}\ppdf(s) \exp(-\rho x)\rpdf(x)(\mu+(1-\mu)\asafe)dsdx+\int_t^{t'}\ppdf(s)\exp(-\rho s)(1-\rcdf(s))\asafe ds\\
	&-(\pcdf(t')-\pcdf(t))(\exp(-\rho t)(1-\rcdf(t))\asafe=\\
		&\int_t^{t'}(\pcdf(t')-\pcdf(x)) \exp(-\rho x)\rpdf(x)(\mu+(1-\mu)\asafe)dx+\int_t^{t'}\ppdf(s)\exp(-\rho s)(1-\rcdf(s))\asafe ds\\
	&-(\pcdf(t')-\pcdf(t))(\exp(-\rho t)(1-\rcdf(t))\asafe.
\end{align*}
Therefore, term A is
\begin{align*}
			&\int_t^{t'}(\pcdf(t')-\pcdf(s)) \exp(-\rho s)\rpdf(s)(\mu+(1-\mu)\asafe)+\ppdf(s)\exp(-\rho s)(1-\rcdf(s))\asafe ds\\
	&-(\pcdf(t')-\pcdf(t))(\exp(-\rho t)(1-\rcdf(t))\asafe.
\end{align*}
Next, consider term $B$:
\begin{align*}
&(1-\pcdf(t'))\{\int_{t}^{t'}\exp(-\rho s)\rpdf(s)(\mu+(1-\mu)\asafe)ds+\exp(-\rho t')(1-\rcdf(t'))\asafe-\exp(-\rho t)(1-\rcdf(t))\asafe\}.
\end{align*}
Adding terms A and B, we have 
\begin{align*}
		&\int_t^{t'}(1-\pcdf(s)) \exp(-\rho s)\rpdf(s)(\mu+(1-\mu)\asafe)ds+\ppdf(s)\exp(-\rho s)(1-\rcdf(s))\asafe ds\\
	&+(1-\pcdf(t'))(\exp(-\rho t')(1-\rcdf(t'))\asafe-(1-\pcdf(t))(\exp(-\rho t)(1-\rcdf(t))\asafe,
\end{align*}
which is exactly the expression in (\ref{AQ1}). 

\textit{Proof of (i).} Note first that by straightforward differentiation.
\begin{align*}
	v'(t)> 0\iff \rhaz(t)>\frac{\rho\asafe}{\mu(1-\asafe)}.
\end{align*}
Note that $\asafe<\psafe\Rightarrow \frac{\rho\asafe}{\mu(1-\asafe)}<\phi_P$. Therefore, $v(\cdot)$ is strictly increasing for $t\leq\tau_P$.

Using Step 1, we have,
\begin{align*}
&	u_A(t')-u_A(t)\leq \\
	&	\int_t^{\infty}\ppdf(s)\min\{v(s)-v(t),v(t')-v(t)\}ds+
\exp(-\rho t')W_\phi^A(t')(\mu-\asafe) a(t')-\exp(-\rho t)W_\phi^A(t)(\mu-\asafe)a(t)\leq\\
	&	(1-\pcdf(t))(v(t')-v(t))+
\exp(-\rho t')W_\phi^A(t')(\mu-\asafe) a(t')-\exp(-\rho t)W_\phi^A(t)(\mu-\asafe)a(t),
\end{align*}
where use has been made of the observation $\min\{v(s)-v(t),v(t')-v(t)\}\}\leq v(t')-v(t)$, completing the proof of (i).
%Substituting
%\begin{align*}
%	v(t')-v(t)=u_A^{FB}(t')-u_A^{FB}(t)+\{\exp(-\rho t)(1-\rcdf(t))-\exp(-\rho t')(1-\rcdf(t'))\}(\mu-\asafe),
%\end{align*}
%we have 
%\begin{align*}
%	&	(1-\pcdf(t))(u_A^{FB}(t')-u_A^{FB}(t)+\{\exp(-\rho t)(1-\rcdf(t))-\exp(-\rho t')(1-\rcdf(t'))\}(\mu-\asafe))+\\
%&\exp(-\rho t')W_\phi^A(t')(\mu-\asafe) a(t')-\exp(-\rho t)W_\phi^A(t)(\mu-\asafe)a(t)=\\
%	&	(1-\pcdf(t))(u_A^{FB}(t')-u_A^{FB}(t)-\exp(-\rho t')(1-\rcdf(t'))\}(\mu-\asafe))+\\
%&\exp(-\rho t')W_\phi^A(t')(\mu-\asafe) a(t')-\exp(-\rho t)W_\phi^A(t)(\mu-\asafe)(1-a(t))
%\end{align*}
%

\textit{Proof of (ii)}. Consider the expression for $u_A(t')-u(t)$; noting that $\asafe<1$ and $a(t)\geq 0$, we have 
	\begin{align*}
	&u_A(t')-u_A(t)\geq \exp(-\rho t')\asafe\int_{t}^{t'}\{w_0^A(s)+w_0^S(s)+w_1^A(s)\}ds+\\
	&\exp(-\rho t')W_\phi^A(t')(\mu a(t')+(1-a(t'))\asafe)-\exp(-\rho t)W_\phi^A(t)(\mu a(t)+(1-a(t))\asafe).
\end{align*}
Using Observation 1, we have
\begin{align*}
	&u_A(t')-u_A(t)\geq 	\exp(-\rho t')(W_\phi^A(t)-W_\phi^A(t'))\asafe+\\
	&\exp(-\rho t')W_\phi^A(t')(\mu a(t')+(1-a(t'))\asafe)-\exp(-\rho t)W_\phi^A(t)(\mu a(t)+(1-a(t))\asafe).
\end{align*}
Simplifying, 
\begin{align*}
%	&u_A(t')-u_A(t)\geq\\& 	\exp(-\rho t')(W_\phi^A(t))\asafe+
%\exp(-\rho t')W_\phi^A(t')(\mu-\asafe) a(t')-\exp(-\rho t)W_\phi^A(t)(\mu a(t)+(1-a(t))\asafe)=\\
%&(\exp(-\rho t')-\exp(-\rho t))W_\phi^A(t)\asafe+
%\exp(-\rho t')W_\phi^A(t')(\mu-\asafe) a(t')-\exp(-\rho t)W_\phi^A(t)(\mu-\asafe) a(t)
	u_A(t')-u_A(t)\geq
	&-(\exp(-\rho t)-\exp(-\rho t'))W_\phi^A(t)\asafe+
	(\exp(-\rho t')W_\phi^A(t')a(t')-\exp(-\rho t)W_\phi^A(t)a(t))(\mu-\asafe)
\end{align*}

\textit{Proof of (iii)}. 

If $\pcdf(t)=\pcdf(t')$, then $\ppdf(s)=0$ almost everywhere in $(t,t')$. This observation, together with the other assumptions implies $w_0^A(s)=(1-\mu)\rpdf(s)(1-\pcdf(t))$, $w_1^A(s)=\mu\rpdf(s)(1-\pcdf(t))$, $w_S^A(s)=0$, for almost all $s\in(t,t')$. Furthermore, $W^A_\phi(t)=(1-\rcdf(t))(1-\pcdf(t))$ and $W^A_\phi(t')=(1-\rcdf(t'))(1-\pcdf(t))$. The result follows from routine substitution.\end{proof}

\begin{remark} If $\ppdf(s)=0$ and $a(s)=1$ for $s\in[t,t']$, then the Right hand side of the inequality in Lemma \ref{abounds} reduces to $(1-\pcdf(t))(u_A^{FB}(t')-u_A^{FB}(t))$. Thus, parts $(i)$ and $(iii)$ of Lemma \ref{abounds} are mutually consistent.
	\end{remark}

%\begin{lemma}\label{aball} Suppose (\ref{AIC}). If $a(t)<1$ and $\pcdf(t)<1$ for some $t$, then
%\begin{enumerate}
%	\item[(i)] an $\epsilon'>0$ such that $a(t')<1$ for $t'\in[t,t+\epsilon)$.
%	\item[(ii)] an $\epsilon''>0$ exists such that for $a(t'')<1$ for $t''\in(t-\epsilon,t]$, provided $t>0$.
%\end{enumerate}
%\end{lemma}
\begin{proof}[Proof of Lemma \ref{aball}]

\textit{Part (i).} Note that $a(t)<1\Rightarrow \apdf(t)>0$ (see \ref{belief}). It follows from (\ref{AIC}) that $u_A(t')-u_A(t)\leq 0$ for any $t'>t$. From Lemma \ref{abounds}, we have 
\begin{align*}
\{\exp(-\rho t')W_\phi^A(t')a(t')-\exp(-\rho t)W_\phi^A(t)a(t)\}(\mu-\asafe)
-\{\exp(-\rho t)-\exp(-\rho t')\}W_\phi^A(t)\asafe\leq 	u_A(t')-u_A(t).
\end{align*}
Combining these inequalities, we have 
\begin{align*}
	\{\exp(-\rho t')W_\phi^A(t')a(t')-\exp(-\rho t)W_\phi^A(t)a(t)\}(\mu-\asafe)
-\{\exp(-\rho t)-\exp(-\rho t')\}W_\phi^A(t)\asafe\leq 0.
\end{align*}
Simplifying, we have
\begin{align*}
	\exp(-\rho t')W_\phi^A(t')a(t')
\leq \exp(-\rho t)W_\phi^A(t)a(t)+\frac{\{\exp(-\rho t)-\exp(-\rho t')\}W_\phi^A(t)\asafe}{\mu-\asafe}.
\end{align*}
Because $\pcdf(t)<1$, for $\epsilon\in(0,\bar{\epsilon})$ we have $\pcdf(t')<1$ for all $t'\in[t,t+\epsilon)$. In other words, $W_\phi^A(t')>0$ for such $t'$. Therefore, we have
\begin{align*}
	a(t')
\leq \underbrace{\frac{\exp(-\rho t)W_\phi^A(t)}{\exp(-\rho t')W_\phi^A(t')}}_{\text{A}}a(t)+\underbrace{\frac{\{\exp(-\rho t)-\exp(-\rho t')\}W_\phi^A(t)\asafe}{(\mu-\asafe)\exp(-\rho t')W_\phi^A(t')}}_{\text{B}}.
\end{align*}
Because $W_\phi^A(\cdot)$ is right-continuous, A is arbitrarily close to 1 provided $\epsilon$ sufficiently small. Simultaneously, $B$ is arbitrarily close to 0. Because $a(t)<1$, we have that the RHS is less than 1.

\textit{Part (ii)}. Note that $a(t)<1\Rightarrow \apdf(t)>0$ (see \ref{belief}). It follows from (\ref{AIC}) that $u_A(t)-u_A(t'')\geq 0$ for any $t''<t$. From Lemma \ref{abounds}, we have 
\begin{align*}
	(1-\pcdf(t''))(v(t)-v(t''))+\exp(-\rho t)W_\phi^A(t)(\mu-\asafe) a(t)-\exp(-\rho t'')W_\phi^A(t'')(\mu-\asafe) a(t'')&\geq  
	u_A(t)-u_A(t'').
\end{align*}
Combining these inequalities we have 
\begin{align*}
	&(1-\pcdf(t''))(v(t)-v(t''))+
\exp(-\rho t)W_\phi^A(t)(\mu-\asafe) a(t)-\exp(-\rho t'')W_\phi^A(t'')((\mu-\asafe) a(t''))\geq 0.
\end{align*}
Because $t''<t$, we have $\exp(-\rho t'')W_\phi^A(t'')\geq \exp(-\rho t)W_\phi^A(t)$. Therefore,
\begin{align*}
	&(1-\pcdf(t''))(v(t)-v(t''))+
\exp(-\rho t'')W_\phi^A(t'')\{(\mu-\asafe) a(t)\}-\exp(-\rho t'')W_\phi^A(t'')((\mu-\asafe) a(t''))\geq 0.
\end{align*}
Furthermore, $\pcdf(t'')\leq \pcdf(t)<1$. Therefore,
\begin{align*}
	&(v(t)-v(t''))+
\exp(-\rho t'')(1-\rcdf(t''))(\mu-\asafe) (a(t)-a(t''))\geq 0.
\end{align*}
Obviously, $\exp(-\rho t'')(1-\rcdf(t''))(\mu-\asafe)>0$. It follows that
\begin{align*}
&a(t'')\leq a(t)+\frac{v(t)-v(t'')}{\exp(-\rho t'')(1-\rcdf(t''))(\mu-\asafe)}.
\end{align*}
Because $v(\cdot)$ is continuous and $a(t)<1$,  for $t''$ sufficiently close to $t$, RHS is smaller than 1.\end{proof}

%
%\begin{lemma}\label{comp1} Suppose an interval $0\leq t_1<t_2\leq \tau_P$ exist such that $\acdf(t_1)=\acdf(t_2)$. If (\ref{PICa}) and (\ref{PICt}), then
%\begin{enumerate}
%\item[(i)] $\ppdf(t)=0$ for all $t\in[t_1,t_2)$.
%\item[(ii)] If $t_1>0$, then there exists $t_0<t_1$ such that $\ppdf(t)=0$ for all $t\in[t_0,t_2)$.
%\end{enumerate}
%
%\end{lemma}
\begin{proof}[Proof of Lemma \ref{comp1}] 

\textit{Part (i)}. Note that if $\acdf(t_1)=\acdf(t_2)$, then $\acdf(t_1)=\acdf(t)=\acdf(t_2)$ for all $t\in[t_1,t_2)$. In addition we have assumed (\ref{PICa}). From Lemma \ref{pbounds} part (iv), we have $u_P(t_2)-u_P(t)=(1-\sigma\acdf(t_1))(u_P^{FB}(t_2)-u_P^{FB}(t))$. We have assumed $t_2\leq \tau_P$. From Lemma \ref{fb}, $u^{FB}_P(t_2)>u^{FB}_P(t)$. Therefore, $u_P(t_2)>u_P(t)$. From (\ref{PICt}), we have $\ppdf(t)=0$.

\textit{Part (ii)}. In the proof of (i) we have shown that $u_P(t_1)<u_P(t_2)$. From continuity of $u_P(\cdot)$ (Lemma \ref{pcont}), we have $u_P(t)<u_P(t_2)$ for all $t\in[t_1-\epsilon,t_1)$, for sufficiently small $\epsilon>0$. From (\ref{PICt}), we have $\ppdf(t)=0$ for such $t$. Together with part (i) we have $\ppdf(t)=0$ for $t\in[t_1-\epsilon,t_2)$.
\end{proof}

%
%\begin{lemma}\label{comp2} Suppose an interval $[t_1,t_2)$ exists such that $\pcdf(t_1)=\pcdf(t_2)<1$. In equilibrium, $\apdf(t)=0\,$ for all $t\in[t_1,t_2)$.	
%\end{lemma}
\begin{proof}[Proof of Lemma \ref{comp2}]

Consider $T_P=\{t|\pcdf(t)=\pcdf(t_2)\}$. Because $\pcdf(t_2)<1$, $\pcdf(\cdot)$ is weakly increasing, and $\pcdf(\tau_P)=1$, the set $T_P$ is bounded from above by $\tau_P$. Therefore, it has a supremum, $t_2'$. In other words, there exists $t_2'\leq \tau_P$ and $\epsilon'>0$ such that for any $\epsilon\in(0,\epsilon')$: (1) $\pcdf(t)=\pcdf(t_2)$ for $t\in(t_2'-\epsilon,t_2')$, and (2) $\pcdf(t)>\pcdf(t_2)$ for $t\in(t_2',t_2'+\epsilon)$. Note that continuity of $u_P(\cdot)$ (Lemma \ref{pcont}) combined with (2) implies that $u_P(t_2')=u_P^*$, the \ppp{}'s equilibrium payoff, and hence $u_P(t_2')\geq u_P(t)$ for all $t\geq 0$.  

\textit{Step 1: Show that either (I) $a(t)=1$ for all almost all $t\in[t_1,t_2')$, or (II) some $t\in(t_1,t_2')$ and $\epsilon>0$ exist such that $a(t'')<1$ for almost all $t''\in(t-\epsilon,t)$ and $a(t')=1$ for all $t'\in(t,t+\epsilon)$.}

Consider $T_A=\{t|t\in(t_1,t_2')\cap a(t)<1\}$ (note the exclusion of $t_1$). If $T_A$ has measure zero, then (I).

 Suppose $T_A$ has strictly positive measure. Consider $t\in T_A$. Using Lemma \ref{aball}, we conclude that an interval $(t-\epsilon,t+\epsilon)\subset T_A$ for some $\epsilon>0$. That is, the set $T$ is open. It follows that $T_A$ is a union of disjoint open intervals, i.e.  $T_A=\cup_{k=1}^K (q_k,p_k)$, where $K\leq \infty$ and $t_1\leq q_1<p_1\leq q_2<p_2...\leq q_k<p_k...\leq t_2$.  That is, $q_k$ is the bottom of interval $k$ and $p_k$ is its top. Note that $p_k=q_{k+1}$ is allowed; in this case, the top of open interval $k$ coincides with the bottom of open interval $k+1$. For ease of notation, if $K=\infty$, let $p_K\equiv \lim_{k\rightarrow \infty}p_k$, which exists because $\{p_k\}$ is a bounded monotone sequence. In addition, whether $K$ is finite or infinite, $p_K\leq t_2'$ by construction.
 
Note that if a $k\in K$ exists such that $p_k<q_{k+1}$, then (II). Indeed,   consider time $p_k$ and $0<\epsilon<\min\{p_k-q_k, q_{k+1}-p_k\}$. If  $t''\in(p_k-\epsilon,p_k)$, then $q_k<t''<p_k$, and hence $a(t'')<1$. Simultaneously, if $t'\in(p_k,p_k+\epsilon)$, then $p_k<t'<q_{k+1}$, and hence $a(t)=1$.

To proceed, consider the remaining possibility that $p_k=q_{k+1}$ for all $k\in K$. In this case, the upper bound of interval $k$ always coincides with the lower bound of interval $k+1$. Thus, the set $T$ is an interval $(q_1,p_K)$, with a countable number of points removed, i.e.  $T=\{(q_1,p_K)-\{p_k\}_{k=1}^\infty\}$. Therefore, for almost all $t\in(q_1,p_K)$, we have $a(t)<1$. In turn (\ref{PICa}) implies that $\bl(t)=\psafe$ for almost all such $t$. Thus,,
\begin{align*}
&	u_p(p_K)-u_P(q_1)=\\
&	\int_{q_1}^{p_K}\exp(-\rho s)\{(w_0^P(s)+w_1^P(s)+w^P_S(s))\}\psafe ds+\exp(-\rho p_K)W_\phi^P(p_K)\psafe-\exp(-\rho q_1)W_\phi^P(q_1)\psafe.
\end{align*}
Using Observation 1, and noting that $\exp(-\rho s)<\exp(-\rho q_1)$ for $s\in (q_1,p_K]$, we have 
\begin{align*}
	u_P(p_K)-u_P(q_1)\leq& 
	\exp(-\rho q_1)(W_\phi^P(q_1)-W_\phi^P(p_K))\psafe+\exp(-\rho p_K)W_\phi^P(p_K)\psafe-\exp(-\rho q_1)W_\phi^P(q_1)\psafe=\\
 -(&\exp(-\rho q_1)-\exp(-\rho p_K))W_\phi^P(p_K))\psafe=\\
  -(&\exp(-\rho q_1)-\exp(-\rho p_K))(1-\rcdf(p_K))(1-\pcdf(t_2)))\psafe<0.
\end{align*}
Therefore, $u_P(p_K)<u_P(q_1)\leq u_P(t_2')$. Because $u_P(\cdot)$ is continuous (see Lemma \ref{pcont}), it must be the case that $p_K<t_2'$.

Finally, consider time $p_K$ and $0<\epsilon<\min\{p_K-q_1,p_K-t_2'\}$. If $t''\in(p_K-\epsilon,p_K)$, then $t''\in(q_1,p_K)$. Recall that $a(t'')<1$ for almost all such $t''$. Furthermore,  if $t'\in(p_K,p_K+\epsilon)$, then $t'\in(p_K,t_2')$, and hence, $a(t')=1$. This completes the proof of Step 1.

\textit{Step 2: Show that $a(t)=1$ for all almost all $t\in[t_1,t_2')$}.

Step 2 is proved by establishing that Case (II) in the statement of Step 1 leads to a contradiction. Therefore, suppose that for some $t\in(t_1,t_2')$ and $\epsilon>0,$ exist such that $a(t'')<1$ for almost all $t''\in(t-\epsilon,t)$, and $a(t')=1$ for all $t'\in(t,t+\epsilon)$. 

Consider $t-\epsilon<t_L<t<t_M<t_H<t+\epsilon\leq t_2'$. By construction, $a(t)=1$ for all $t \in[t_M,t_H]$. In addition, $\pcdf(t_M)=\pcdf(t_H)=\pcdf(t_2)<1$. From Lemma \ref{abounds} part (iii), we have 
\begin{align*}
	\frac{u_A(t_H)-u_A(t_M)}{1-\pcdf(t_2)}=u_A^{FB}(t_H)-u_A^{FB}(t_M).
\end{align*}
%Note that $u_A^{FB}(\cdot)$ is continuous. Therefore, for $t_M$ arbitrarily close to $t$, we have $u_A^{FB}(t_H)-u_A^{FB}(t_M)$ arbitrarily close to $(1-\pcdf(t_2'))(u_A^{FB}(t_H)-u_A^{FB}(t))$.

Using Lemma \ref{abounds} part (ii), we have
\begin{align*}
&u_A(t_M)-u_A(t_L)\geq\\
	&	\{\exp(-\rho t_M)W_\phi^A(t_M)-\exp(-\rho t_L)W_\phi^A(t_L)a(t_L)\}(\mu-\asafe)
-\{\exp(-\rho t_L)-\exp(-\rho t_M)\}W_\phi^A(t_L)\asafe.
\end{align*}
Note that $a(t_L)\leq 1$ and $W_\phi^A(t_L)\leq 1$. It follows that 
\begin{align*}
&u_A(t_M)-u_A(t_L)\geq\\
	&	\{\exp(-\rho t_M)W_\phi^A(t_M)-\exp(-\rho t_L)W_\phi^A(t_L)\}(\mu-\asafe)
-\{\exp(-\rho t_L)-\exp(-\rho t_M)\}\asafe.
\end{align*}
Next, note that $\pcdf(t_M)=\pcdf(t_L)=\pcdf(t_2)<1$. Substituting and simplifying, we have 
\begin{align*}
&\frac{u_A(t_M)-u_A(t_L)}{1-\pcdf(t_2)}\geq\\ &-
	\{\exp(-\rho t_L)(1-\rcdf(t_L))-\exp(-\rho t_M)(1-\rcdf(t_M))\}(\mu-\asafe)
-\{\exp(-\rho t_L)-\exp(-\rho t_M)\}\asafe\\
&=-(\gamma(t_L)-\gamma(t_M)),
\end{align*}
where $\gamma(t)\equiv \exp(-\rho t)\{(1-\rcdf(t))(\mu-\asafe)+\asafe\}$. 

Next, note that 
\begin{align*}
\frac{u_A^{FB}(t_H)-u_A^{FB}(t_L)}{1-\pcdf(t_2)}=\frac{u_A^{FB}(t_H)-u_A^{FB}(t_M)}{1-\pcdf(t_2)}+\frac{u_A^{FB}(t_M)-u_A^{FB}(t_L)}{1-\pcdf(t_2)}.
\end{align*}
Using the previous bounds, it follows that 
%\begin{align*}
%&\frac{u_A(t_H)-u_A(t_L)}{1-\pcdf(t_2)}\geq u_A^{FB}(t_H)-u_A^{FB}(t_M)\\
%&-\{\exp(-\rho t_L)(1-\rcdf(t_L))-\exp(-\rho t_M)(1-\rcdf(t_M))\}(\mu-\asafe)
%-\{\exp(-\rho t_L)-\exp(-\rho t_M)\}\asafe
%\end{align*}
\begin{align*}
&\frac{u_A(t_H)-u_A(t_L)}{1-\pcdf(t_2)}\geq u_A^{FB}(t_H)-u_A^{FB}(t_M)-(\gamma(t_L)-\gamma(t_M)).	
\end{align*}
Choose $t_M$ just above $t$. By continuity of $\gamma(\cdot)$, for $t_L$ sufficiently close to $t$, the difference $\gamma(t_L)-\gamma(t_M)\approx 0$. At the same time Lemma \ref{fb} guarantees $u_A^{FB}(t_H)-u_A^{FB}(t_M)>0$. Therefore for some $\epsilon''>0$, we have that $t''\in(t-\epsilon'',t)$ implies $u_A(t_H)-u_A(t'')>0$. From (\ref{AIC}), we have $\apdf(t'')=0$. Hence, $\bl(t'')=1$, which in turn implies $a(t'')=1$. This contradicts the initial assumption that $a(t'')<1$ for almost all $t''\in(t-\epsilon,t)$.

\textit{Step 3: Show that $\apdf(t)=0$ for all $t\in[t_1,t_2')$.}

From Step 2, $a(t)=1$ for almost all $t\in[t_1,t_2')$. Consider $t_1\leq t<t_2'$. For any such $t$, it is possible to find a $t'\in(t,t_2')$ for which $a(t')=1$. Note that
\begin{align*}
&\frac{	u_A(t')-u_A(t)}{1-\pcdf(t_2)}=\\
&\int_{t}^{t'}\exp(-\rho s)\rpdf(s)(\mu+(1-\mu)\asafe)ds+\exp(-\rho t')(1-\rcdf(t'))\mu-\exp(-\rho t)(1-\rcdf(t))(\asafe+a(t)(\mu-\beta))\geq \\
&\int_{t}^{t'}\exp(-\rho s)\rpdf(s)(\mu+(1-\mu)\asafe)ds+\exp(-\rho t')(1-\rcdf(t'))\mu-\exp(-\rho t)(1-\rcdf(t))\mu=\\
&u_A^{FB}(t')-u_A^{FB}(t)>0,
\end{align*}
where the last equality follows from $t'<t_2'\leq\tau_P<\tau_A$. It follows that for any $t\in[t_1,t_2')$, there exists some $t'$ such that $u_A(t')>u_A(t)$. From (\ref{AIC}), we have $\apdf(t)=0$. Obviously $t_2'\geq t_2$, and $\apdf(t)=0$ holds for $t\in[t_1,t_2)$.
\end{proof}

%\begin{lemma}\label{gapexp} Suppose $0<t_1<t_2\leq \tau_P$ and $\pcdf(t_1)=\pcdf(t_2)<1$. In equilibrium, $\pcdf(t)=\pcdf(t_2)$ for all $t\in[0,t_2]$.
%\end{lemma}
\begin{proof}[Proof of Lemma \ref{gapexp}] 

Consider $T_P=\{t|\pcdf(t)=\pcdf(t_2)\}$. Because $T_P$ is bounded from below by 0 and $\pcdf(\cdot)$ is right-continuous, $T_P$ has a minimum, $\underline{t}_1\in[0,t_1]$. To prove the result, it will be shown that $\underline{t}_1=0$. 

Suppose that $\underline{t}_1>0$. Consider interval $[\underline{t}_1,t_2')$ where $t_1<t_2'<t_2$. Applying Lemma \ref{comp2}, we have $\apdf(t)=0$ for all $t\in[\underline{t}_1,t_2')$. Next, consider interval $[\underline{t}_1,t_2'']$ where $t_1<t_2''<t_2'$. Because $\apdf(t)=0$ for $t\in[t_1,t_2'']$, we have $\acdf(t_1)=\acdf(t_2'')$. Furthermore, $t_2''<t_2<\tau_P$, where the last  inequality follows from $\pcdf(t_2)<1$ and Lemma \ref{ptop}. Using Lemma \ref{comp1} part (ii), there exists $t_0<\underline{t}_1$ such that $\ppdf(t)=0$ for all $t\in[t_0,t_2'')$.  Consider interval $[t_0,t_2''']$ where $t_1<t_2'''<t_2''$. Because $\ppdf(t)=0$ for all such $t$, we have $\pcdf(t_0)=\pcdf(t_2''')$. Furthermore, $t_2'''\in[t_1,t_2]$, and therefore $\pcdf(t_2''')=\pcdf(t_2)$. Hence, $\pcdf(t_0)=\pcdf(t_2)$. Because $t_0<\underline{t}_1$, we have that $\underline{t}_1$ is not the minimum of $T_P$.
\end{proof}

%\begin{lemma}\label{gapexa} Suppose $0<t_1<t_2$, $\pcdf(t_2)<1$, and $\acdf(t_1)=\acdf(t_2)$. In equilibrium, $\acdf(t)=\acdf(t_2)$ for all $t\in[0,t_2]$.
%\end{lemma}

\begin{proof}[Proof of Lemma \ref{gapexa}]

Consider $T_A=\{t|\acdf(t)=\acdf(t_2)\}$. Because $T_A$ is bounded from below by 0 and $\acdf(\cdot)$ is right-continuous, $T_A$ has a minimum, $\underline{t}_1\in[0,t_1]$. To prove the result, it will be shown that $\underline{t}_1=0$. 

Suppose that $\underline{t}_1>0$. Consider interval $[\underline{t}_1,t_2')$ where $t_1<t_2'<t_2$. Applying Lemma \ref{comp1} part (ii), we have $\ppdf(t)=0$ for all $t\in[t_0,t_2')$, for some $t_0<\underline{t}_1$. By implication, $\ppdf(t)=0$ for all $t\in[t_0,t_2'']$, where $t_1<t_2''<t_2'$. Hence, $\pcdf(t_0)=\pcdf(t_2'')$. Furthermore, because $t_2''<t_2$, we have $\pcdf(t_2'')\leq \pcdf(t_2)<1$. Applying Lemma \ref{comp2}, we have $\apdf(t)=0$ for all $t\in[t_0,t_2'')$. By implication, $\apdf(t)=0$ for all $t\in[t_0,t_2''']$, where $t_1<t_2'''<t_2''$. Therefore, $\acdf(t_0)=\acdf(t_2''')=\acdf(t_2)$, where the last equality follows from $t_2'''\in(t_1,t_2)$. Because $t_0<\underline{t}_1$ and $\acdf(t_0)=\acdf(t_2)$, we have that $\underline{t}_1$ is not the minimum of $T_A$.	
\end{proof}
\subsection{Proof of Corollary 2}
\begin{proof}[Proof of Corollary \ref{corcs}]

\textit{Part (i)}. Obviously, $\sigma$ has no effect on $\tau_P$. To see that it reduces $\mint$, recall the following implicit function relating $\mint$ and $\sigma$ derived in Proposition \ref{charval}.
\begin{align*}
\exp\{-\int_{\mint}^{\tau_P}\frac{\mu(1-\psafe)\rhaz(s)-\rho\psafe}{\psafe-\mu}ds\}=1-\sigma.
\end{align*}
Differentiating with respect to $\sigma$,
\begin{align*}
\frac{d\mint}{d\sigma}\{\frac{\mu(1-\psafe)\rhaz(\mint)-\rho\psafe}{\psafe-\mu}\}\exp\{-\int_{\mint}^{\tau_P}\frac{\mu(1-\psafe)\rhaz(s)-\rho\psafe}{\psafe-\mu}\}=-1.
\end{align*}
Simplifying, 
\begin{align*}
\frac{d\mint}{d\sigma}\{\frac{\mu(1-\psafe)\rhaz(\mint)-\rho\psafe}{\psafe-\mu}\}(1-\sigma)=-1\Rightarrow\\
\frac{d\mint}{d\sigma}=\{\frac{\mu(1-\psafe)\rhaz(\mint)-\rho\psafe}{\psafe-\mu}\}(1-\sigma)=-\frac{\psafe-\mu}{1-\sigma}\frac{1}{\mu(1-\psafe)(\rhaz(\mint)-\frac{\rho\psafe}{\mu(1-\psafe)})}<0,
\end{align*}
where the last inequality follows because $\rhaz(s)>\rho\psafe/(\mu(1-\psafe))$ for $s<\tau_P$, and $\mint<\tau_P$. Next, note that
%%\begin{align*}
%%		\acdf(t)&=\frac{1}{\sigma}\Big(1-\exp\{-\int_{\mint}^{t}\frac{\mu(1-\psafe)\rhaz(s)-\rho\psafe}{\psafe-\mu}ds\}\Big)
%\end{align*}
\begin{align*}
\frac{\partial\acdf(t)}{\partial\mint}\frac{d\mint}{d\sigma}&=-\frac{1}{\sigma}(1-\sigma\acdf(t))\frac{\mu(1-\psafe)\rhaz(\mint)-\rho\psafe}{\psafe-\mu}(-\frac{\psafe-\mu}{1-\sigma}\frac{1}{\mu(1-\psafe)(\rhaz(\mint)-\frac{\rho\psafe}{\mu(1-\psafe)})})\\
	&=\frac{1-\sigma\acdf(t)}{\sigma(1-\sigma)}\geq \frac{1}{\sigma}.
\end{align*}
Meanwhile, we have 
\begin{align*}
	\frac{\partial\acdf(t)}{\partial\sigma}&=-\frac{1}{\sigma^2}\acdf(t)\geq -\frac{1}{\sigma^2}.
\end{align*}
Hence, $d\acdf(t)/d\sigma\geq 1/\sigma-1/\sigma^2>0$ for $\sigma\in(0,1)$. Therefore, an increase in $\sigma$ increases $\acdf(\cdot)$, resulting in a FOSD shift toward smaller $t$. 

Note that $\pcdf(t)$ depends on $\sigma$ only via $\mint$ and \begin{align*}
\frac{d\pcdf(t)}{d\mint}=-(1-\pcdf(t))\frac{\asafe(1-\mu)\rhaz(\mint)-\rho\mu}{\mu-\asafe}<0, 	
 \end{align*}
 where the last inequality follows because $\rhaz(\cdot)$ is decreasing, and $\mint<\tau_A$. Because an increase in $\sigma$ reduces $\mint$, it increases $\pcdf(\cdot)$, resulting in a FOSD shift toward smaller $t$.
 
%  Each player's equilibrium payoff is $u_i^{FB}(\mint)$. Because $\mint<\tau_P<\tau_A$, each player's first best payoff is increasing in $\mint$. Increases in $\sigma$ reduce $\mint$, but do not otherwise affect the first best payoffs. Such increases therefore reduce the equilibrium payoffs of both players.
 
 \textit{Part (ii).} First note that
 \begin{align*}
 	\rhaz(\tau_P)=\frac{\rho\psafe}{\mu(1-\psafe)}\Rightarrow 
% 	\tau_P'(\psafe)=\frac{\rho\mu(1-\psafe)+\rho\mu\psafe}{\rhaz'(\tau_P)\mu^2(1-\psafe)^2}
 \tau_P'(\psafe)=\frac{\rho}{\rhaz'(\tau_P)\mu(1-\psafe)^2}<0,
 \end{align*}
where the inequality follows from $\rhaz(\cdot)<0$.
 
 Next, recall that Proposition  \ref{charval} gives an implicit function relating $\mint$, $\tau_P$, and $\sigma$,
\begin{align*}
\exp\{-\int_{\mint}^{\tau_P}\frac{\mu(1-\psafe)\rhaz(s)-\rho\psafe}{\psafe-\mu}ds\}=1-\sigma.
\end{align*}
Let 
\begin{align*}
\ahaz(s)&\equiv \frac{\mu(1-\psafe)\rhaz(s)-\rho\psafe}{\psafe-\mu}\\
h_A(s)&\equiv \frac{d\ahaz(s)}{d\psafe}=\frac{-(\mu\rhaz(s)+\rho)(\psafe-\mu)-\mu(1-\psafe)\rhaz(s)+\rho\psafe}{(\psafe-\mu)^2}
	=-\frac{\mu(1-\mu)}{(\psafe-\mu)^2}(\rhaz(s)-\frac{\rho}{1-\mu}).
\end{align*}
Differentiating with respect to $\theta$, we have 
\begin{align*}
\exp\{-\int_{\mint}^{\tau_P}\ahaz(s)ds\}\Big(\ahaz(\mint)\tau_M'(\psafe)-\ahaz(\tau_P)\tau_P'(\psafe)-\int_{\tau_M(\psafe)}^{\tau_P(\psafe)}h_A(s)ds\Big)=0.
\end{align*}
Note that $\rhaz(\tau_P)=\rho\psafe/(\mu(1-\psafe))$, which implies $\ahaz(\tau_P)=0$. By implication,
\begin{align*}
\ahaz(\mint)\tau_M'(\psafe)-\int_{\tau_M}^{\tau_P}h_A(s)ds=0\Rightarrow \tau_M'(\psafe)=\frac{\int_{\tau_M}^{\tau_P}h_A(s)ds}{\ahaz(\mint)}.
\end{align*}
Next, note that for $s\in[\mint,\tau_P]$, 
\begin{align*}
	\rhaz(s)-\frac{\rho}{1-\mu}\geq 	\frac{\rho\psafe}{\mu(1-\psafe)}-\frac{\rho}{1-\mu}=\frac{\rho(\psafe-\mu)}{\mu(1-\mu)(1-\psafe)}>0.
\end{align*}
It follows that $h_A(s)<0$ for such $s$. Simultaneously, $\rhaz(s)>\rho\psafe/(\mu(1-\psafe))$ and hence, $\ahaz(s)>0$. Combining these observations,  $\mint'(\psafe)<0$.

Next, note that 
\begin{align*}
	\acdf(t)=\frac{1}{\sigma}(1-\exp(-\int_{\mint}^t\ahaz(s)ds)).
\end{align*}
It follows that 
\begin{align*}
	\frac{d\acdf(t)}{d\psafe}=\frac{1}{\sigma}\exp(-\int_{\mint}^t\ahaz(s)ds)(\int_{\mint}^t h_A(s)ds-\ahaz(\mint)\mint'(\psafe)).
\end{align*}
Substituting for $\mint'(\psafe)$, we have 
\begin{align*}
	\frac{d\acdf(t)}{d\psafe}&=\frac{1}{\sigma}\exp(-\int_{\mint}^t\ahaz(s)ds)(\int_{\mint}^t h_A(s)ds-\int_{\mint}^{\tau_P} h_A(s)ds)\\
	&=\frac{1}{\sigma}\exp(-\int_{\mint}^t\ahaz(s)ds)(-\int_{t}^{\tau_P} h_A(s)ds).
\end{align*}
Recalling that $h_A(\cdot)<0$ for such $s$, we have $d\acdf(t)/d\psafe>0$, and hence an increase in $\psafe$ results in a FOSD shift toward smaller $t$. 

Finally, note that $\pcdf(\cdot)$ depends on $\psafe$ only through $\mint$, and for $t\in[\mint,\tau_P)$
\begin{align*}
	\frac{d\pcdf(t)}{d\mint}=(1-\pcdf(t))\frac{\asafe(1-\mu)\rhaz(\mint)-\rho\mu}{\asafe-\mu}>0,
\end{align*}
where the last inequality follows because $\mint<\tau_A$. Thus, for all $t\in[\mint,\tau_P)$ the distribution shifts up, and the upper bound $\tau_P$ where the distribution has an atom, also shifts left, resulting in a FOSD shift toward smaller $t$. \end{proof}
\newpage

\end{document}